\documentclass[a4paper,UKenglish, cleveref, autoref, thm-restate]{lipics-v2021}
\usepackage{stmaryrd}
\let\stmaryrdLightning\lightning

\usepackage{subcaption} \usepackage{listings}
\usepackage{multicol}
\usepackage[font=small]{caption}
\usepackage{tikz}
\usepackage{amsmath, amsfonts, stmaryrd, xcolor, mathpartir, verbatim, xspace}
\usepackage{wasysym, float}
\usepackage{etoolbox}
\usepackage{graphicx}
\usepackage{mathtools}
\nolinenumbers

\usepackage[T1]{fontenc}
\usepackage{soul}
\usetikzlibrary{shadows,positioning}
\usepackage{changepage}
\usepackage{placeins}
\usepackage{paralist}
\usepackage{hyperref}
\usepackage{thmtools}
\usepackage{thm-restate}
\usepackage[numbers, sort&compress]{natbib}
\usepackage{wrapfig}
\usepackage[scaled=0.85]{beramono}

\floatstyle{plain}
\restylefloat{figure}
\DeclareCaptionFormat{myformat}{\hrulefill\\#1#2#3}
\DeclareMathAlphabet{\mathttbf}{\encodingdefault}{\ttdefault}{bx}{n}
  \colorlet{colBox}{blue!20}
  \colorlet{colBorder}{gray!70}
  \tikzset
    {mybox/.style=
      {rectangle,rounded corners,drop shadow,minimum height=1cm,
       minimum width=2cm,align=center,fill=colBox,draw=colBorder,line width=1pt
      },
     myarrow/.style=
      {line width=1pt,-stealth,rounded corners
      },
     mylabel/.style={text=#1}
    }
\lstdefinelanguage{ensemble}{
	basicstyle=\ttfamily\scriptsize,
	numbersep=5pt,
	numbers={left},
	xleftmargin=3.0ex,
	captionpos=b,
	morecomment=[l]{//}, morecomment=[s]{/*}{*/}, morestring=[b]", 	tabsize=1,			commentstyle=\color{black}, stringstyle=\color{black}, emph={
		actor, presents, stage, type, is, struct,
		boot, behaviour, constructor, if, then, else, try, except,
		follows
	},
	emphstyle={\color{red}\bfseries},
	emph={[2]=,
		actor, presents, stage, type, is, struct, for, while, do, send, on, receive, from, out, in, select, default, continue, establish, topology,
		boot, new,  connect, disconnect, to, and, or, where, of, publish, unpublish,
    proc, project, as, stop, return, spawn, at, migrate, replace, with, switch,
    case, break, link, unlink, global, local, role, protocol, self, choice, rec,
    explicit, aux, discover
	},
	emphstyle={[2]\color{blue}\textbf},
emph={[3]=,
		mov, true, false
	},
	emphstyle={[3]\color{pink}\textbf},
emph={[4]=,
		integer, query, string, int, real, bool, interface, struct, enum, any
	},
	emphstyle={[4]\color{purple}\textbf}
}

\usepackage{letltxmacro}
\newcommand*{\SavedLstInline}{}
\LetLtxMacro\SavedLstInline\lstinline
\DeclareRobustCommand*{\lstinline}{\ifmmode
    \let\SavedBGroup\bgroup
    \def\bgroup{\let\bgroup\SavedBGroup
      \hbox\bgroup
    }\fi
  \SavedLstInline
}

\newcommand{\mysubsection}[1]{\medskip\noindent
\textbf{#1}}

\newcommand{\calcwd}[1]{{\textbf{\textsf{#1}}}}
\newcommand{\mkwd}[1]{{\textsf{#1}}}

\newcommand{\seq}[1]{\overrightarrow{#1}}
\newcommand{\bstop}{{\calcwd{stop}}\xspace}
\newcommand{\new}[1]{\calcwd{new} \; #1}
\newcommand{\replace}[2]{\calcwd{replace} \; #1 \; \calcwd{with} \; #2}
\newcommand{\newdiscover}[1]{\calcwd{discover} \; #1}

\newcommand{\config}[1]{\mathcal{#1}}
\newcommand{\actorsep}{,}
\newcommand{\eactor}[4]{\langle #1 \actorsep #2 \actorsep #3 \actorsep #4\rangle}

\newcommand{\teval}{\longrightarrow_{\mathsf{M}}}

\newcommand{\ceval}{\longrightarrow}

\newcommand{\actordef}[3]{\calcwd{actor} \; #1 \; \calcwd{follows} \; #2 \; \{ #3 \} }

\newcommand{\bcirc}{\bullet\xspace}
\newcommand{\wcirc}{\circ\xspace}

\newcommand{\totheleft}[1]{\begin{flushleft}#1\end{flushleft}}
\newcommand{\role}[1]{{\color{purple}{\ensuremath{\mathttbf{#1}}}}}

\newcommand{\continue}[1]{\calcwd{continue} \: #1}
\newcommand{\annst}[2]{#1 \mathop{::} #2}
\newcommand{\annstone}[1]{#1 \mathop{::}}

\newenvironment{proofcase}[1]
  {\totheleft{\textbf{Case \textsc{#1}}}}
  {}

\newcommand{\dom}[1]{\mkwd{dom}(#1)}

\newcommand{\efflet}[3]{\calcwd{let} \: #1 \Leftarrow #2 \: \calcwd{in} \: #3}
\newcommand{\efflettwo}[2]{\calcwd{let} \: #1 \Leftarrow #2 \: \calcwd{in}}
\newcommand{\effreturn}[1]{\calcwd{return} \: #1}
\newcommand{\one}{\textbf{1}}

\newcommand{\deriv}[1]{\mathbf{#1}}

\newcommand{\jarg}[1]{\{ #1 \} \:\:}
\AtBeginDocument{

}

\newcommand{\EnsembleS}{EnsembleS\xspace}

\newcommand{\recty}[2]{\mu #1 . #2}

\newcommand{\globalend}{\mkwd{end}\xspace}

\newcommand{\localend}{\mkwd{end}\xspace}
\newcommand{\midspace}{\, \mid \,}

\newcommand{\roles}[1]{\mkwd{roles}(#1)}

\newcommand{\defeq}{\triangleq}

\newcommand{\project}[2]{#1 \upharpoonright #2}

\newcommand{\self}{\calcwd{self}\xspace}

\newcommand{\header}[1]{\totheleft{\textbf{#1}}}
\newcommand{\subheader}[1]{\totheleft{\textit{#1}}}
\newcommand{\subheadersig}[2]{\totheleft{\textit{#1} \hfill \framebox{#2}}}
\newcommand{\headersig}[2]{\totheleft{\textbf{#1} \hfill \framebox{#2}}}
\newcommand{\headersigarg}[2]{\totheleft{\textbf{#1} \hfill #2}}
\newcommand{\rowspace}{1em}
\newcommand{\rowskip}{\vspace{\rowspace}}
\newcommand{\size}[1]{|#1|}
\newcommand{\subj}[1]{\mkwd{subj}(#1)}
\newcommand{\obj}[1]{\mkwd{obj}(#1)}

\newcommand{\isoutput}[1]{\mkwd{isOutput}(#1)}
\newcommand{\isinput}[1]{\mkwd{isInput}(#1)}

\newcommand{\annarrow}[1]{\xrightarrow{#1}}
\newcommand{\syncannarrow}[1]{\xRightarrow{#1}}

\newcommand{\bl}{\begin{array}{l}}
\newcommand{\el}{\end{array}}

\newcommand{\vseq}[3]{{#1} \vdash {#2} {:} {#3}}
\newcommand{\dseq}[1]{\vdash {#1}}

\newcommand{\cseq}[3]{#1; #2 \vdash #3}
\newcommand{\zseq}[3]{#1; #2 \vdashzap #3}

\newcommand{\prop}{\ensuremath{\varphi}}
\newcommand{\proparg}[1]{\ensuremath{\prop(#1)}}
\newcommand{\rolesetp}{\tilde{\prole}}
\newcommand{\rolesetq}{\tilde{\qrole}}
\newcommand{\rolesetr}{\tilde{\rrole}}
\newcommand{\rolesets}{\tilde{\srole}}
\newcommand{\sessindextwo}[2]{#1[#2]}
\newcommand{\sessindexrolesnoty}[3]{#1[#2]\langle #3 \rangle}
\newcommand{\sessindexroles}[4]{\sessindexrolesnoty{#1}{#2}{#3}{:}#4}
\newcommand{\commdir}{\dagger}
\newcommand{\disconndir}{\ddagger}
\newcommand{\localcomm}[4][\commdir]{#2{#1}{#3(#4)}}
\newcommand{\msg}[2]{#1(#2)}
\newcommand{\behaviour}{\kappa}

\newcommand{\tya}{A}
\newcommand{\tyb}{B}
\newcommand{\safe}[1]{\mkwd{safe}(#1)}
\newcommand{\waitact}{{\#}{\uparrow}}
\newcommand{\disconnact}{{\#}{\downarrow}}

\newcommand{\ltslbl}{\gamma}

\newcommand{\ltscomm}[6][\commdir]{#2 {:} #3 {#1} #4 {::} #5(#6)}

\newcommand{\ltspair}[4]{#1 {:} #2, #3 {::} #4}
\newcommand{\ltsconnpair}[4]{#1 {:} #2 \twoheadrightarrow #3 {::} #4}
\newcommand{\ltsdisconnpair}[3]{#1 {:} #2 {\#} #3}
\newcommand{\ltsdisconnact}[4][\disconndir]{#2 {:} #3 {#1} #4}
\newcommand{\ltswait}[3]{\ltsdisconnact[\waitact]{#1}{#2}{#3}}
\newcommand{\ltsdisconn}[3]{\ltsdisconnact[\disconnact]{#1}{#2}{#3}}

\newcommand{\tyactive}[1]{\ensuremath{\mkwd{active}(#1)}}
\newcommand{\tyinactive}[1]{\ensuremath{\mkwd{inactive}(#1)}}

\newcommand{\rtenv}{\Delta}

\newcommand{\hyproject}[3]{#1 \upharpoonright_{#3} \role{#2}}
\newcommand{\recenv}{\Phi}
\newcommand{\globact}{\pi}
\newcommand{\locacta}{\alpha}
\newcommand{\locactb}{\beta}
\newcommand{\locact}{\locacta}
\newcommand{\loctya}{S}
\newcommand{\outloctya}{S^{\bcirc}}
\newcommand{\inloctya}{S^{\wcirc}}
\newcommand{\loctyb}{T}
\newcommand{\loctyc}{U}
\newcommand{\locty}{\loctya}
\newcommand{\then}{\,{.}\,}
\newcommand{\sumact}[2]{\Sigma_{#1} (#2)}
\newcommand{\outsumact}[2]{\Sigma^{\bcirc}_{#1} (#2)}
\newcommand{\insumact}[2]{\Sigma^{\wcirc}_{#1} (#2)}
\newcommand{\globalsend}[4]{#1 \rightarrow #2 : #3(#4)}
\newcommand{\globalconn}[4]{\role{#1} \twoheadrightarrow \role{#2} : #3(#4)}
\newcommand{\globaldisconn}[2]{\role{#1} \# \role{#2}}

\newcommand{\newconnsess}[3]{\sessindexrolesnoty{#1}{#2}{#3}}

\newcommand{\newsend}[3]{\calcwd{send} \; #1(#2) \; \calcwd{to} \; {#3}}
\newcommand{\newconn}[4]{\calcwd{connect} \; #1(#2) \; \calcwd{to} \; {#3}
\; \calcwd{as} \; #4}

\newcommand{\newrecv}[2]{\calcwd{receive} \; \calcwd{from} \; #1 \; \{ #2 \} }
\newcommand{\newrecvone}[1]{\calcwd{receive} \; \calcwd{from} \; #1 \; }
\newcommand{\newaccept}[2]{\calcwd{accept} \; \calcwd{from} \; #1 \; \{ #2 \} }
\newcommand{\simplerecv}[3]{\calcwd{receive} \; #1(#2) \; \calcwd{from} \; #3}
\newcommand{\simpleaccept}[3]{\calcwd{accept} \; #1(#2) \; \calcwd{from} \; #3}
\newcommand{\newwait}[1]{\calcwd{wait} \; #1}
\newcommand{\newdisconn}[1]{\calcwd{disconnect} \; \calcwd{from} \; #1 }

\newcommand{\localtycomm}[4]{{#1}{#2}{#3}({#4})}
\newcommand{\localtysend}[3]{{#1}{!}{#2}({#3})}
\newcommand{\localtyrecv}[3]{{#1}{?}{#2}({#3})}
\newcommand{\localtyconn}[3]{{#1}{!!}{#2}({#3})}
\newcommand{\localtyaccept}[3]{{#1}{??}{#2}({#3})}
\newcommand{\localtywait}[1]{\waitact{#1}}
\newcommand{\localtydisconn}[1]{\disconnact{#1}}

\newcommand{\prole}{\role{p}\xspace}
\newcommand{\qrole}{\role{q}\xspace}
\newcommand{\rrole}{\role{r}\xspace}
\newcommand{\srole}{\role{s}\xspace}
\newcommand{\trole}{\role{t}\xspace}
\newcommand{\actorcls}{u}

\newcommand{\pidty}[1]{\mkwd{Pid}(#1)}
\newcommand{\proto}{P}
\newcommand{\ty}[1]{\mkwd{ty}(#1)}

\newcommand{\tseq}[6][\loctyb]{ \jarg{#1} #2 \midspace {#3} \triangleright {#4}
{:} {#5} \triangleleft {#6} }
\newcommand{\bseq}[3][\loctyb]{ \jarg{#1} {#2} \vdash {#3}}

\newcommand{\connstate}{\sigma}
\newcommand{\finished}[1]{\mkwd{end}(#1)}

\newcommand{\sgroup}{\mathcal{S}}

\makeatletter
 \renewenvironment{abstract}{\vskip0.7\bigskipamount
   \noindent
   \rlap{\color{lipicsLineGray}\vrule\@width\textwidth\@height1\p@}\hspace*{7mm}\fboxsep1.5mm\colorbox[rgb]{1,1,1}{\raisebox{-0.4ex}{\large\selectfont\sffamily\bfseries\abstractname}}\vskip3\p@
   \fontsize{9}{12}\selectfont
   \noindent\ignorespaces}
   {
    \ifx\@hideLIPIcs\@undefined
     \vskip0.7\baselineskip\noindent
    \subjclassHeading
    \ifx\@ccsdescString\@empty
      \textcolor{red}{Author: Please fill in 1 or more \string\ccsdesc\space macro}\else
      \@ccsdescString
    \fi
    \vskip0.7\baselineskip
    \noindent\keywordsHeading
    \ifx\@keywords\@empty
      \textcolor{red}{Author: Please fill in \string\keywords\space macro}\else
      \@keywords
    \fi
      \ifx\@DOIPrefix\@empty\else
        \vskip0.7\baselineskip\noindent
        \doiHeading\href{https://doi.org/\@DOIPrefix.\@EventAcronym.\@EventYear.\@ArticleNo}{\@DOIPrefix.\@EventAcronym.\@EventYear.\@ArticleNo}\fi
    \ifx\@category\@empty\else
      \vskip0.7\baselineskip\noindent
      \categoryHeading\@category
    \fi
    \ifx\@relatedversion\@empty\else
      \vskip0.7\baselineskip\noindent
      \relatedversionHeading\@relatedversion
    \fi
    \ifx\@supplement\@empty\else
      \vskip0.7\baselineskip\noindent
      \supplementHeading\@supplement
    \fi
    \ifx\@funding\@empty\else
      \vskip0.7\baselineskip\noindent
      \fundingHeading\ifx\authoranonymous\relax\textcolor{red}{Anonymous funding}\else\@funding\fi
    \fi
    \ifx\@acknowledgements\@empty\else
      \vskip0.7\baselineskip\noindent
      \acknowledgementsHeading\ifx\authoranonymous\relax\textcolor{red}{Anonymous acknowledgements} \else\@acknowledgements\fi
    \fi
    \fi
  }
\def\@maketitle{\newpage
  \null\vskip-\baselineskip
  \vskip-\headsep
  \@titlerunning
  \@authorrunning
\parindent\z@ \raggedright
  \if!\@title!\def\@title{\textcolor{red}{Author: Please fill in a title}}\fi
    {\LARGE\sffamily\bfseries\mathversion{bold}\@title \par}\vskip 1em
    \ifx\@author\orig@author
\else
      {\def\thefootnote{\@arabic\c@footnote}\setcounter{footnote}{0}\fontsize{9.5}{12}\selectfont\@author}\fi
    \bgroup
      \immediate\openout\tocfile=\jobname.vtc
      \protected@write\tocfile{
			\let\footnote\@gobble
			\let\thanks\@gobble
			\def\footnotemark{}
			\def\and{and }\def\,{ }
			\def\\{ }
		}{\string\contitem
	        \string\title{\@title}\string\author{\@authorsfortoc}\string\page{\@ArticleNo:\thepage--\@ArticleNo:\number\numexpr\getpagerefnumber{LastPage}}}\closeout\tocfile
    \egroup
  \par}
\makeatother

\newcommand{\colsep}{15pt}

\newcommand{\sesstype}[1]{\mkwd{sessionType}(#1)}
\newcommand{\bhv}[1]{\mkwd{behaviour}(#1)}

\newcommand{\zapsymb}{\stmaryrdLightning}
\newcommand{\zap}[1]{\zapsymb #1}
\newcommand{\zaptwo}[2]{\zapsymb \sessindextwo{#1}{#2}}
\newcommand{\raiseexn}{\calcwd{raise}}

\newcommand{\trycatch}[2]{\calcwd{try} \: #1 \: \calcwd{catch} \: #2}
\newenvironment{mathparsmall}
{\begin{mathpar}\footnotesize}{\end{mathpar}}
\newcommand{\fn}[1]{\mkwd{fn}(#1)}

\newcommand{\zapenv}{\Theta}
\newcommand{\zapfn}[2]{\mkwd{zap}(#1, #2)}
\newcommand{\zapindex}[2]{#1[#2] : \zapsymb }
\newcommand{\lblzap}[3]{#1 : #2 \zapsymb #3}
\newcommand{\vdashzap}{\vdash_{\zapsymb}}
\newcommand{\failed}[1]{\mkwd{failed}(#1)}
\newcommand{\oflat}[1]{\mkwd{flat}(#1)}
\newcommand{\prog}[1]{\mkwd{prog}(#1)}

\newcommand{\ready}[1]{\mkwd{ready}(#1)}
\newcommand{\confzero}{\mathbf{0}}
\newcommand{\tylit}[1]{\mkwd{#1}}
\newcommand{\var}[1]{\textit{#1}}
\newcommand{\ep}{E_{\mkwd{P}}}

\newcommand{\synceval}{\Longrightarrow}
\newcommand{\synclbl}{\rho} \newcommand{\syncevalstar}{\synceval^*}
\newcommand{\maybesynceval}{\synceval^{?}}
\newcommand{\secref}[1]{\S\ref{#1}}
\newcommand{\secrefp}[1]{(\S\ref{#1})}

\newcommand{\enableseq}[2]{#1 \vdash #2}
\newcommand{\unf}[1]{\mkwd{unf}(#1)}
\authorrunning{P.\ Harvey, S.\ Fowler, O.\ Dardha, and S.\ J.\ Gay} 

\Copyright{Paul Harvey, Simon Fowler, Ornela Dardha, and Simon J.\ Gay} 



\category{} 

\relatedversion{} 

\supplement{}

\funding{Supported by EPSRC grants EP/T014628/1 (STARDUST), EP/K034413/1 (ABCD), EP/L01503X/1 (CDT in Pervasive Parallelism), ERC Consolidator Grant Skye (682315), and by the EU HORIZON 2020 MSCA RISE project 778233 (BehAPI).}

\acknowledgements{Thanks to Phil Trinder for helpful comments and discussions, and to the anonymous reviewers for exceptionally detailed reviews.}

\hideLIPIcs  

\author{Paul Harvey}
{Rakuten Mobile Innovation Studio}
{paul@paul-harvey.org}
{https://orcid.org/0000-0003-1243-938X}
{}

\author{Simon Fowler}
{School of Computing Science, University of Glasgow}
{Simon.Fowler@glasgow.ac.uk}
{https://orcid.org/0000-0001-5143-5475}
{}

\author{Ornela Dardha}
{School of Computing Science, University of Glasgow}
{Ornela.Dardha@glasgow.ac.uk}
{https://orcid.org/0000-0001-9927-7875}
{}

\author{Simon J. Gay}
{School of Computing Science, University of Glasgow}
{Simon.Gay@glasgow.ac.uk}
{https://orcid.org/0000-0003-3033-9091}
{}

\EventEditors{Manu Sridharan and Anders M{\o}ller}
\EventNoEds{2}
\EventLongTitle{35th European Conference on Object-Oriented Programming (ECOOP 2021)}
\EventShortTitle{ECOOP 2021}
\EventAcronym{ECOOP}
\EventYear{2021}
\EventDate{July 12--16, 2021}
\EventLocation{Aarhus, Denmark (Virtual Conference)}
\EventLogo{}
\SeriesVolume{194}
\ArticleNo{12}

\begin{document}

\title{Multiparty Session Types for Safe Runtime Adaptation in an Actor Language\\(Extended version)}
\titlerunning{Multiparty Session Types for Safe Runtime Adaptation in an Actor Language}

\maketitle

\begin{abstract}
Human fallibility, unpredictable operating environments, and
the heterogeneity of hardware devices
are driving the need for software to be able to \emph{adapt}
as seen in the Internet of Things or telecommunication networks.
Unfortunately, mainstream programming languages do not readily allow a software
component to sense and respond to its operating environment, by
\emph{discovering}, \emph{replacing}, and \emph{communicating} with 
components that are not part of the original system design, while
maintaining static correctness guarantees.  In particular, if a new component is
discovered at runtime, there is no guarantee that its communication
behaviour is compatible with existing components.

We address this problem by using \emph{multiparty session types with explicit
connection actions}, a type formalism used to model distributed communication
protocols. By associating session types with software components, the discovery
process can check protocol compatibility and, when required, correctly replace
components without jeapordising safety.

We present the design and implementation of \EnsembleS, the \emph{first}
actor-based language with adaptive features and a static session type system,
and apply it to a case study based on an adaptive DNS server. We formalise the
type system of \EnsembleS and prove the safety of well-typed programs, making
essential use of recent advances in \emph{non-classical} multiparty session
types.
\end{abstract}

\section{Introduction}
The era of single monolithic stand-alone computers has long been replaced by a
landscape of heterogeneous and distributed computers and software applications.
Technologies such as the
IoT~\cite{XiaYWV12:iot}, self-driving cars~\cite{10.1007/978-3-319-00476-1_14}, or autonomous networks~\cite{5371088} bring the new challenge of needing to successfully operate in face of ever-changing environments, technologies, devices, and human errors, necessitating the need to \emph{adapt}.

Here, we define \emph{dynamic self-adaptation}---hereafter referred to as \emph{adaptation}---as the ability of a software component to sense and respond to its operating environment, by \emph{discovering}, \emph{replacing}, and \emph{communicating} with other software components at runtime that are not part of the original system design \cite{10.5555/303461.342800, rellermeyer2007r}.
There are many examples of adaptive systems, as well as the mechanisms of
adaptation they leverage, such as discovery~\cite{jara2012light},
modularisation~\cite{DBLP:journals/dse/HaytonBDHH99}, dynamic code loading and
migration~\cite{CesariniV16, 1369163}. Commercially, Steam's in-home
streaming system\footnote{\url{http://store.steampowered.com/streaming/}}
enables video games to dynamically transfer their input/output across a range of
devices. Academically, RE$^\mathrm{X}$~\cite{199370} enables software to
self-assemble predefined components, using machine learning to reconfigure the
software in response to environmental changes.

Despite strong interest in adaption and substantial work on the mechanisms of
adaptation, current programming languages either lack the capabilities to ensure
that adaptation can be achieved safely and correctly, or they check correctness
dynamically, resulting in runtime overheads which may not be acceptable for
resource-constrained devices.

Specifically, if an adaptive system {discovers} new software components at
runtime,
these components must interact with the system in a purposeful manner. In
concurrent and distributed systems, such interaction goes beyond a simple
function call / return expressed with standard types and type systems:
interaction involves complex \emph{communication protocols} that constrain the
sequence and type of data exchanged. For example, knowing that two components
communicate integers and strings does not describe if or when they will be sent
or received.
In spite of growing interest in the topic, for example, the recent formation of the United
Nations group considering \emph{creative
adaptation}\footnote{\url{https://www.itu.int/en/ITU-T/focusgroups/an/Pages/default.aspx}},
mainstream programming languages do not support the
\emph{specification} and \emph{verification} of communication protocols in
concurrent and distributed systems. In turn, errors are discovered late in the
development process and potentially after deployment.

Even where all components are known statically, communication safety cannot be
guaranteed: as an example, the RE$^\mathrm{X}$ system's programming language specifies
sequential call / return interfaces for components, but not communication
protocols for concurrent components.
The adaptation in the Steam in-home streaming system is even more limited, being restricted to detection of input/output devices from a set of compatible possibilities. In both cases, the adaptive aspects of the software have been defined and designed ahead of time, as opposed to being composed \emph{on-demand} at runtime, leaving no scope for extending the system via runtime discovery and replacement.

This situation brings us to a key research question:
\begin{center}
\textbf{RQ}:
Can a programming language support \emph{static (compile-time) verification} of safe runtime dynamic self-adaptation, i.e., \emph{discovery, replacement and communication}?
\end{center}

The problem of static verification of safe communication is addressed by \emph{multiparty session types} \cite{Honda93, HondaVK98, Honda:2008:MAS:1328897.1328472}.
Multiparty session types (MPSTs) are a type formalism used to specify the
\emph{type}, \emph{direction} and \emph{sequence} of communication actions
between two or more participants. Session types guarantee that software conforms
to predefined communication protocols, rather than risking errors manifesting
themselves at runtime.

There is already some work in the literature on adaptation and session types,
but it does not answer our research question. We discuss related work in
\S\ref{sec:related}, but in brief, the state-of-the-art has some combination of
the following limitations: theory for a formal model such as the $\pi$-calculus
\cite{CoppoDV15,CastellaniDP16,DiGiustoP15,DiGiustoP15a}, rather than a
real-world programming language; omission of some aspects of adaptation, such as
runtime discovery \cite{Hu2017}; or verification by runtime monitoring
\cite{NeykovaY14,NeykovaY17,Fowler2016}, as opposed to static checking.

To answer our research question,
we implement \EnsembleS, the first actor language leveraging MPSTs
to provide compile-time verification of safe dynamic runtime adaptation:
we can statically guarantee that a discovered actor will comply
with a communication protocol, and guarantee that replacing an actor's behaviour
(e.g., to fix a bug) will not jeapordise communication safety.
Key to our approach is the combination of the actor
paradigm~\cite{Hewitt:1973:UMA:1624775.1624804}, for its process addressability
and explicit message passing, with \emph{explicit connection
actions}~\cite{Hu2017} in multiparty session types, which allow discovered
actors to be invited into a session.

\subparagraph{Contributions.}
The overarching contribution of this work is the design, implementation, and
formalisation of a language which supports dynamic self-adaptation while
guaranteeing communication safety. We achieve this through a novel integration
of an actor-based language and multiparty session types with explicit connection
actions. Specifically, we introduce:

\begin{enumerate}
\item
\textbf{\EnsembleS and its compiler} (\autoref{sec:ensemble-sessions}): we
present an actor language, \EnsembleS, which supports safe adaptable
applications using MPSTs. Our framework supports:
\begin{itemize}
\item MPST specifications, both standard and using explicit
  connection actions (\autoref{sub:channel-connections});
\item MPSTs to provide guarantees of protocol compliance in runtime discovery (\autoref{sub:discovery});
\item automatic generation of application code from MPSTs~(\autoref{sub:auto-skeleton-generation})
\end{itemize}
\item
\textbf{An adaptive DNS case study} (\autoref{sec:DNS}): using MPSTs and runtime discovery to show safe dynamic self-adaptation can be achieved in a non-trivial software service
\item
\textbf{A core calculus for \EnsembleS} (\autoref{sec:core-calculus}): we
formalise \EnsembleS and prove type safety and progress.
\end{enumerate}

The formalism makes several technical contributions: it is the first
actor-based calculus with statically-checked MPSTs; and it is the first
\emph{calculus} to provide a language design and semantics for explicit
connection actions, which had previously only been explored at the type level.
Our design requires exception handling in the style of~\citet{MostrousV18}
and~\citet{FowlerLMD19:stwt}, and the metatheory makes essential (and novel) use
of \emph{non-classical} multiparty session types~\cite{ScalasY19}.

The implementation and examples are available in the paper's companion artifact.
 
\begin{figure}[t]
  \small
\begin{minipage}[t]{0.475\textwidth}
~\header{Global protocol}
\vspace{-1em}
\begin{lstlisting}[basicstyle=\scriptsize, numbers={left}, language = ensemble]
global protocol Bookstore
    (role Sell, role Buy1, role Buy2) {
  book(string) from Buy1 to Sell;
  book(int) from Sell to Buy1;
  quote(int) from Buy1 to Buy2;
  choice at Buy2 {
    agree(string) from Buy2 to Buy1, Sell;
    transfer(int) from Buy1 to Sell;
    transfer(int) from Buy2 to Sell;
  } or {
    quit(string) from Buy2 to Buy1, Sell;
  } }
\end{lstlisting}\end{minipage}
\hfill
\begin{minipage}[t]{0.475\textwidth}
~\header{Local protocol for \lstinline+Sell+}
\vspace{-1em}
\begin{lstlisting}[language = ensemble, label={lst:scribble-buy1}]
local protocol Bookstore_Sell
    (self Sell,role Buy1,role Buy2) {
  book(string) from Buy1;
  book(int) to Buy1;
  choice at Buy2{
    agree(string) from Buy2;
    transfer(int) from Buy1;
    transfer(int) from Buy2;
  } or {
    quit(string) from Buy2;
} }
\end{lstlisting}
\end{minipage}
\caption{Global and local protocols for Bookstore}\label{lst:bookstore}
\end{figure}

\section{Multiparty Session Types}
\label{sec:session-types}

\emph{Multiparty session types} \cite{Honda:2008:MAS:1328897.1328472} are a type
formalism used to describe communication protocols in
concurrent and distributed systems.  An MPST describes communication among
multiple software components or participants, by specifying the \emph{type} and
the \emph{direction} of data exchanged, which is given as a sequence of send and
receive actions.

We first introduce MPSTs (formalised in \autoref{sec:core-calculus}) via
\emph{Scribble}~\cite{Yoshida:2013:SPL:3092395.3092399}, a specification
language for communicating protocols based on the theory of
multiparty session types.  We start with a \emph{global type}, which describes
the interactions among all communicating participants.  Using the Scribble tool,
a global protocol can be \emph{validated}, guaranteeing its correctness, and
then \emph{projected} for each participant. Projection returns a \emph{local
type}, which describes communication actions from the viewpoint of that
participant.

\subparagraph{\texttt{Bookstore} example.}
\autoref{lst:bookstore} shows the classic \texttt{Bookstore} (also known
as \emph{Two-Buyer}) example, written in Scribble. 
We have three communicating
participants (\emph{roles}): two buyers \texttt{Buy1} and \texttt{Buy2}, and one seller
\texttt{Sell}, where the buyers wish to buy a book from the seller.
\texttt{Buy1} sends the title of the book of type
{\color{purple}\textbf{string}} to \texttt{Sell} (line 3).
Next, \texttt{Sell} sends the price of the book of type
{\color{purple}\textbf{int}} to \texttt{Buy1} (line 4).  At this stage,
\texttt{Buy1} invites \texttt{Buy2} to share the cost of the
book, by sending them a \texttt{quote} of type {\color{purple}\textbf{int}} that
\texttt{Buy2} should pay (line 5). It is \texttt{Buy2}'s \emph{internal choice}
(line 6) to either \texttt{agree} (line 7), or \texttt{quit} the protocol (line
11). After agreement, both \texttt{Buy1} and \texttt{Buy2} \texttt{transfer}
their \texttt{quote} to \texttt{Sell} (lines 8 and 9, respectively).

Projecting the \texttt{Bookstore} global protocol into each of the communicating
participants returns their local protocols. \autoref{lst:bookstore} shows the
local protocol for \texttt{Sell}; we omit \texttt{Buy1} and \texttt{Buy2} as
they are similar.  Note that the local protocol only includes actions relevant
to \texttt{Sell}.

\subparagraph{Explicit connection actions.}
The \texttt{Bookstore} protocol assumes that all roles are connected at the
start of the session. This is undesirable when a participant is only needed for
\emph{part} of a session, or the identity of a participant depends on data
exchanged in the protocol.

\begin{figure}[t]
  \begin{minipage}{0.5\textwidth}
\begin{lstlisting}[basicstyle=\scriptsize, numbers={left}, language = ensemble]
explicit global protocol OnlineStore
    (role Customer, role Store, role Courier) {
  login(string) connect Customer to Store;
  do Browse(Customer, Store, Courier);
}

aux global protocol Deliver
    (role Customer, role Store, role Courier) {
  address(string) from Customer to Store;
  deliver(string) connect Store to Courier;
  ref(int) from Courier to Store;
  disconnect Courier and Store;
  ref(int) from Store to Customer;
  disconnect Store and Customer;
}
\end{lstlisting}
\end{minipage}
\hfill
\begin{minipage}{0.475\textwidth}
\begin{lstlisting}[basicstyle=\scriptsize, numbers={left}, language = ensemble]
aux global protocol Browse
    (role Customer, role Store, role Courier) {
  item(string) from Customer to Store;
  price(int) from Store to Customer;
  choice at Customer {
    do Browse(Customer, Store, Courier);
  } or {
    do Deliver(Customer, Store, Courier);
  } or {
    quit() from Customer to Store;
    disconnect Store and Customer;
  }
}
\end{lstlisting}
\end{minipage}
\caption{Global protocol for \texttt{OnlineStore}}
\label{fig:onlinestore}
\end{figure}

Consider Figure~\ref{fig:onlinestore}, which details the protocol
for an online shopping service, inspired by the travel agency protocol detailed
by~\citet{Hu2017}. The protocol is organised as three
subprotocols: \texttt{OnlineStore}, the entry-point;
\texttt{Browse}, where the customer
repeatedly requests quotes for items; and \texttt{Deliver}, where
the store requests delivery from a courier.
In contrast to \texttt{Bookstore}, each connection must be established explicitly
(note that \lstinline+connect+ replaces \lstinline+from+ when initiating a
connection).

Note in particular that \texttt{Courier} is only involved in the
\texttt{Deliver} subprotocol. The store can therefore \emph{choose}
which courier to use based on, for example, the weight of the item or the
customer's location. Furthermore, it is not necessary to involve the courier if
the customer does not choose to make a purchase.
 
\section{\EnsembleS: An Actor Language for Runtime Adaptation}
\label{sec:ensemble-sessions}
In this section, we present \EnsembleS, a new session-typed actor-based
language based on Ensemble~\cite{harvey-thesis,
Harvey:2017:AAJ:3144555.3144559}.
\EnsembleS actors are addressable, single-threaded entities with share-nothing
semantics, and communicate via message passing. However, differently from the
classic definition of the actor model
\cite{Hewitt:1973:UMA:1624775.1624804,DBLP:books/daglib/0066897}, the
communication model in \EnsembleS is channel-based.
\EnsembleS supports both \emph{static} and \emph{dynamic} topologies:
\begin{description}
    \item[Static Topologies] All participants are present at the start of the session and remain involved for
the duration of the session. This is based on traditional MPSTs~\cite{Honda:2008:MAS:1328897.1328472}.
    \item[Dynamic Topologies] Participants can connect and disconnect during a session. This builds on the more recent idea of explicit connection actions~\cite{Hu2017}.
\end{description}

\subsection{\EnsembleS: basic language features}
\label{sub:lang-overview}

\begin{wrapfigure}{l}{0.425\textwidth}
\begin{lstlisting}[escapechar=@,language=ensemble]
type Isnd is interface(out integer output)
type Ircv is interface(in integer input)

stage home {
	actor sender presents Isnd {
		value = 1;
		constructor() {}
		behaviour {
			send value on output;
			value := value + 1;
	} }
	actor receiver presents Ircv {
		constructor() {}
		behaviour{
			receive data from input;
			printString("\nreceived: ");
			printInt(data);
	} }
  boot {
    s = new sender();
    r = new receiver();
    establish topology(s, r);
  }
}
  \end{lstlisting}
  \caption{A simple \EnsembleS program\vspace{-0.75em}}
\label{lst:ensemble-send-recv}
\end{wrapfigure}
An \EnsembleS actor has its own private state
and a single thread of control expressed as a \texttt{behaviour}
clause, which is repeated until explicitly told to stop.
Every actor executes within a \texttt{stage}, which represents a memory space.
Actors do not share state, but instead communicate via message passing along
half-duplex, simply-typed channels.

\autoref{lst:ensemble-send-recv} shows a simple \EnsembleS program which
defines, instantiates and connects two actors, one of which sends linearly
increasing values to the other.
The program defines two interfaces \lstinline+Isnd+ and \lstinline+Ircv+, declaring an output and input channel respectively.
The \lstinline+boot+ clause (lines 19--23) is
executed first and creates an instance of each actor (lines 20--21), using the
appropriate \lstinline+constructor+ (lines 7 and 13, respectively). This creates
and begins executing new threads for each actor, which follow the logic of the
relevant \lstinline+behaviour+ clause. Next, the \lstinline+boot+ clause binds the
actor's channels together (line 22, discussed later in
\autoref{sub:channel-connections}).
Once bound, the
\lstinline+sender+ actor sends the contents of \lstinline+value+ on its channel,
increments it, and goes back to the beginning of its behaviour loop (lines
8--11). The \lstinline+receiver+ actor waits for a message, binds the message to
\lstinline+data+, displays it, and returns to the top of its behaviour loop
(lines 14--18).
\EnsembleS inherits Ensemble's native support for
runtime software adaptation actions~\cite{Harvey:2017:AAJ:3144555.3144559}:

\begin{description}
    \item[\underline{Discover}] The ability to \emph{locate} an arbitrary actor or stage reference at runtime, given an \texttt{interface} and \texttt{query}.
    \item[Install] Given an actor type, the ability to \texttt{spawn} it at a specified stage.
    \item[Migrate] The ability for an executing actor to \emph{move} to another
        stage.
    \item[\underline{Replace}] The ability to \emph{replace} an executing actor A by a new instantiation of actor B, the latter continuing at the same stage as A, if A and B have the same \texttt{interface}.
    \item[\underline{Interact}]: Given an actor reference (either spawned, discovered, or communicated),  the ability to \emph{connect} to its channels at runtime and then communicate.
\end{description}

We focus on the \underline{underlined} actions and apply session types to guarantee communication safety. The reason for this choice is that discover, replace and interact are actions that modify \textit{how} actors operate, whereas the other actions, install and migrate, affect \textit{where} actors operate, but not their behaviour.

\begin{figure*}[t]
	\centering
  {\small
    \clearpage{}\begin{tikzpicture}[font=\sffamily]
  \node[mybox=colBox] (SGP) {Scribble\\Global Protocol};
  \node[mybox=colBox,right=of SGP] (SLP) {Scribble\\Local Protocol};
  \node[mybox=colBox,right=of SLP] (EST) {EnsembleS\\Template};
  \node[mybox=colBox,right=of EST] (EXE) {Executable};
\draw[myarrow=SLP] (SGP) -- (SLP);
  \draw[myarrow=EST] (SLP) -- (EST);
  \draw[myarrow=EXE] (EST) -- (EXE);
  \path (SGP) -- node[below, yshift=-0.75cm]{Scribble Tool} (SLP);
  \path (SLP) -- node[below, yshift=-0.75cm]{StMungo} (EST);
  \path (EST) -- node[below, yshift=-0.75cm]{EnsembleS Compiler} (EXE);
\end{tikzpicture}\clearpage{}
\vspace{-2em}
  }
		\caption{Automatic Actor Skeleton Generation Process}
	\label{fig:skeleton-generation}
\end{figure*}
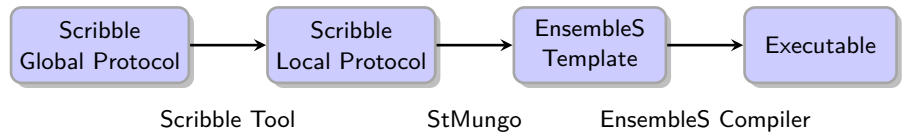

\subsection{Session types in \EnsembleS}
\label{sub:lang-mod}
\label{sub:auto-skeleton-generation}

A \texttt{session} type in \EnsembleS represents a communication protocol for an actor, i.e., a local protocol (or local session type) validated and projected from a global session type.

We extend the StMungo~\cite{Kouzapas:2016:TPM:2967973.2968595} tool to generate \EnsembleS template code that supports session types. \autoref{fig:skeleton-generation} shows an overview of the actor template code generation from a global session type, and \autoref{lst:ensemble-static-session} shows an example of the generated code.

First, a developer defines a global session type in Scribble~\cite{Yoshida:2013:SPL:3092395.3092399} (\autoref{fig:skeleton-generation}, first stage).
The Scribble tool checks that the protocol is well-formed and valid according to MPST theory and \emph{projects} the global protocol into local protocols for each participant (\autoref{fig:skeleton-generation}, second stage).
For each local protocol, the StMungo tool produces (\autoref{fig:skeleton-generation}, third stage) $i$) the \texttt{session} type, $ii$) the \textcolor{purple}{\bf interface} and type definitions, and $iii$) the \textcolor{red}{\bf actor} template.
The generated code is parsed by the \EnsembleS compiler, producing executable code (\autoref{fig:skeleton-generation}, fourth stage).

Let us now look at the \texttt{Buy1} local protocol, given in \autoref{lst:bookstore}.
Following the code generation process in \autoref{fig:skeleton-generation}, the \EnsembleS template items $i$), $ii$) and $iii$) for \texttt{Buy1} correspond respectively to the code blocks starting in lines 3, 14, and 24 in \autoref{lst:ensemble-static-session}.

The \texttt{Buy1} local protocol is translated as an \EnsembleS \texttt{session} type in \autoref{lst:ensemble-static-session} (lines 3--12).
It shows a sequence of send and receive actions (lines 4--6), followed by a choice at \texttt{Buy2} (lines 7--12), which determines the next set of communication actions.

Following session type specifications,
\EnsembleS channels define both the payload type and the \texttt{session} that this channel expects to interact with (lines 14--21,  \autoref{lst:ensemble-static-session}). The \EnsembleS compiler uses this information to ensure that the \texttt{session} of each channel matches the \texttt{session} associated with the actor it is connected to.

An \texttt{actor} may \texttt{follow} a \texttt{session} type (line 24, \autoref{lst:ensemble-static-session}).
This tells the \EnsembleS compiler that the logic within the \texttt{behaviour} clause of that actor must follow the communication protocol defined in the \texttt{session}.

It is important to note that the code generation in
\autoref{fig:skeleton-generation} is optional and the \EnsembleS typechecker is
independent of this process.

\begin{figure}[t]
\setlength{\columnsep}{\colsep}
  \begin{lstlisting}[ escapechar=@,
  language=ensemble,
basicstyle=\ttfamily\scriptsize,
  numbers={left}, multicols=2]
// FILE AUTOMATICALLY GENERATED
//************SESSIONS**************
type Buy1 is session(
	book(string) to Sell;
	book(int) from Sell;
	quote(int) to Buy2;
	choice at Buy2{
		Choice0_agree(string) from Buy2;
		transfer(int) to Sell;
	}	or {
		Choice0_quit(string) from Buy2;
	} )
//***********INTERFACES*************
type Buy1I is interface(
	out {Seller, string} toSell_string,
	in {Seller, integer} fromSell_integer,
	out {Buy2, integer} toBuy2_integer,
	in {Buy2, Choice0} fromBuy2_agreequit,
	in {Buy2, string} fromBuy2_string,
	out {Sell, integer} toSell_integer,
)
//*************ACTORS***************
stage home{
actor Buy1A presents Buy1I follows Buy1 {
 constructor() {}
 behaviour {
 payload1 = "";
 send payload1 on toSell_string;
 receive payload2 from fromSell_integer;
 payload3 = 42;
 send payload3 on toBuy2_integer;
 // Receive choice from other actor
 receive payload4 from fromBuy2_agreequit;
 switch(payload4) {
  case Choice0_agree:
    receive payload5 from fromBuy2_string;
    payload6 = 42;
    send payload6 on toSell_integer;
    break;
  case Choice0_quit:
    receive payload7 from fromBuy2_string;
    break;
  }
} }
// Omitted: Buy2A and SellA actors
boot {
 buyer1 = new Buy1A();
 buyer2 = new Buy2A();
 seller = new SellA();
 // other actors...
 establish topology(buyer1,buyer2,seller);
} }
\end{lstlisting}
\caption{EnsembleS static session template}
\label{lst:ensemble-static-session}
\end{figure}

\subsection{Channel connections: static and dynamic}
\label{sub:channel-connections}
If an actor follows a \texttt{session} type, then its channel connections must be \emph{1-1}. This is the standard linearity requirement for session types: if there are multiple senders on one channel, then their messages can interfere and it is not possible to statically check that the session is followed correctly.
\EnsembleS avoids this problem by using a single channel for each message type between each pair of participants.
For example, in \autoref{lst:bookstore}, each of the three actors communicates strings and integers with both of the other actors. Because channels are unidirectional, each actor therefore has 8 channels: 2 to send strings and 2 to send integers to both other actors, and similarly 4 channels for receiving.

\mysubsection{Static connections.}
\label{par:topology-stmt}
When using session types with \emph{static} topologies, and all actors in the
session are known from the beginning of the application, \EnsembleS provides the
\texttt{establish topology} statement to create the connections between the
specified \texttt{session} actors (line 22, \autoref{lst:ensemble-send-recv};
line 51, \autoref{lst:ensemble-static-session}).
A compile-time error is generated if the topology is ill-defined (e.g., if the
\lstinline+session+s do not compose or if the channels do not match).

\mysubsection{Dynamic connections.}
\EnsembleS supports reconfigurable channels and \emph{dynamic} connections, via \texttt{link} and \texttt{unlink} statements.
The \texttt{link} statement takes two references to actors which follow \texttt{session}s (line 6, \autoref{lst:ensemble-session-discovery}). It connects all of the channels of the two specified actors such that the channel and actor's \texttt{session}s match.
A compile-time error is raised if the \texttt{session}s are incompatible.
Conversely, the \texttt{unlink} statement disconnects (line 9, \autoref{lst:ensemble-session-discovery}).

\subsection{Adaptation via discovery and replacement}
\label{sub:discovery}

\begin{wrapfigure}{l}{0.45\textwidth}
\begin{lstlisting}[ escapechar=@,
  language=ensemble,
basicstyle=\ttfamily\scriptsize,
  numbers={left}]
// define query alpha
query_a = alpha();
actor_s = discover(
	Buyer1_interface, @\fcolorbox{black}{green}{Buyer1\_session}@, query_a);
if (actor_s[0].length > 1){
	@\fcolorbox{black}{yellow}{link me with actor\_s[0];}@
	msg = "book";
	send msg on toB_string;
	@\fcolorbox{black}{yellow}{unlink Buyer1\_session;}@
}
  \end{lstlisting}
  \caption{Session type-based discovery \vspace{-1em}}
  \label{lst:ensemble-session-discovery}
\end{wrapfigure}
\EnsembleS supports runtime discovery of \emph{local} or \emph{remote} actor
instances. As an example, in a sensor network, it may be desirable to connect to a
sensor which has a battery level above a certain threshold. The \EnsembleS query
language allows the user to define a query on non-functional properties (such as
battery level, signal strength, or name), as well as the channels exposed by an
actor's interface. This ensures that any discovered actor has the correct number
and type of channels, and satisfies user's preferences. To ensure that the
discovered actor also obeys a declared protocol, \EnsembleS uses
\texttt{session} types in the discovery process. The green box in
\autoref{lst:ensemble-session-discovery} shows how a \texttt{session} is
used in the actor discovery process, and the yellow box shows how such actors
are connected together.
Runtime discovery does not appear in the \texttt{session} because it does not affect the behaviour of an application directly. 

\begin{figure}[t]
  \begin{lstlisting}[escapechar=@,
  language=ensemble,
basicstyle=\ttfamily\scriptsize,
  numbers={left}, multicols=2]
// session and interface definitions
actor fastA presents accountingI
		    follows accountingSession{
	constructor() {}
	behaviour{
		receive data on input;
		quicksort(data);
		send data on output;
  }
}

actor slowA presents accountingI
		    follows accountingSession{
	pS= new property[2] of property("",0);
	constructor() {
		pS[0]:= new property("serial",823);
		pS[1]:= new property("version",2);
		publish pS;
	}
	behaviour{
		receive data on input;
		bubblesort(data);
		send data on output;
}	}

query alpha() { $serial==823 && $version<4; }

actor main presents mainI {
 constructor() { }
 behaviour {
	// Find the slow actors matching query
	actor_s = discover(accountingI,
    accountingSession, alpha());
	// Replace them with efficient versions
	if(actor_s[0].length > 1) {
    replace actor_s[0] with fastA();
	}
} }
  \end{lstlisting}
\caption{Session type-based replacement}
\label{lst:ensemble-session-replacement}
\end{figure}
\EnsembleS also supports the replacement of executing actors, much like the
hot-code swapping in Erlang~\cite{CesariniV16}.  The new actor
must present the same interface as it \emph{takes over} the channels of the
actor being replaced at the location it was executing.  Replacement happens at
the beginning of an actor's \texttt{behaviour} loop. Replacement has many uses,
such as updating, changing, or extending some of the functionalities of existing
software, and is particularly useful in embedded
systems~\cite{Hughes:2009:LLC:1821748.1821787, Hui:2004:DBD:1031495.1031506}.
The existing and new actors must follow the same \texttt{session}
type, guaranteeing that replacement will not break existing actor
interactions.

\autoref{lst:ensemble-session-replacement} shows an example of a \emph{main}
actor searching for actors of type \texttt{slowA} (line 32), and replacing them
with new actors of type \texttt{fastA} (lines 35--37). The \texttt{slowA} actors are located by defining a query (line 26) over user-defined properties, which are published (lines 16--18). The discovery process is the same as above, but now the discovered actors are used for replacement rather than just communication.

\subsection{Implementation}

\EnsembleS is implemented in C, and supports reference-counted garbage
collection and exceltions. Applications are compiled to Java source code,
and then to custom Java class files for use with a custom VM~\cite{6775058}.
These applications can be executed on the desktop, parallel accelerators
(e.g.~GPUs), Raspberry Pi, Lego NXT, and Tmote Sky hardware platforms, and use a
range of networking technologies.

Compact representations of session types are retained at runtime in order to
support discovery.
\EnsembleS skeleton generation code is based on the StMungo
tool~\cite{Kouzapas:2016:TPM:2967973.2968595}, which is implemented as an ANTLR
listener, and session typechecking is supported by modifying the original
Ensemble typechecker to ensure that each communication action is permitted by
the actor's declared session type.

Since \EnsembleS builds directly on top of the original Ensemble implementation,
it inherits Ensemble's runtime system. Performance results can be found
in~\cite{Harvey:2017:AAJ:3144555.3144559}.

\clearpage

\section{Case study: DNS}
\label{sec:DNS}
\begin{figure}[t]
\setlength{\columnsep}{\colsep}
\begin{lstlisting}[ escapechar=@,
  language=ensemble,
basicstyle=\ttfamily\scriptsize,
  numbers={left}, multicols=2]
type Iclient is interface(
	out{RootServer,string} RootServer_stringOut,
	in {RootServer,string} RootServer_stringIn,
	out{ZoneServer,string} ZoneServer_stringOut,
	in {ZoneServer,string} ZoneServer_stringIn,
	in {ZoneServer,choice_enum}	ZoneServer_choiceIn,
	in {RootServer,choice_enum}	RootServer_choiceIn)

type choice_enum is
	enum(TLDResponse,PartialResolution,
	InvalidDomain,ResolutionComplete,
	InvalidTLD)

query find_name(string n){ $name == n; }

actor c presents Iclient
	follows Client {
	dom_name = "nii.ac.jp";
	constructor() { }
	behaviour{
		rootQuery = find_name("jp");
		// Find Root Server
		root_s =
			discover(IServer, RootServer, rootQuery);
		// search until root_s non-empty
		link me with root_s[0];
		send domain_name on RootServer_stringOut;
		receive c_msg	from RootServer_choiceIn;
		switch(c_msg){
			case TLDResponse:
				receive ZoneServerAddr_msg
					from RootServer_stringIn;
			unlink RootServer;
while(true) Lookup : {
	// Find ZoneServer
	zone_s =
    discover(IServer,	ZoneServer,
      find_name(ZoneServerAddr_msg));
	link me with zone_s[0];
	// Ask ZoneServer
	send dom_name on ZoneServer_stringOut;
	receive c_msg2 from ZoneServer_choiceIn;
	switch(c_msg2){
		case PartialResolution:
			receive str_msg from ZoneServer_stringIn;
			ZoneServerAddr_msg := str_msg;
			unlink ZoneServer;
			continue Lookup;
		case InvalidDomain:
			receive str_msg from ZoneServer_stringIn;
			unlink ZoneServer;
			break;
		case ResolutionComplete:
			receive str_msg from ZoneServer_stringIn;
			unlink ZoneServer;
			break Lookup;
		}
		// keep looking
	}
case InvalidTLD:
	 receive str_msg from RootServer_stringIn;
	 unlink RootServer;
} } }
  \end{lstlisting}
\caption{EnsembleS DNS client}
\label{lst:ensemble-dns-client}
\end{figure}

 \begin{wrapfigure}{l}{0.525\textwidth}
\begin{lstlisting}[ escapechar=@,language=ensemble]
type Client is session(
  connect RootServer;
  RootRequest(DomainName) to RootServer;
  choice at RootServer{
    TLDResponse(ZoneServerAddress) from RootServer;
    disconnect RootServer;
    rec Lookup {
      connect ZoneServer;
      ResolutionRequest(DomainName) to ZoneServer;
      choice at ZoneServer {
        PartialResolution(ZoneServerAddress)
          from ZoneServer;
        disconnect ZoneServer;
        continue Lookup;
      } or {
        InvalidDomain(String) from ZoneServer;
        disconnect ZoneServer;
      } or {
        ResolutionComplete(IPAddress) from ZoneServer;
        disconnect ZoneServer;
      }
    }
  } or {
    InvalidTLD(String) from RootServer;
    disconnect RootServer;
  }
)
\end{lstlisting}
\caption{EnsembleS DNS client \texttt{session} type \vspace{-1em}}
\label{lst:ensemble-dns-client-session}
\end{wrapfigure}
 To illustrate the use of session types for adaptive programming, we consider a
real-world case study: the domain name system (DNS). DNS is a hierarchical,
globally distributed translation system that converts an internet host name
(domain name) into its corresponding numerical Internet Protocol (IP)
address~\cite{RFC1035}.

The process begins by transmitting a domain name to one of many well-known \emph{root} servers. This server either rejects bad requests, or provides the information to contact a \emph{zone} server. The zone server may know the IP address of the domain name; if not it refers the request to another zone server. This process continues until either the IP address is returned, or the name cannot be found.

To develop an adaptive DNS example, we assume no \emph{a priori} information
about server location, and instead use explicit discovery to find
root and zone servers based on session types and server properties. We use an
existing Scribble description of DNS as a starting point~\cite{Fowler2016}. To
illustrate adaptation we focus on the client who is querying DNS.

\autoref{lst:ensemble-dns-client-session} shows the \texttt{session} type for the client actor which asks DNS to resolve a domain name. The client first asks for a root server (lines 2--3), and then either is informed that the request is invalid (lines 24--25) or recursively queries zone servers (lines 7--22) until the IP address is found (lines 19--20), or an error is reported (lines 16--17). Based on this \texttt{session}, StMungo generates \EnsembleS types and interface definitions and a skeleton actor. Minimally completing the generated skeleton produces the code in \autoref{lst:ensemble-dns-client}.

In this example, discovery is used to locate the root server (lines 21--25, in
\autoref{lst:ensemble-dns-client}) and the zone server (line 37). In each case,
the \texttt{session} for the relevant server is provided to ensure that the
discovered actor follows the expected protocol.  When either server is located,
the client \texttt{link}s with it (lines 27 and 39), enabling communication.
When communication with the server is no longer required, the client
\texttt{unlink}s explicitly (lines 34, 47, 51, 55, 62). 

Although explicit discovery is used at the language level, there is nothing to
prevent the implementation of discovery from caching the addresses of the root
and zone servers. This does not affect the use of sessions in discovery or the
safety they provide, as the type-based guarantees are still enforced. However,
this would potentially improve performance of the system. Additionally, if a
cached entry becomes stale, the full discovery process can again be used without
code modification or degradation in trust.

A version of DNS which uses discovery allows the system to become more flexible
and resilient to changing operational conditions, such as topology changes in
the servers and their data. Session types ensure compatibility
with the discovered actors.
 
\section{A Core Calculus for EnsembleS}\label{sec:core-calculus}

In this section, we provide a formal characterisation of \textsc{EnsembleS}.
In doing so, we show that our integration of adaptation with multiparty session
types is safe, allowing adaptation while ruling out communication mismatches.

\subparagraph{Relationship to implementation.}
Our core calculus aims to distil the essence of the interplay between adaptation
and session-typed communication with explicit connection actions.
Therefore, we concentrate on a functional core calculus rather than an
imperative one: imperative variable binding serves only to clutter the
formalism, and our fine-grain call-by-value representation can be thought of as
an intermediate language.

Interfaces and unidirectional, simply-typed channels in
\textsc{EnsembleS} are an implementation artifact:
sending on a channel whose type changes is equivalent to sending on multiple
channels with different types.
Moreover, following theoretical accounts of multiparty session
types~\cite{Honda:2008:MAS:1328897.1328472, CoppoDYP16, Hu2017},
instead of having send and receive (resp.\ connect and accept) operations
followed by branching (as done in Mungo and StMungo), we have
unified \calcwd{send} and \calcwd{receive} constructs which communicate a label
along with the message payload.

Since session typing is the interesting part of discovery, we omit
properties and queries from the formalism; their inclusion is routine.
Finally, we concentrate on \emph{dynamic} topologies with explicit
connection actions rather than static topologies since they are important for
adaptation and more interesting technically.

\subsection{Syntax}

\begin{figure}[t]
{\footnotesize
\header{Syntax of Types and Terms}
\[
\begin{array}{lrcl}
\text{Actor class names} & \actorcls \\
\text{Actor definitions} & D & ::= & \actordef{\actorcls}{\loctya}{M} \\
\text{Roles} & \prole, \qrole, \srole, \trole \\
\text{Recursion Labels} & l \\
\text{Behaviours} & \behaviour & ::= & M \midspace \bstop \\
\text{Types} & A, B & ::= & \pidty{S} \midspace \one \\
\text{Values} & V, W & ::= & x \midspace () \\
\text{Actions} & L & ::= & \effreturn{V} \midspace \continue{l} \midspace \raiseexn \\
                    & & \midspace & \new{\actorcls} \midspace \self \midspace
                     \replace{V}{\behaviour} \midspace \newdiscover{S}  \\
                    & & \midspace & \newconn{\ell}{V}{W}{\prole}
                    \midspace  \newaccept{\prole}{\ell_i(x_i) \mapsto M_i}_i \\
                    & & \midspace & \newsend{\ell}{V}{\prole} \midspace
                    \newrecv{\prole}{\ell_i(x_i) \mapsto M_i}_i \\
                    & & \midspace & \newwait{\prole} \midspace \newdisconn{\prole} \\
\text{Computations} & M, N & ::= & \efflet{x}{M}{N} \midspace \trycatch{L}{M}
\midspace \annst{l}{M} \midspace L
\end{array}
\]
\header{Syntax of Session Types}
\[
\begin{array}{lrcl}
  \text{Session Actions} & \locacta, \locactb & ::= &
    \localtysend{\prole}{\ell}{\tya} \midspace \localtyconn{\prole}{\ell}{\tya}
      \midspace \localtyrecv{\prole}{\ell}{\tya}
      \midspace \localtyaccept{\prole}{\ell}{\tya} \midspace \localtywait{\prole} \\
\text{Session Types} & \loctya, \loctyb, \loctyc & ::= &
  \sumact{i \in I}{\alpha_i \then \locty_i} \midspace \recty{X}{\loctya} \midspace X \midspace
  \localtydisconn{\prole} \midspace \localend \\
\text{Communication Actions} & \commdir & ::= & {!} \midspace {?} \\
\text{Disconnection Actions} & \disconndir & ::= & {\waitact} \midspace {\disconnact}
\end{array}
\]
}
\caption{Syntax}
\label{fig:syntax}
\end{figure}

\subparagraph{Definitions.}
Figure~\ref{fig:syntax} shows the syntax of Core \textsc{EnsembleS} terms and types.
We let $\actorcls$ range over actor class names and $D$ range over definitions;
each definition $\actordef{\actorcls}{\loctya}{M}$ specifies the actor's class name,
session type, and behaviour.
Like class tables in Featherweight Java~\cite{IPW01:featherweight-java},
we assume a fixed mapping from class names to definitions.

\subparagraph{Values.}
Since our calculus is inherently effectful, we work in the setting of
\emph{fine-grain call-by-value}~\cite{LPT03:fgcbv}, where we have an
explicit static stratification of values and computations and an explicit
evaluation order similar to A-normal form~\cite{FlanaganSDF93}.
Values $V, W$ describe data that has been computed, and for the sake of simplicity,
consist of variables and the unit value.
Other base values (such as integers or booleans) can be encoded or added
straightforwardly.

\subparagraph{Computations.}
The $\efflet{x}{M}{N}$ construct evaluates $M$, binding its result to $x$ in
$N$.  The calculus supports exception handling over a single \emph{action} $L$
using $\trycatch{L}{M}$, where $M$ is evaluated if $L$ raises an exception, and
labelled recursion using $\annst{l}{M}$, stating that
inside term $M$, a process can recurse to label $l$ using $\continue{l}$.
\emph{Actions} $L$ denote the basic steps of a computation.
The $\effreturn{V}$ construct denotes a value.

\subparagraph{Concurrency and adaptation constructs.}
The $\new{\actorcls}$ construct spawns a new actor of class $\actorcls$ and
returns its PID. The $\self$ construct returns the current actor's PID.
An actor can replace the behaviour of
itself or another actor $V$ using $\replace{V}{\behaviour}$. An actor can
\emph{discover} other actors following a session type $\loctya$ using the
$\newdiscover{\loctya}$ construct, which returns the PID of the discovered
actor.

\subparagraph{Session communication constructs.}
An actor can connect to an actor $W$ playing role $\prole$ using
$\newconn{\ell}{V}{W}{\prole}$, sending a message with label $\ell$ and payload
$V$.
An actor can accept a connection from another actor playing role
$\prole$ using $\newaccept{\prole}{\msg{\ell_i}{x_i} \mapsto M_i}_i$, which
allows an actor to receive a choice of messages; given a message with label
$\ell_j$, the payload is bound to $x_j$ in the continuation $N_j$.
Once connected, an actor can communicate using the $\calcwd{send}$ and
$\calcwd{receive}$ constructs. An actor can disconnect from
$\prole$ using $\newdisconn{\prole}$, and await the disconnection of
$\prole$ using $\newwait{\prole}$.

\subparagraph{Types.}
Types, ranged over by $\tya, \tyb$, include the unit type $\one$ and process IDs
$\pidty{\loctya}$; the parameter $\loctya$ refers to the statically-known
initial session type of the actor (i.e., the session type declared in the
$\calcwd{follows}$ clause of a definition).
Unlike in channel-based session-typed systems, process IDs themselves need not
be linear: any number of actors can have a \emph{reference} to another actor,
but each actor may only be in a single session at a time. PIDs can be passed as
payloads in session communications.

\subparagraph{Session types.}

Session types are ranged over by $\loctya, \loctyb, \loctyc$ and follow the
formulation of~\citet{Hu2017}. A session type can be a choice of actions,
written $\sumact{i \in I}{\locact\then\loctya}$, a recursive session type
$\recty{X}{\loctya}$ binding recursion variable $X$ in continuation $\loctya$,
a recursion variable $X$, a disconnection action $\localtydisconn{\prole}$,
or the finished session $\localend$.
The syntax of session types is more liberal than traditional
`directed' presentations in order to allow output-directed choices to send or
connect to different roles.

Session actions
$\locact$ involve sending ($!$), receiving ($?$), connecting ($!!$), or
accepting ($??$) a message $\ell(A)$ with label $\ell$ and type $A$;
or awaiting another participant's disconnection
($\waitact$).
As well as disallowing self-communication, following~\citet{Hu2017},
we require the following syntactic restrictions on session types:

\begin{definition}[Syntactic validity]
  A choice type $S = \sumact{i \in I}{\locact_i \then \loctya_i}$ is \emph{syntactically valid}
  if:
  \begin{enumerate}
    \item it is an output choice, i.e., each $\locact_i$ is a send or connection
      action; or
    \item it is a \emph{directed} input choice, i.e., $S =\sumact{i \in I}{\localtyrecv{\prole}{\ell_i}{A_i} {.} \loctya_i}$
      or
      $S = \sumact{i \in I}{\localtyaccept{\prole}{\ell_i}{A_i} {.} \loctya_i}$; or
    \item the choice consists of single wait action $\localtywait{\prole} \then
      \loctya$.
  \end{enumerate}
\end{definition}
In the remainder of the paper, we assume that all session types are
syntactically valid.

\subparagraph{Session correlation.}
The most general form of explicit connection actions allows a participant to
leave and re-join a session, or accept connections from multiple different
participants. Such generality comes at a cost, since care must be taken to
ensure that the \emph{same} participant plays the role throughout the session.

To address this \emph{session correlation} issue,~\citet{Hu2017} propose two
solutions: either augment
global types with type-level assertions and check conformance dynamically, or adopt
a lightweight syntactic restriction which requires that
each local type must contain at most a single accept action as its top-level
construct. We opt for the latter, enforcing the constraint as part of our safety
property (\secref{sec:formalism:preservation}), and by requiring that
$\localtydisconn{\prole}$ does not have a continuation. (Note that the behaviour
will repeat, so $\prole$ will be able to accept again after disconnecting).
As~\citet{Hu2017} show, this design still supports the most common use cases of
explicit connection actions.

\subparagraph{Global types.}
Traditional MPST works~\cite{Honda:2008:MAS:1328897.1328472, CoppoDYP16}
use \emph{global types} to describe the interactions between participants at a
global level, which are then projected into \emph{local types}; projectability
ensures safety and deadlock-freedom.

Since we are using explicit connection actions, traditional approaches are
insufficiently flexible as they do not account for certain roles being present
in certain branches but not others. Following~\citet{ScalasDHY17} and
subsequently non-classical MPSTs~\cite{ScalasY19}, we instead formulate our
typing rules and safety properties using collections of local types.

It is, however, still convenient to write a global type and have local types
computed programatically. Global types are defined as follows:

{\footnotesize
\[
\begin{array}{lrcl}
  \text{Global Actions} & \globact & ::= & \globalsend{\prole}{\qrole}{\ell}{\tya}
    \midspace
    \globalconn{\prole}{\qrole}{\ell}{\tya}
    \midspace
    \globaldisconn{\prole}{\qrole} \\
  \text{Global Types} & G & ::= & \sumact{i \in I}{\globact_i \then G_i} \midspace \recty{X}{G} \midspace X \midspace \globalend
\end{array}
\]}
Global actions $\globact$ describe interactions between participants:
$\globalsend{\prole}{\qrole}{\ell}{\tya}$ states that role $\prole$ sends a
message with label $\ell$ and payload type $\tya$ to $\qrole$. Similarly,
$\globalconn{\prole}{\qrole}{\ell}{\tya}$ states that $\prole$ connects to
$\qrole$ by sending a message with label $\ell$ and payload type $\tya$. The
disconnection action $\globaldisconn{\prole}{\qrole}$ states that role $\prole$
disconnects from role $\qrole$.

We can write the \texttt{OnlineStore} example from~\autoref{sec:session-types}
as follows:

{\footnotesize
\[
  \bl
  \globalconn{\role{Customer}}{Store}{\var{login}}{\tylit{String}} \then
  \mu \mkwd{Browse} \then \\
  \quad
  {
    \bl
    \quad
    {
      \bl
      \globalsend{\role{Customer}}{\role{Store}}{\var{item}}{\tylit{String}}
        \then \globalsend{\role{Store}}{\role{Customer}}{\var{price}}{\tylit{Int}}
        \then \mkwd{Browse}
      \el
    }
    \\
    +
    \\
    \quad
    {
      \bl
      \globalsend{\role{Customer}}{\role{Store}}{\var{address}}{\tylit{String}}
        \then \globalconn{\role{Store}}{\role{Courier}}{\var{deliver}}{\tylit{String}}
        \then \\
      \globalsend{\role{Courier}}{\role{Store}}{\var{ref}}{\tylit{Int}}
        \then \globaldisconn{\role{Courier}}{\role{Store}} \then \globalsend{\role{Store}}{\role{Customer}}{\var{ref}}{\tylit{Int}}
        \then \\
      \globaldisconn{\role{Store}}{\role{Customer}} \then \globalend
      \el
    }
    \\
    +
    \\
    \quad
    {
      \bl
      \globalsend{\role{Customer}}{\role{Store}}{\var{quit}}{\one} \then
      \globaldisconn{\role{Store}}{\role{Customer}} \then \globalend
      \el
    }
    \el
  }
  \el
\]}Although projectability in our setting does not necessarily guarantee safety and
deadlock-freedom, we show a projection algorithm, adapted from that
of~\citet{Hu2017}, in Appendix~\ref{appendix:projection}. The resulting local
types can then be checked for safety~\secrefp{sec:formalism:preservation}.

\subparagraph{Protocols and Programs.}
Terms do not live in isolation; they refer to a set of \emph{protocols}, and
evaluate in the context of an actor.
A \emph{protocol} maps role names to local session types.

\begin{definition}[Protocol]
  A \emph{protocol} is a set $\{ \prole_i : \loctya_i \}_i$ mapping role names to
  session types.
\end{definition}

As an example, consider the protocol for the online shop example:

{\footnotesize
\[
  \left\{
    \bl
      \role{Customer} : \localtyconn{\role{Store}}{\var{login}}{\tylit{String}}
      \then \mu \mkwd{Browse} \then \\
      {
        \bl
        \qquad
          \localtysend{\role{Store}}{\var{item}}{\tylit{String}} \then
          \localtyrecv{\role{Store}}{\var{price}}{\tylit{Int}} \then
          \mkwd{Browse} \\
        \quad +
        \quad
          \localtysend{\role{Store}}{\var{address}}{\tylit{String}} \then
          \localtyrecv{\role{Store}}{\var{ref}}{\tylit{Int}} \then
          \localtywait{\role{Store}} \then \localend \\
        \quad +
        \quad
          \localtysend{\role{Store}}{\var{quit}}{\one} \then
          \localtywait{\role{Store}} \then \localend,
          \vspace{0.5em}
        \el
      }
\\
\role{Store} :
        \localtyaccept{\role{Customer}}{\var{login}}{\mkwd{String}} \then
        \mu \mkwd{Browse} \then \\
        {
          \bl
          \qquad
          \localtyrecv{\role{Customer}}{\var{item}}{\mkwd{String}} \then
          \localtysend{\role{Customer}}{\var{price}}{\mkwd{Int}} \then
          \mkwd{Browse} \\
\quad + \quad
          \localtyrecv{\role{Customer}}{\var{address}}{\mkwd{String}} \then
          \localtyconn{\role{Courier}}{\var{deliver}}{\mkwd{String}} \then
          \localtyrecv{\role{Courier}}{\var{ref}}{\mkwd{Int}} \then \\
          \qquad
          \localtywait{\role{Courier}} \then
          \localtysend{\role{Customer}}{\var{ref}}{\mkwd{Int}} \then
          \localtydisconn{\role{Customer}} \\
\quad + \quad
          \localtyrecv{\role{Customer}}{\var{quit}}{\one} \then
          \localtydisconn{\role{Customer}},
          \vspace{0.5em}
          \el
      }
\\
\role{Courier} :
      \localtyaccept{\role{Store}}{\var{deliver}}{\tylit{String}} \then
      \localtysend{\role{Store}}{\var{ref}}{\tylit{Int}} \then
      \localtydisconn{\role{Store}}
    \el
  \right\}
\]
}

We can now consider an implementation of a $\role{Store}$ actor, which uses
discovery to find a courier.  We write $\simplerecv{\ell}{x}{\prole}; M$ and
$\simpleaccept{\ell}{x}{\prole}; M$ as
syntactic sugar for
$\newrecv{\prole}{ \msg{\ell}{x} \mapsto M}$ and
$\newaccept{\prole}{ \msg{\ell}{x} \mapsto M}$ respectively, and write
$M; N$ as syntactic sugar for $\efflet{x}{M}{N}$ for a fresh variable $x$.
We assume the existence of a function \mkwd{lookupPrice}, and define
$\mkwd{CourierType}$ as
$
\localtyaccept{\role{Store}}{\var{deliver}}{\tylit{String}} \then
      \localtysend{\role{Store}}{\var{ref}}{\tylit{Int}} \then
      \localtydisconn{\role{Store}}
$.

{\footnotesize
\[
  \bl
  \calcwd{actor} \: \mkwd{Store} \: \calcwd{follows} \: \ty{\role{Store}} \: \{ \\
  {\bl
  \simpleaccept{\var{login}}{\var{credentials}}{\role{Customer}}; \\
  \annstone{\mkwd{Browse}} \\
  \quad
  {
    \bl
    \newrecvone{\role{Customer}} \{ \\
      \quad
        {
          \bl
            \msg{\var{item}}{\var{name}} \mapsto \\
            \quad
            {
              \bl
            \newsend{\var{price}}{\mkwd{lookupPrice}(\var{name})}{\role{Customer}};
            \\
            \continue{\mkwd{Browse}}
              \el
            }
            \\
            \msg{\var{address}}{\var{addr}} \mapsto \\
            \quad
            {
              \bl
              \efflettwo{\var{pid}}{\newdiscover{\mkwd{CourierType}}} \\
              \newconn{\var{deliver}}{\var{addr}}{\var{pid}}{\role{Courier}}; \\
              \simplerecv{\var{ref}}{\var{r}}{\role{Courier}}; \\
              \newwait{\role{Courier}}; \\
              \newsend{\var{ref}}{\var{r}}{\role{Customer}}; \\
              \newdisconn{\role{Customer}}
              \el
            } \\
            \msg{\var{quit}}{()} \mapsto \newdisconn{\role{Customer}} \\
          \el
        } \\
      \}
    \el
  }
  \el
  } \\
\}
  \el
\]
}

A \emph{program} consists of actor definitions, protocol definitions, and the
`boot' clause to be run in order to set up initial actor communication.

\begin{definition}[Program]
  An \EnsembleS \emph{program} is a 3-tuple $(\seq{D}, \seq{P}, M)$ of
  a set of definitions, protocols, and an initial term to be evaluated.
\end{definition}

In the context of a program, we write $\ty{\prole}$ to refer to the session type
associated with role $\prole$ as defined by the set of protocols.
Given an actor definition $\actordef{\actorcls}{\loctya}{M}$,
  we define $\sesstype{\actorcls} = \loctya$ and $\bhv{u} = M$.

\subsection{Typing rules}

\begin{figure}[t]
{\footnotesize
  \begin{minipage}[t]{0.26\textwidth}
  \headersig{Definition typing}{$\dseq{D}$}
  \begin{mathpar}
    \inferrule
    [T-Def]
    {
      \tseq[\loctya]{\cdot}{\loctya}{M}{\tya}{\localend}
    }
    { \dseq{\actordef{\actorcls}{\loctya}{M}} }
  \end{mathpar}
  \end{minipage}
  \hfill
  \begin{minipage}[t]{0.26\textwidth}
  \headersig{Value typing}{$\vseq{\Gamma}{V}{\tya}$}
    \begin{mathpar}
      \inferrule
      [T-Var]
      { x : A \in \Gamma }
      { \vseq{\Gamma}{x}{A} }

      \inferrule
      [T-Unit]
      { }
      { \vseq{\Gamma}{()}{\one} }
    \end{mathpar}
  \end{minipage}
  \hfill
  \begin{minipage}[t]{0.42\textwidth}
\headersig{Behaviour typing}{$\bseq[\loctya]{\Gamma}{\behaviour}$}
    \begin{mathpar}
    \inferrule
    [T-Stop]
    { }
    { \bseq[\loctya]{\Gamma}{\bstop} }

    \inferrule
    [T-Body]
    { \tseq[\loctya]{\Gamma}{\loctya}{M}{A}{\localend} }
    { \bseq[\loctya]{\Gamma}{M} }
  \end{mathpar}
  \end{minipage}

  \headersig{Typing rules for computations}{$\tseq{\Gamma}{\loctya}{M}{A}{\loctya'}$}

\subheader{Functional Rules}
\begin{mathpar}
  \inferrule
    [T-Let]
    {
      \tseq{\Gamma}{\loctya}{M}{\tya}{\loctya'} \\
      \tseq{\Gamma, x : \tya}{\loctya'}{N}{\tyb}{\loctya''}
    }
    { \tseq{\Gamma}{\loctya}{\efflet{x}{M}{N}}{B}{\loctya''} }

  \inferrule
    [T-Return]
    { \vseq{\Gamma}{V}{A} }
    { \tseq{\Gamma}{\loctya}{\effreturn{V}}{\tya}{\loctya} }

  \inferrule
    [T-Rec]
    { \tseq{\Gamma, l : \loctya}{\loctya}{M}{\tya}{\loctya'} }
    {
      \tseq{\Gamma}{\loctya}{\annst{l}{M}}{\tya}{\loctya'}
    }

  \inferrule
    [T-Continue]
    { }
    { \tseq{\Gamma, l : \loctya}{\loctya}{\continue{l}}{\tya}{\loctya'} }
\end{mathpar}

\subheader{Actor / Adaptation Rules}
\begin{mathpar}
  \inferrule
  [T-New]
  { \sesstype{u} = \loctyc  }
  { \tseq{\Gamma}{\loctya}{\new{u}}{\pidty{\loctyc}}{\loctya} }

  \inferrule
  [T-Self]
  { }
  { \tseq{\Gamma}{\loctya}{\self}{\pidty{\loctyb}}{\loctya}  }

  \inferrule
  [T-Discover]
  { }
  {  \tseq{\Gamma}{\loctya}{\newdiscover{\loctyc}}{\pidty{\loctyc}}{\loctya} }

  \inferrule
  [T-Replace]
  { \vseq{\Gamma}{V}{\pidty{\loctyc}} \\ \bseq[\loctyc]{\Gamma}{\behaviour}  }
  { \tseq{\Gamma}{\loctya}{\replace{V}{\behaviour}}{\one}{\loctya} }
\end{mathpar}
}

\caption{Typing rules (1)}
\label{fig:term-typing-1}
\end{figure}

Figures~\ref{fig:term-typing-1} and~\ref{fig:term-typing-2} show the typing
rules for \EnsembleS. Value typing, with judgement $\vseq{\Gamma}{V}{A}$, states
that under environment $\Gamma$, value $V$ has type $A$.
Judgement $\dseq{D}$ states that an actor definition $\actordef{u}{\loctya}{M}$
is well-typed if its body is typable under, and fully consumes, its statically-defined session type
$\loctya$.
The behaviour typing judgement $\bseq[\loctya]{\Gamma}{\behaviour}$ states that
given static session type $\loctya$, behaviour $\behaviour$ is well-typed under
$\Gamma$. Specifically, $\bstop$ is always well-typed, and $M$ is well-typed if
it is typable under and fully consumes $\loctya$.

\subsubsection{Term typing.}\label{sec:formalism:term-typing}
The typing judgement for terms
$\tseq[\loctyb]{\Gamma}{\loctya}{M}{A}{\loctya'}$ reads ``in an actor following $\loctyb$, under typing environment $\Gamma$ and with current session
type $\loctya$, term $M$ has type $A$ and updates the session type to
$\loctya'$''. Note that the term typing judgement, reminiscent of parameterised
monads~\cite{A09:parameterised-monads}, contains a session precondition
$\loctya$ and may perform
some session communication actions to arrive at postcondition $\loctya'$.

\subparagraph{Functional rules.}
Rule \textsc{T-Let} is a sequencing operation: given a construct
$\efflet{x}{M}{N}$ where $M$ has pre-condition $\loctya$ and post-condition
$\loctya'$, and where $N$ has pre-condition $\loctya'$ and post-condition
$\loctya''$, the overall construct has pre-condition $\loctya$ and
post-condition $\loctya''$.

Following Kouzapas \emph{et al.}~\cite{Kouzapas:2016:TPM:2967973.2968595}, we formalise recursion through
annotated
expressions: term $\annst{l}{M}$ states that $M$ is an expression which can loop to
$l$ by evaluating $\continue{l}$.
We take an equi-recursive view of session types, identifying recursive sessions
with their unfolding ($\mu X. S = S \{ \mu X . S / X\}$), and assume that recursion is guarded.
Rule \textsc{T-Rec}
extends the typing environment with a recursion label defined at the current
session type.
Rule \textsc{T-Continue} ensures that the pre-condition must match the label
stored in the environment, but has arbitrary type and any post-condition since
the return type and post-condition depend on the enclosing loop's base case.

\subparagraph{Actor and adaptation rules.}
Rule \textsc{T-New} states that creating an actor of
class $u$ returns a PID parameterised by the session type declared in the class
of $u$. Rule \textsc{T-Self} retrieves a PID for the current actor,
parameterised by the statically-defined session type of the local actor (i.e.,
the $\loctyb$ in the judgement $\tseq[\loctyb]{\Gamma}{\loctya}{M}{A}{\loctya'}$).
Rule \textsc{T-Discover} states $\newdiscover{\loctyc}$ returns a PID of type
$\pidty{\loctyc}$. Finally, given a behaviour $\behaviour$ typable under a
static session type $\loctyc$, and a process ID with the matching static type
$\pidty{\loctyc}$, \textsc{T-Replace} allows replacement, and returns the unit
type.

\begin{figure}[t]

{\footnotesize
\subheader{Exception handling rules}
\vspace{-1em}
\begin{mathpar}
  \inferrule
  [T-Raise]
  { }
  { \tseq{\Gamma}{\loctya}{\raiseexn}{\tya}{\loctya'} }

  \inferrule
  [T-Try]
  {
    \tseq{\Gamma}{\loctya}{L}{A}{\loctya'} \\
    \tseq{\Gamma}{\loctya}{M}{A}{\loctya'}
  }
  { \tseq{\Gamma}{\loctya}{\trycatch{L}{M}}{A}{\loctya'} }
\end{mathpar}

\subheader{Session communication rules}
\begin{mathpar}
  \inferrule
  [T-Conn]
  {
    \localtyconn{\prole_j}{\ell_j}{\tya_j} \in \{\locact_i\}_{i \in I} \\
    \vseq{\Gamma}{V}{\tya_j} \\
    \vseq{\Gamma}{W}{\pidty{\loctyb}} \\
    \loctyb = \ty{\prole_j}
  }
  { \tseq{\Gamma}
      {\sumact{i \in I}{\locact_i \then \loctya_i}}
      {\newconn{\ell_j}{V}{W}{\prole_j}}
      {\one}
      {\loctya'_j}
  }

  \inferrule
  [T-Send]
  { \localtysend{\prole_j}{\ell_j}{\tya_j} \in \{ \locact_i \}_{i \in I} \\
    \vseq{\Gamma}{V}{\tya_j} }
  { \tseq{\Gamma}{\sumact{i \in
  I}{\locact_i \then \loctya_i}}{\newsend{\ell_j}{V}{\prole_j}}{\one}{\locty'_j}   }

  \inferrule
  [T-Accept]
  {
    (\tseq{\Gamma, x_i : \tyb_i}{\locty_i}{M_i}{\tya}{S})_{i \in I}
  }
  { \tseq{\Gamma}
         {\sumact{i \in I}{\localtyaccept{\qrole}{\ell_i}{\tyb_i} \then \locty_i}}
         {\newaccept{\qrole}{\ell_i(x_i) \mapsto M_i}_{i \in I}}
         {\tya}
         {\locty}
  }

  \inferrule
  [T-Recv]
  {
    (\tseq{\Gamma, x_i : \tyb_i}{\locty_i}{M_i}{A}{S})_{i \in I}
  }
  { \tseq{\Gamma}
         {\sumact{i \in I}{\localtyrecv{\qrole}{\ell_i}{\tyb_i} \then \locty_i}}
         {\newrecv{\qrole}{\ell_i(x_i) \mapsto M_i}_{i \in I}}
         {A}
         {\locty}
  }

  \inferrule
  [T-Wait]
  { }
  { \tseq{\Gamma}
         {\localtywait{\qrole} \then \loctya}
         {\newwait{\qrole}}
         {\one}
         {\locty}
  }

  \inferrule
  [T-Disconn]
  { }
  { \tseq{\Gamma}
    {\localtydisconn{\qrole}}
         {\newdisconn{\qrole}}
         {\one}
         {\localend} }
\end{mathpar}
}

  \caption{Typing rules (2)}
  \label{fig:term-typing-2}
\end{figure}

\subparagraph{Exception handling rules.}
Figure~\ref{fig:term-typing-2} shows the rules for exception handling and
session communication. \textsc{T-Raise} denotes raising an exception; since
it does not return, it can have an arbitrary return type and postcondition. Rule
\textsc{T-Try} types an exception handler $\trycatch{L}{M}$ which acts over a
single action $L$. If $L$ raises an exception, then $M$ is evaluated instead.
Since $L$ only scopes over a single action, the $\calcwd{try}$ and $\calcwd{catch}$
clauses have the same pre- and post-conditions to allow the action to be retried
if necessary.

\begin{remark}
    Following~\citet{MostrousV18}, our $\trycatch{L}{M}$ construct scopes over a
    \emph{single} action and is discarded afterwards.
We opt for this simple approach since in our setting exceptions are a means
    to an end, but (at the coast of a more involved type system) we could
    potentially scope over multiple actions as long as the handler is compatible
    with all potential exit conditions~\cite{Gordon20}. We leave a
    thorough exploration to future work.
\end{remark}

\subparagraph{Session communication rules.}
Rule \textsc{T-Conn} types a term $\newconn{\ell_j}{V}{W}{\prole_j}$.
Given the precondition is a choice type
containing a branch $\localtyconn{\prole}{\ell_j}{\tya_j} \then \loctya'_j$,
and the remote actor reference is $W$ of type $\pidty{S}$,
the rule ensures that $S$ is compatible with the type of $\prole_j$,
and ensures that the label and payload are compatible with the session type.
The session type is then advanced to $\loctya'_j$.
Rule \textsc{T-Send} follows the same pattern.

Given a session type $\sumact{i \in I}{\localtyaccept{\prole}{x_i}{A_i}} \then
S_i$, rule \textsc{T-Accept} types term $\newaccept{\prole}{\msg{\ell_i}{x_i} \mapsto M_i}_{i
\in I}$, enabling an actor to accept connections with messages
$\ell_i$, binding the payload $x_i$ in each continuation $M_i$.
Like $\calcwd{case}$ expressions in functional languages, each continuation must
be typable under an environment extended with $x_i : A_i$, under session type
$\loctya_i$, and each branch must have same result type and postcondition.  Rule
\textsc{T-Recv} is similar.

Rule \textsc{T-Wait} handles waiting for a participant $\prole$ to disconnect
from a session, requiring a pre-condition of $\localtywait{\prole} \then \loctya$,
returning the unit type and advancing the session type to $\loctya$.
Rule \textsc{T-Disconnect} is similar and advances the session type to $\localend$.

\subsection{Operational semantics}

We describe the semantics of \EnsembleS via a deterministic reduction relation
on terms, and a nondeterministic reduction relation on configurations.

\begin{figure}[t]
  {\footnotesize \header{Runtime syntax} \vspace{-2em}}
  {\footnotesize
    \[
  \begin{array}{lccl}
\text{Names} & n & ::= & a \midspace s \vspace{0.3em} \\
\text{Configurations} & \config{C}, \config{D}, \config{E} & ::= & (\nu n)
  \config{C} \midspace \config{C} \parallel \config{D} \midspace
  \eactor{a}{M}{\connstate}{\behaviour} \midspace \zap{\sessindextwo{s}{\prole}}
  \midspace \confzero
  \\
    \text{Connection state} & \connstate & ::= & \bot \midspace
    \newconnsess{s}{\prole}{\rolesetq} \vspace{0.3em} \\
  \text{Runtime environments} \!\! & \rtenv & ::= & \cdot \midspace \rtenv, a : \loctya \midspace
  \rtenv, \sessindexroles{s}{\prole}{\rolesetq}{\loctya} \vspace{0.3em} \\
  \text{Evaluation contexts} & E & ::= & F \midspace \efflet{x}{E}{M} \\
  \text{Top-level contexts} & F & ::= & [~] \midspace \trycatch{[~]}{M} \\
  \text{Pure contexts} & \ep & ::= & [~] \midspace \efflet{x}{\ep}{M} \\
  \end{array}
\]
}

{\footnotesize
  \headersig{Term reduction}{$M \teval N$}
  \[
    \begin{array}{lrcl}
      \textsc{E-Let} & \efflet{x}{\effreturn{V}}{M} & \teval & M \{ V / x \} \\
      \textsc{E-TryReturn} & \trycatch{\effreturn{V}}{M} & \teval & \effreturn{V} \\
      \textsc{E-TryRaise} & \trycatch{\raiseexn}{M} & \teval
                          & M \\
      \textsc{E-Rec} & \annst{l}{M} & \teval & M \{ \annst{l}{M} / \continue{l} \} \\
      \textsc{E-LiftM} & E[M] & \teval & E[N] \quad \text{if } M \teval N
    \end{array}
  \]
}
  {\footnotesize
  \headersig{Configuration reduction (1)}{$\config{C} \ceval \config{D}$}
  \subheader{Actor / adaptation rules}
  \begin{mathpar}
    \inferrule
    [E-Loop]
    { }
    { \eactor{a}{\effreturn{V}}{\bot}{M} \ceval \eactor{a}{M}{\bot}{M} }

    \inferrule
    [E-New]
    { b \text{ is fresh} \\ \bhv{u} = M }
    {
      \eactor{a}{E[\new{u}]}{\connstate}{\behaviour} \ceval \\\\
      (\nu b) (\eactor{a}{E[\effreturn{b}]}{\connstate}{\behaviour} \parallel
     \eactor{b}{M}{\bot}{M})
    }

    \inferrule
    [E-Replace]
    { }
    { \eactor{a}{E[\replace{b}{\behaviour'}]}{\connstate_1}{\behaviour_1} \parallel
      \eactor{b}{N}{\connstate_2}{\behaviour_2} \ceval \\\\
      \eactor{a}{E[\effreturn{()}]}{\connstate_1}{\behaviour_1} \parallel
      \eactor{b}{N}{\connstate_2}{\behaviour'} }

    \inferrule
    [E-ReplaceSelf]
    { }
    { \eactor{a}{E[\replace{a}{\behaviour'}]}{\connstate}{\behaviour}
      \ceval \\\\
      \eactor{a}{E[\effreturn{()}]}{\connstate}{\behaviour'}}

    \inferrule
    [E-Discover]
    { \sesstype{b} = \loctya \\\\
      \neg ((N = \effreturn{V} \vee N = \raiseexn) \wedge \behaviour_2 = \bstop)  }
    { \eactor{a}{E[\newdiscover{\loctya}]}{\connstate_1}{\behaviour_1} \parallel
        \eactor{b}{E'[N]}{\connstate_2}{\behaviour_2} \ceval\\\\
      \eactor{a}{E[\effreturn{b}]}{\connstate_1}{\behaviour_1} \parallel
      \eactor{b}{E'[N]}{\connstate_2}{\behaviour_2}
    }

    \inferrule
    [E-Self]
    { }
    { \eactor{a}{E[\self]}{\connstate}{\behaviour}
        \ceval
      \eactor{a}{E[\effreturn{a}]}{\connstate}{\behaviour} }
  \end{mathpar}
  }

  \caption{Operational semantics (1)}
  \label{fig:config-reduction-1}
  \end{figure}

\subsubsection{Runtime syntax}

Figure~\ref{fig:config-reduction-1} shows the runtime syntax and the first part
of the reduction rules for \EnsembleS.

Whereas static syntax and typing rules describe code that a user would write,
runtime syntax arises during evaluation. We introduce two types of runtime name:
$s$ ranges over \emph{session names}, which are created when a process initiates
a session, and $a$ ranges over \emph{actor names}, which uniquely identify each
actor once it has been spawned by $\calcwd{new}$.

\subparagraph{Configurations.}
Configurations, ranged over by $\config{C}, \config{D}, \config{E}$, represent
the concurrent fragment of the language. Like in the
$\pi$-calculus~\cite{Milner99:picalc},
name restrictions $(\nu n)\config{C}$ bind name $n$ in $\config{C}$,
$\config{C} \parallel \config{D}$ denotes $\config{C}$ and $\config{D}$ running
in parallel, and the $\confzero$ configuration denotes the inactive process.

Actors are represented at runtime as a 4-tuple
$\eactor{a}{M}{\connstate}{\behaviour}$, where $a$ is the actor's runtime name;
$M$ is the term currently evaluating; $\connstate$ is the connection state; and
$\behaviour$ is the actor's current behaviour.
A connection state is either \emph{disconnected}, written $\bot$, or playing role
$\prole$ in session $s$ and connected to roles $\rolesetq$, written
$\newconnsess{s}{\prole}{\rolesetq}$.

Inspired by~\citet{MostrousV18} and~\citet{FowlerLMD19:stwt}, a
\emph{zapper thread} $\zap{\sessindextwo{s}{\prole}}$ indicates that participant
$\prole$ in session $s$ cannot be used for future communications, for example
due to the actor playing the role crashing due to an unhandled exception.

\subparagraph{Runtime typing environments.}
Whereas $\Gamma$ is an unrestricted typing environment used for typing values
and configurations, we introduce $\rtenv$ as a linear runtime environment.
Runtime environments can contain entries of type $a : S$, stating that actor $a$
has \emph{statically-defined} session type $S$, and entries of type
$\sessindexroles{s}{\prole}{\rolesetq}{S}$, stating that in session $s$, role
$\prole$ is connected to roles $\rolesetq$ and \emph{currently} has session type
$S$.

\subparagraph{Evaluation contexts.}
Due to our fine-grain call-by-value presentation, evaluation contexts $E$ allow
nesting only in the immediate subterm of a $\calcwd{let}$ expression.
The top-level frame $F$ can either be a hole, or a single, top-level exception
handler. Pure contexts $\ep$ do not include exception handling frames.

To run a program, we place it in an \emph{initial configuration}.

\begin{definition}[Initial configuration]
  An \emph{initial configuration} for an \EnsembleS program with boot clause $M$
  is of the form $(\nu a)(\eactor{a}{M}{\bot}{\bstop})$.
\end{definition}

\subsubsection{Reduction rules}

  \begin{figure}[t]
  \footnotesize{
  \headersig{Configuration reduction (2)}{$\config{C} \ceval \config{D}$}
  \subheader{Session reduction rules}
  \begin{mathpar}
\inferrule
    [E-ConnInit]
    { j \in I }
    { \eactor{a}{E[F[\newconn{\ell_j}{V}{b}{\qrole}]]}{\bot}{\behaviour_1} \parallel
      \eactor{b}{E'[F'[\newaccept{\prole}{\ell_i(x_i) \mapsto M_i}_{i \in
      I}]]}{\bot}{\behaviour_2} \ceval \\
(\nu s)
      (\eactor{a}{E[\effreturn{()}]}{\newconnsess{s}{\prole}{\qrole}}{\behaviour_1} \parallel
      \eactor{b}{E'[M_j \{ V / x_j \}]}{\newconnsess{s}{\qrole}{\prole}}{\behaviour_2})
    }

\inferrule
    [E-Conn]
    { \qrole \not\in \rolesetr }
    {
      \eactor{a}{E[F[\newconn{\ell_j}{V}{b}{\qrole}]]}{\newconnsess{s}{\prole}{\rolesetr}}{\behaviour_1} \parallel
        \eactor{b}{E'[F'[\newaccept{\prole}{\ell_i(x_i) \mapsto N_i}_{i \in
        I}]]}{\bot}{\behaviour_2} \ceval \\
\eactor{a}{E[\effreturn{()}]}{\newconnsess{s}{\prole}{\rolesetr, \qrole}}{\behaviour_1} \parallel
        \eactor{b}{E'[N_j \{ V / x_j \}]}{\newconnsess{s}{\qrole}{\prole}}{\behaviour_2}
    }

    \inferrule
    [E-ConnFail]
    {
      ((N = \effreturn{V} \vee N = \ep[\raiseexn]) \wedge \behaviour_2 = \bstop) \vee
      \connstate_2 \ne \bot
    }
    {
      \eactor{a}{E[\newconn{\ell_j}{V}{b}{\qrole}]}{\connstate_1}{\behaviour_1} \parallel
        \eactor{b}{N}{\connstate_2}{\behaviour_2} \ceval
\eactor{a}{E[\raiseexn]}{\connstate_1}{\behaviour_1} \parallel
        \eactor{b}{N}{\connstate_2}{\behaviour_2}
    }

    \inferrule
    [E-Disconn]
    { }
    { \eactor{a}{E[F[\newwait{\qrole}]]}{\newconnsess{s}{\prole}{\rolesetr, \qrole}}{\behaviour_1}
      \parallel
      \eactor{b}{E'[F'[\newdisconn{\prole}]]}{\newconnsess{s}{\qrole}{\prole}}{\behaviour_2}
      \ceval \\
      \eactor{a}{E[\effreturn{()}]}{\newconnsess{s}{\prole}{\rolesetr}}{\behaviour_1} \parallel
        \eactor{b}{E'[\effreturn{()}]}{\bot}{\behaviour_2}
    }

    \inferrule
    [E-Comm]
    { j \in I \\ \qrole \in \rolesetr \\ \prole \in \rolesets }
    {
      \eactor{a}{E[F[\newsend{\ell_j}{V}{\qrole}]]}{\newconnsess{s}{\prole}{\rolesetr}}{\behaviour_1} \parallel
        \eactor{b}{E'[F'[\newrecv{\prole}{\ell_i(x_i) \mapsto M_i}_{i \in
        I}]]}{\newconnsess{s}{\qrole}{\rolesets}}{\behaviour_2} \ceval \\
\eactor{a}{E[\effreturn{()}]}{\newconnsess{s}{\prole}{\rolesetr}}{\behaviour_1} \parallel
        \eactor{b}{E'[M_j \{ V / x_j \}]}{\newconnsess{s}{\qrole}{\rolesets}}{\behaviour_2}
    }

    \inferrule
    [E-Complete]
    { }
    {
      (\nu s)(\eactor{a}{\effreturn{V}}{\newconnsess{s}{\prole}{\emptyset}}{\behaviour})
      \ceval
      \eactor{a}{\effreturn{V}}{\bot}{\behaviour}
    }
  \end{mathpar}
}
  \caption{Operational semantics (2)}
  \label{fig:config-reduction-2}
\end{figure}

\begin{figure}[t]
  {\footnotesize
  \headersig{Configuration reduction (3)}{$\config{C} \ceval \config{D}$}
  \subheader{Exception handling rules}
  \begin{mathpar}
    \inferrule
    [E-CommRaise]
    { \subj{M} = \qrole }
    {
      \eactor{a}{E[M]}{\newconnsess{s}{\prole}{\rolesetr}}{\behaviour} \parallel
      \zap{\sessindextwo{s}{\qrole}}
        \ceval \\\\
      \eactor{a}{E[\raiseexn]}{\newconnsess{s}{\prole}{\rolesetr}}{\behaviour}
      \parallel \zap{\sessindextwo{s}{\qrole}}
    }

\inferrule
    [E-FailS]
    { }
    { \eactor{a}{\ep[\raiseexn]}{\newconnsess{s}{\prole}{\rolesetr}}{\behaviour} \ceval \\\\
      \eactor{a}{\raiseexn}{\bot}{\behaviour} \parallel \zap{\sessindextwo{s}{\prole}} }

    \inferrule
    [E-FailLoop]
    { }
    { \eactor{a}{\ep[\raiseexn]}{\bot}{M} \ceval\\\\ \eactor{a}{M}{\bot}{M} }
  \end{mathpar}

  \subheader{Administrative rules}
\begin{mathpar}
  \inferrule
   [E-LiftM]
   { M \teval M' }
   {
     \eactor{a}{E[M]}{\connstate}{\behaviour} \ceval
     \eactor{a}{E[M']}{\connstate}{\behaviour}
   }

  \inferrule
  [E-Equiv]
  { \config{C} \equiv \config{C}' \\
    \config{C}' \ceval \config{D}' \\\\
    \config{D}' \equiv \config{D} }
  { \config{C} \ceval \config{D} }

   \inferrule
   [E-Par]
   { \config{C} \ceval \config{C}' }
   { \config{C} \parallel \config{D} \ceval \config{C}' \parallel \config{D} }

   \inferrule
   [E-Nu]
   { \config{C} \ceval \config{D} }
   { (\nu n) \config{C} \ceval (\nu n) \config{D}}
 \end{mathpar}

 \headersig{Configuration equivalence}{$\config{C} \equiv \config{D}$}
\begin{mathpar}
  \config{C} \parallel \config{D} \equiv \config{D} \parallel \config{C}

  \config{C} \parallel (\config{D} \parallel \config{E}) \equiv (\config{C} \parallel \config{D}) \parallel \config{E}

  (\nu n_1)(\nu n_2)\config{C} \equiv (\nu n_2)(\nu n_1) \config{C}

  \config{C} \parallel (\nu n)\config{D} \equiv (\nu n)(\config{C} \parallel
  \config{D}) \quad \text{if } n \not\in \fn{\config{C}}

  (\nu s)(\zaptwo{s}{\prole_1} \parallel \cdots \parallel \zaptwo{s}{\prole_n}) \parallel
  \config{C} \equiv \config{C}

  \config{C} \parallel \confzero \equiv \config{C}
\end{mathpar}
 }

 \caption{Operational semantics (3)}
  \label{fig:config-reduction-3}
\end{figure}

Term reduction $\teval$ is standard $\beta$-reduction, save for
\textsc{E-TryRaise} which evaluates the failure continuation in the case of an
exception.
We consider four
subcategories of configuration reduction rules: actor and adaptation rules; session
communication rules; exception handling rules; and administrative rules.

\subparagraph{Actor / adaptation rules.}
Given a
fully-evaluated actor, \textsc{E-Loop} runs the term specified by the actor's
behaviour.
Rule \textsc{E-New} allows actor $a$ to spawn a new actor of class $\actorcls$
by creating a fresh runtime actor name $b$ and a new actor process of the form
$\eactor{b}{M}{\bot}{M}$ where $M$ is the behaviour specified by $\actorcls$,
returning the process ID $b$.
Rules \textsc{E-Replace} and \textsc{E-ReplaceSelf} handle replacement by
changing the behaviour of an actor, returning the unit value to the caller.
Rule \textsc{E-Discover} returns the process ID of an actor $b$ if it has the
desired static session type $S$.
Rule \textsc{E-Self} returns the PID of the local actor.

\subparagraph{Session communication rules.}
An actor begins a session by connecting to another actor while disconnected;
such a case is handled by rule \textsc{E-ConnInit}. Suppose we have a
disconnected actor $a$ evaluating a connection statement
$\newconn{\ell_j}{V}{b}{\prole}$, evaluating in parallel with a disconnected
actor $b$ evaluating an accept statement $\newaccept{\prole}{\ell_i(x_i) \mapsto
M_i}_{i \in I}$. Rule \textsc{E-ConnInit} returns the unit value to actor $a$;
creates a fresh session name restriction $s$, sets the connection
state of $a$ to $\newconnsess{s}{\prole}{\qrole}$ and of $b$ to
$\newconnsess{s}{\qrole}{\prole}$;
accepting actor $b$ then evaluates continuation $M_j$ with $V$
substituted for $x_j$.
Since exception handlers only scope over a single communication action, the
top-level frames $F, F'$ in each actor are discarded if the
communication succeeds.
Rule \textsc{E-Conn} handles the case where the connecting actor is already part
of a session and behaves similarly to \textsc{E-ConnInit}, without creating a
new session name restriction.
A connection can fail if an actor attempts to connect to another actor which is
terminated or is
already involved in a session; in these cases, \textsc{E-ConnFail} raises an
exception in the connecting actor.

Rule \textsc{E-Disconn} handles the case where an actor $b$ leaves a session,
synchronising with an actor $a$. In this case, the unit value is returned to
both callers, and the connection state of $b$ is set to $\bot$.
Rule \textsc{E-Comm} handles session communication when two participants are
already connected to the same session, and is similar to \textsc{E-Conn}.
Rule \textsc{E-Complete} garbage collects a session after it has completed and
sets the initiator's connection state to $\bot$.

\subparagraph{Exception handling rules.}
Exception handling rules allow safe session communication in
the presence of exceptions. Rule \textsc{E-CommRaise} states that if an actor is
attempting to communicate with a role no longer present due to an
exception, then an exception should be raised. We write $\subj{E[M]} = \prole$ if
$M \in \{ \newsend{\ell}{V}{\prole}, \newrecv{\prole}{\ell_i(x_i)\mapsto N_i}_i,
  \newwait{\prole}, \newdisconn{\prole}\}$.
Rule \textsc{E-FailS} states that if a connected actor encounters an unhandled
exception, then a zapper thread will be generated for the current role, the
actor will become disconnected, and the current evaluation context will be
discarded. Rule \textsc{E-FailLoop} restarts an actor encountering an
unhandled exception.

\subparagraph{Administrative rules.}
The remaining rules are administrative: \textsc{E-LiftM} allows term reduction
inside an actor; \textsc{E-Equiv} allows reduction modulo structural congruence;
\textsc{E-Par} allows reduction under parallel composition; and \textsc{E-Nu}
allows reduction under name restrictions.

\subparagraph{Configuration equivalence.}
Reduction includes configuration equivalence $\equiv$, defined as the smallest
congruence relation satisfying the axioms in
Figure~\ref{fig:config-reduction-3}.
The equivalence rules extend the usual $\pi$-calculus structural congruence
rules with a `garbage collection' equivalence, which allows us to
discard a session where all participants have exited due to an error.

\subsection{Metatheory}
We now turn our attention to showing that session typing allows runtime
adaptation and discovery while precluding communication mismatches and
deadlocks.

\subsubsection{Runtime typing}
To reason about the metatheory, we introduce typing rules for
configurations~(\autoref{fig:rt-typing}): the judgement
$\cseq{\Gamma}{\rtenv}{\config{C}}$ states that configuration $\config{C}$ is
well-typed under term typing environment $\Gamma$ and runtime typing environment
$\rtenv$.

\begin{figure}[t]
{\footnotesize
    \headersigarg
    {Runtime Typing Rules}
    {\framebox{$\vphantom{\rtenv;\config{C}}\vseq{\Gamma}{V}{A}$}\:\framebox{$\cseq{\Gamma}{\rtenv}{\config{C}}$}}
    \vspace{-1em}
  \begin{mathpar}
    \inferrule
    [T-Pid]
    { \cseq{\Gamma, a : \pidty{\loctya}}{\rtenv, a : \loctya}{\config{C}} }
    { \cseq{\Gamma}{\rtenv}{(\nu a) \config{C}} }

    \inferrule
    [T-Session]
    {
      \rtenv' = \{\sessindexroles{s}{\prole_i}{\rolesetq_i}{\locty_{\prole_i}}\}_{i \in I} \\
      \proparg{\rtenv'} \\
      s \not\in \rtenv \\
      \cseq{\Gamma}{\rtenv, \rtenv'}{\config{C}} \\\\
      \prop \text{ is a safety property}
    }
    { \cseq{\Gamma}{\rtenv}{(\nu s) \config{C}} }

    \inferrule
    [T-Par]
    { \cseq{\Gamma}{\rtenv_1}{\config{C}} \\ \cseq{\Gamma}{\rtenv_2}{\config{D}} }
    { \cseq{\Gamma}{\rtenv_1, \rtenv_2}{\config{C} \parallel \config{D} } }

    \inferrule
    [T-Zap]
    { }
    { \cseq{\Gamma}{\sessindexroles{s}{\prole}{\rolesetq}{S}}{\zap{\sessindextwo{s}{\prole}}} }

    \inferrule
    [T-Zero]
    { }
    { \cseq{\Gamma}{\cdot}{\confzero} }

    \inferrule
    [T-DisconnectedActor]
    {
\loctyb = \loctya \vee \loctyb = \localend \\
      a : \pidty{\loctya} \in \Gamma \\\\
      \tseq[\loctya]{\Gamma}{\loctyb}{M}{A}{\localend} \\
      \bseq[\loctya]{\Gamma}{\behaviour}
    }
    { \cseq{\Gamma}{a : \loctya}{\eactor{a}{M}{\bot}{\behaviour}} }

    \inferrule
    [T-ConnectedActor]
    {
      a : \pidty{\loctyb} \in \Gamma \\\\
      \tseq{\Gamma}{\loctya}{M}{A}{\localend} \\
      \bseq{\Gamma}{\behaviour}
    }
    { \cseq{\Gamma}{a : \loctyb, \sessindexroles{s}{\prole}{\rolesetq}{\locty}}{
      \eactor{a}{M}{\newconnsess{s}{\prole}{\rolesetq}}{\behaviour}}
    }
  \end{mathpar}
}
  \caption{Runtime typing rules}
  \label{fig:rt-typing}
\end{figure}

Rule \textsc{T-Pid} types actor name restriction $(\nu a)\config{C}$
by adding a PID into the term
environment, and extending the runtime typing environment $a:S$; the linearity
of the runtime typing environment therefore means that the system must contain precisely
one actor with name $a$.

Session name restrictions $(\nu s)\config{C}$ are typed by
\textsc{T-Session}.
We follow the formulation of~\citet{ScalasY19} which types multiparty
sessions using a parametric safety property $\varphi$; we discuss safety
properties in more depth in Section~\ref{sec:formalism:preservation}.
Let $\rtenv'$ be a
runtime typing environment containing only mappings of the form
$\sessindexroles{s}{\prole_i}{\rolesetq_i}{S_i}$. Assuming $\rtenv$ does not contain any mappings
involving session $s$ and $\rtenv'$ satisfies $\varphi$, the rule states that $\config{C}$ is typable under typing environment $\Gamma$ and runtime typing environment $\rtenv, \rtenv'$.
It is sometimes convenient to annotate session $\nu$-binders with their
environment, e.g., $(\nu s : \rtenv') \config{C}$.

Rule \textsc{T-Par} types each subconfiguration of a parallel composition by
splitting the linear runtime environment. Rule \textsc{T-Zap} types a zapper
thread $\zap{\sessindextwo{s}{\prole}}$, assuming the runtime environment
contains an entry $\sessindexroles{s}{\prole}{\rolesetq}{\loctya}$ for any session type
$\loctya$.

Finally, rules \textsc{T-DisconnectedActor} and \textsc{T-ConnectedActor} type
disconnected and connected actor configurations respectively.
Given an actor with name $a$ and static session type $\loctyb$, both rules require that
the typing environment contains $a : \pidty{\loctyb}$ and
runtime environment contains $a : \loctyb$.
Both rules require that the current session type is fully consumed by the
currently-evaluating term and that the actor's behaviour should be typable under
$\loctyb$.
Rule \textsc{T-DisconnectedActor} requires that the currently-evaluating term
must be typable under either $\loctyb$ or $\localend$, whereas to type a
connection state of $\newconnsess{s}{\prole}{\rolesetq}$ and current session type
$\loctya$, \textsc{T-ConnectedActor} requires an entry
$\sessindexroles{s}{\prole}{\rolesetq}{\loctya}$ in the runtime environment.

\subsubsection{Preservation}\label{sec:formalism:preservation}
We now prove that reduction preserves typability and thus that actors only
perform communication actions specified in their session types.  Due to our use
of explicit connection actions, classical MPST approaches are too limited for
our purposes.  Our approach, following that of~\citet{ScalasY19}, is to
introduce a labelled transition system (LTS) on local types, and specify a
generic safety property based around local type reduction. The property can then
refined; in our case, we will later specialise the property in order to prove
progress.

\begin{figure}[t]
{\footnotesize
\header{Labels}
\[
  \begin{array}{lrcl}
    \text{Labels} & \ltslbl & ::= & \ltscomm{s}{\prole}{\qrole}{\ell}{A} \midspace
      \ltsconnpair{s}{\prole}{\qrole}{\ell} \midspace
      \ltsdisconnact{s}{\prole}{\qrole} \\
    \text{Synchronisation labels} & \synclbl & ::= &
      \ltspair{s}{\prole}{\qrole}{\ell} \midspace
      \ltsconnpair{s}{\prole}{\qrole}{\ell} \midspace
      \ltsdisconnpair{s}{\prole}{\qrole}
  \end{array}
\]

  \header{Reduction on runtime typing environments}

  \subheadersig{Local Reduction}{$\Delta \annarrow{\ltslbl} \Delta'$}
\begin{mathpar}
  \inferrule
  [ET-Conn]
  { \exists j \in I . \locact_j = \localtyconn{\qrole}{\ell_j}{\tya_j} \\\\
    \ty{\qrole} = \sumact{k \in K}{\localtyaccept{\prole}{\ell_k}{\tyb_k} \then \loctyb_k} \\
    j \in K \\
    \tya_j = \tyb_j
  }
  {
    \sessindexroles{s}{\prole}{\rolesetr}{\sumact{i \in I}{\locact_i \then \loctya_i}}
    \annarrow{\ltsconnpair{s}{\prole}{\qrole}{\ell_j}}
    \sessindexroles{s}{\prole}{\rolesetr, \qrole}{\locty_j},
    \sessindexroles{s}{\qrole}{\prole}{\loctyb_j}
  }

  \inferrule
  [ET-Act]
  { \exists j \in I . \locact_j = \localcomm{\qrole}{\ell_j}{\tya_j} \\
    \qrole \in \rolesetr
  }
  { \sessindexroles{s}{\prole}{\rolesetr}{\sumact{i \in I}{\locact_i \then \locty_i}}
      \annarrow{\ltscomm{s}{\prole}{\qrole}{\ell_j}{\tya_j}}
      \sessindexroles{s}{\prole}{\rolesetr}{\locty_j}
  }

  \inferrule
  [ET-Wait]
  { }
  { \sessindexroles{s}{\prole}{\rolesetr, \qrole}{\localtywait{\qrole} \then S} \annarrow{\ltswait{s}{\prole}{\qrole}}
    \sessindexroles{s}{\prole}{\rolesetr}{S}
  }

  \inferrule
  [ET-Disconn]
  { }
  { \sessindexroles{s}{\prole}{\qrole}{\localtydisconn{\qrole}}
      \annarrow{\ltsdisconn{s}{\prole}{\qrole}} \cdot } \\

  \inferrule
    [ET-Rec]
    { \rtenv, \sessindexroles{s}{\prole}{\rolesetq}{\locty \{ \recty{X}{\locty} / X \}} \annarrow{\ltslbl} \rtenv' }
    { \rtenv, \sessindexroles{s}{\prole}{\rolesetq}{\recty{X}{\locty}} \annarrow{\ltslbl} \rtenv' }

  \inferrule
  [ET-Cong1]
  { \rtenv \annarrow{\ltslbl} \rtenv' }
  { \rtenv, \sessindexroles{s}{\prole}{\rolesetq}{\locty} \annarrow{\ltslbl} \rtenv', \sessindexroles{s}{\prole}{\rolesetq}{\locty} }

  \inferrule
  [ET-Cong2]
  { \rtenv \annarrow{\ltslbl} \rtenv' }
  { \rtenv, a {:} \loctya \annarrow{\ltslbl} \rtenv', a {:} \loctya }

\end{mathpar}

  \subheadersig{Synchronisation}{$\Delta \syncannarrow{\synclbl} \Delta'$}
\begin{mathpar}
  \inferrule
  [ET-ConnSync]
  {
    \rtenv \annarrow{\ltsconnpair{s}{\prole}{\qrole}{\ell}} \rtenv'
  }
  { \rtenv \syncannarrow{\ltsconnpair{s}{\prole}{\qrole}{\ell}} \rtenv' }

  \inferrule
  [ET-Comm]
  { \rtenv_1 \annarrow{\ltscomm[{!}]{s}{\prole}{\qrole}{\ell}{\tya}} \rtenv'_1
    \\
  \rtenv_2 \annarrow{\ltscomm[{?}]{s}{\qrole}{\prole}{\ell}{\tya}} \rtenv'_2 }
  { \rtenv_1, \rtenv_2 \syncannarrow{\ltspair{s}{\prole}{\qrole}{\ell}} \rtenv'_1, \rtenv'_2 }

  \inferrule
  [ET-Disconn]
  { \rtenv_1 \annarrow{\ltswait{s}{\prole}{\qrole}} \rtenv'_1 \\
    \rtenv_2 \annarrow{\ltsdisconn{s}{\qrole}{\prole}} \rtenv'_2 }
  { \rtenv_1, \rtenv_2 \syncannarrow{\ltsdisconnpair{s}{\prole}{\qrole}} \rtenv'_1, \rtenv'_2 }
\end{mathpar}
}
\caption{Labelled transition system for runtime typing environments}
\label{fig:rt-lts-reduction}
\end{figure}

\subparagraph{Reduction on runtime typing environments.}
Figure~\ref{fig:rt-lts-reduction} shows the LTS on runtime typing environments.
There are two judgements: $\rtenv \annarrow{\ltslbl} \rtenv'$, which handles
reduction of individual local types, and a \emph{synchronisation} judgement
$\rtenv \syncannarrow{\synclbl} \rtenv'$.

Rule \textsc{ET-Conn} handles the reduction of role $\prole$,
where the choice session type contains a connection action
$\localtyconn{\qrole}{\ell_j}{A_j}\then\loctya'_j$.
If $\qrole$ has a statically-defined session type $\sumact{k \in
K}{\localtyaccept{\prole}{\ell_k}{B_k} \then \loctyb_k}$
which can accept $\ell_j$ from from $\prole$, and the
payload types match, reduction
advances $\prole$'s session type,
adds $\qrole$ to $\prole$'s connected role set,
and introduces an entry for $\qrole$ into the
environment. The reduction emits a label
$\ltsconnpair{s}{\prole}{\qrole}{\ell_j}$.

Given a role $\prole$ connected to $\qrole$ with a session choice containing a
send or receive action $\localcomm{\qrole}{\ell_j}{A} \then \loctya'_j$,
rule \textsc{ET-Act} will
emit a label $\ltscomm{s}{\prole}{\qrole}{\ell_j}{A_j}$ and advance the session
type of $\prole$.

Rule \textsc{ET-Wait} handles the reduction of $\localtywait{\qrole}\then\loctya$ actions, $\sessindexroles{s}{\prole}{\rolesetr,
\qrole}{\localtywait{\qrole}\then\loctya}$, where $\prole$ waits for $\qrole$ to
disconnect: the reduction emits label $\ltsdisconnact[\waitact]{\prole}{\qrole}$
and removes $\qrole$ from $\prole$'s connected roles. Similarly,
rule \textsc{ET-Disconn} handles disconnection, by emitting label
$\ltsdisconnact[\disconnact]{\prole}{\qrole}$ and removing the entry from the
environment.
\textsc{ET-Rec} handles recursive types, and the \textsc{ET-Cong} rules
handle reduction of sub-environments.

Rule \textsc{ET-ConnSync} states that connection is a synchronisation action,
and rules \textsc{ET-Comm} and \textsc{ET-Disconn} handle synchronisation
between dual actions in sub-environments, emitting synchronisation labels
$\ltspair{s}{\prole}{\qrole}{\ell}$ and $\ltsdisconnpair{s}{\prole}{\qrole}$
respectively. We omit the congruence rules for synchronisation actions.
We say that a runtime environment \emph{reduces}, written $\rtenv \synceval$, if
there exists some $\rtenv'$ such that $\rtenv \synceval \rtenv'$.

\subparagraph{Safety.}
A \emph{safety property} describes a set of invariants on typing environments
which allow us to prove preservation. Since the type system is parametric
in the given safety property, we can tweak the property to permit or rule out
different typing environments satisfying particular behavioural properties;
however, we need only prove type preservation once, using the weakest safety
property. Our safety property is different to the safety property described
by~\citet{ScalasY19} in order to account for explicit connection actions.

{
\raggedright
\begin{definition}[Safety Property]\label{def:safety}
  $\prop$ is a \emph{safety property} of runtime typing contexts $\rtenv$ if:
  \begin{enumerate}
    \item $\proparg{\rtenv, \sessindexroles{s}{\prole}{\rolesetr}{\sumact{i \in I}{\locacta_i \then
            \loctya_i}}, \sessindexroles{s}{\qrole}{\rolesets}{\sumact{j \in J}{
    \localtyrecv{\prole}{\ell_j}{\tyb_j} \then \loctyb_j}}}$ implies that if
      $\localtycomm{\qrole}{{!}}{\ell_k}{\tya_k} \in \{ \locact_i \}_{i \in I}$,
      then $k \in J$, $\qrole \in \rolesetr$, $\prole \in \rolesets$, and $\tya_k = \tyb_k$.
\item $\proparg{\rtenv, \sessindexroles{s}{\prole}{\rolesetr}{\sumact{i \in I}{\locact_i \then
      \locty_i}}}$ implies that
      if $\locact_i = \localtyconn{\qrole}{\ell_j}{\tya_j} \in \{ \locact_i
      \}_{i \in I}$,
      then $\qrole \not\in \rolesetr$,
      $\sessindexrolesnoty{s}{\qrole}{\rolesets} \not\in \dom{\rtenv}$, and
      $\ty{\qrole} = \sumact{k \in K}{\localtyaccept{\prole}{\ell_k}{\tyb_k} \then \loctyb_k}$
      with $j \in K$ and $\tya_j = \tyb_j$.
    \item $\proparg{\rtenv, \sessindexroles{s}{\prole}{\rolesetq}{\recty{X}{S}}}$ implies
      $\proparg{\rtenv, \sessindexroles{s}{\prole}{\rolesetq}{S \{ \recty{X}{S} / X \}}}$
    \item $\proparg{\rtenv}$ and $\rtenv \synceval \rtenv'$ implies $\proparg{\rtenv'}$
  \end{enumerate}

  A runtime typing environment is \emph{safe}, written
  $\safe{\rtenv}$, if $\proparg{\rtenv}$ for a safety property $\prop$.
\end{definition}
}
Clause (1) ensures that communication actions between participants are
compatible: if $\prole$ is sending a message with label $\ell$ and payload type
$\tya$ to $\qrole$, and $\qrole$ is receiving from $\prole$,
then the two roles must be connected,
and $\qrole$ must be able to receive $\ell$ with a matching payload.

Clause (2) states that if $\prole$ is connecting to a role $\qrole$ with label
$\ell$, then
$\qrole$ should not already be involved in the session, and should be able to
accept from $\prole$ on message label $\ell$ with a compatible payload type.
The requirement that $\qrole$ is not already involved in the session rules
out the \emph{correlation} errors described in
Section~\ref{sec:formalism:term-typing}.
Clause (3) handles recursion, and clause (4) requires that
safety is preserved under environment reduction.

\subparagraph{Concretising the safety property.}
In order to deduce that a runtime typing environment $\rtenv$ is safe, we define
$\proparg{\rtenv} = \{ \rtenv' \midspace \rtenv \syncevalstar \rtenv' \}$
and verify that $\prop$ is a safety property by ensuring that it satisfies all
clauses in~\autoref{def:safety}.

\subparagraph{Properties on protocols and programs.}
It is useful to distinguish active and inactive session types, depending on
whether their associated role is currently involved in a session, and identify
the initiator of a session.

\begin{definition}[Active and Inactive Session Types]
A session type $S$ is \emph{inactive}, written $\tyinactive{S}$, if $S =
\localend$ or $S = \sumact{i \in I}{\localtyaccept{\prole}{\ell_i}{\tya_i} \then {\loctya_i}}$. Otherwise, $S$ is \emph{active}, written $\tyactive{S}$.
\end{definition}

\begin{definition}[Initiator, unique initiator]
Given a protocol $\proto$, a role $\prole : \loctya_\prole \in \proto$ is an \emph{initiator} if
$S_\prole = \sumact{i \in I}{ \locact_i \then \loctya_i }$, and each $\locact_i$ is a connection action
$\localtyconn{\qrole}{\ell_i}{A_i}$.
Role $\prole$ is a \emph{unique initiator} of $\proto$ if
$\tyinactive{\loctya_{\qrole}}$ for all
$\qrole \in \proto \setminus \{ \prole : \loctya_{\prole} \} $.
\end{definition}

A protocol is \emph{well-formed} if it is safe and has a unique initiator.

\begin{definition}[Well-formed protocol]
  A protocol $P = \{ \prole_i : \loctya_i \}_{i \in I}$ is \emph{well-formed}
  if it has a unique initiator $\qrole$ of type $\loctya$ and
  $\safe{ \sessindexroles{s}{\qrole}{\emptyset}{\loctya} }$ for any $s$.
\end{definition}

By way of example, the online shopping protocol is well-formed: \role{Customer}
is the protocol's unique initiator, and it is straightforward to verify that
$\safe{ \sessindexroles{s}{\role{Customer}}{\emptyset}{\,\ty{\role{Customer}}} }$.

\begin{definition}[Well-formed program]
  A program $(\seq{D}, \seq{P}, M)$ is \emph{well-formed} if:
\begin{enumerate}
    \item For each actor definition $D = \actordef{u}{\loctya}{N} \in \seq{D}$,
        there exists some role $\prole \in \seq{P}$ such that $\ty{\prole}
          = \loctya$, and
        $\tseq[\loctya]{\cdot}{\loctya}{N}{A}{\localend}$
    \item Each protocol $P \in \seq{P}$ is well-formed and has a distinct set of
      roles
    \item The `boot clause' $M$ is typable under the empty typing environment
      and does not perform any communication actions:
      $\tseq[\localend]{\cdot}{\localend}{M}{A}{\localend}$
  \end{enumerate}
\end{definition}
When discussing the metatheory, we only consider configurations defined with
respect to a well-formed program. Specifically, we henceforth assume that each
actor definition in the system follows a session type matched by a role in a
given protocol, assume each role belongs to a single protocol, and assume that
all protocols are well-formed.

Given a safe runtime environment, configuration reduction preserves typability;
details can be found in Appendix~\ref{appendix:proofs}.
We write $\mathcal{R}^?$ for the reflexive closure of a relation $\mathcal{R}$.

\begin{restatable}[Preservation (Configurations)]{theorem}{configpres}\label{thm:configpres}
    Suppose $\cseq{\Gamma}{\rtenv}{\config{C}}$ with $\safe{\rtenv}$ and where
    $\config{C}$ is defined wrt.\ a well-formed program.
    If $\config{C} \ceval \config{C}'$, then there exists some $\rtenv'$ such that $\rtenv \maybesynceval
  \rtenv'$ and $\cseq{\Gamma}{\rtenv'}{\config{C}'}$.
\end{restatable}

Preservation shows that each actor conforms to its session type, and that
communication never introduces unsoundness due to mismatching payload types.

\subsubsection{Progress}\label{sec:formalism:progress}
We now show a progress property, which shows that given
\emph{protocols} which satisfy a progress property, \EnsembleS
\emph{configurations} do not get stuck due to deadlocks.

A \emph{final} runtime typing environment contains a single, disconnected role
of type $\localend$, reflecting the intuition that all roles will eventually
disconnect from a protocol initiator.

\begin{definition}[Final environment]
  An environment $\rtenv$ is \emph{final}, written $\finished{\rtenv}$,
  $\rtenv = \{
      \sessindexroles{s}{\prole}{\emptyset}{\localend} \}$ for some $s$ and
      $\prole$.
\end{definition}

So far, we have considered \emph{safe} protocols, which ensure the absence of
communication mismatches. We say that an environment \emph{satisfies progress}
if each active role can eventually perform an action, each potential send is
eventually matched by a receive, and non-reducing environments are final.
Let $\roles{\synclbl}$ denote the roles referenced in a
synchronisation label (i.e., $\roles{\synclbl} = \{ \prole, \qrole\} $ for
$\synclbl \in \{ \ltsconnpair{s}{\prole}{\qrole}{\ell},
\ltspair{s}{\prole}{\qrole}{\ell},
\ltsdisconnpair{s}{\prole}{\qrole} \}$).

\begin{definition}[Progress (Runtime typing
  environments)]\label{def:env-progress}
  A runtime typing environment $\rtenv$ \emph{satisfies progress}, written
  $\prog{\rtenv}$, if:
  \begin{itemize}
  \item \textit{(Role progress)} for each $\sessindexroles{s}{\prole_i}{\rolesetq_i}{\loctya_i} \in
      \rtenv$ s.t.\ $\tyactive{\loctya_i}$, $\rtenv
      \synceval^* \rtenv' \syncannarrow{\synclbl}$ with $\prole \in \roles{\synclbl}$
    \item \textit{(Eventual comm.)}
      if $\rtenv \!\syncevalstar\! \rtenv' \!\annarrow{\ltscomm[{!}]{s}{\prole}{\qrole}{\ell}{A}}$,
      then $\rtenv' \syncannarrow{\seq{\synclbl}} \rtenv''
        \annarrow{\ltscomm[{?}]{s}{\qrole}{\prole}{\ell}{A}}$,
        with $\prole \not\in \roles{\seq{\synclbl}}$
    \item \textit{(Correct termination)} $\rtenv \syncevalstar
     \rtenv' \not\synceval$ implies $\finished{\rtenv}$
   \end{itemize}
\end{definition}

The online shopping example satisfies progress, since all roles will always
eventually be able to fire an action once connected, and since all roles
disconnect, the non-reducing final environment will be of the form
$\sessindexroles{s}{\role{Customer}}{\emptyset}{\localend}$.

\begin{definition}[Progress (Programs)]
    A well-formed program $(\seq{D}, \seq{P}, M)$ \emph{satisfies progress}
    if each $P \in \seq{P}$ has a unique initiator $\qrole$ of type $\loctya$ and
    $\prog{ \sessindexroles{s}{\qrole}{\emptyset}{\loctya} }$ for any $s$.
\end{definition}

It is useful to define a \emph{configuration context} $\config{G}$ as
the one-hole context $\config{G} ::=  [~] \midspace (\nu s)\config{G} \midspace \config{G}
\parallel \config{C}$.
A \emph{session} consists of a session name restriction and all connected actors
and zapper threads.  Each well-typed configuration can be written as a sequence
of sessions, followed by all disconnected actors.

\begin{definition}[Session]
  A configuration is a \emph{session} $\sgroup$ if it can be written:\\
$
  (\nu s)(\eactor{a_1}{M_1}{\newconnsess{s}{\prole_1}{\rolesetq_1}}{\behaviour_1}
  \parallel \cdots \parallel
  \eactor{a_m}{M_m}{\newconnsess{s}{\prole_m}{\rolesetq_m}}{\behaviour_m}
  \parallel
  \zap{\sessindextwo{s}{\prole_{m + 1}}}
  \parallel \cdots
  \parallel
  \zap{\sessindextwo{s}{\prole_{n}}}
  )
$
\end{definition}

An actor is \emph{terminated} if it has reduced to a value or has an unhandled
exception, and has the behaviour $\bstop$.  An unmatched discover occurs when no
other actors match a given session type.  An actor is \emph{accepting} if it is
ready to accept a connection.

\begin{definition}[Terminated actor, unmatched discover, accepting actor]\hfill
  \begin{itemize}
    \item An actor $\eactor{a}{M}{\connstate}{\behaviour}$ is \emph{terminated}
      if $M = \effreturn{V}$ or $M = \ep[\raiseexn]$, and $\behaviour = \bstop$.
\item An actor $\eactor{a}{E[\newdiscover{S}]}{\connstate}{\behaviour}$ which is a subconfiguration
      of $\config{C}$ has an \emph{unmatched discover} if no other
      non-terminated actor in
      $\config{C}$ has session type $S$.
\item An actor $\eactor{a}{M}{\connstate}{\behaviour}$
      is \emph{accepting} if $M =
      E[\newaccept{\prole}{\msg{\ell_j}{x_j} \mapsto N_j}_j]$ for
      some evaluation context $E$ and role $\prole$.
\end{itemize}
\end{definition}

Unhandled exceptions will propagate through a session, progressively cancelling
all roles. A \emph{failed session} consists of only zapper threads.

\begin{definition}[Failed session]
  We say that a session $\sgroup$ is a \emph{failed session}, written
  $\failed{\sgroup}$, if
  $\sgroup \equiv (\nu s)(\zaptwo{s}{\prole_1} \parallel \cdots \parallel \zaptwo{s}{\prole_n})$.
\end{definition}

The key \emph{session progress} lemma establishes the reducibility of each session which does not
contain an unmatched discover and is typable under a reducible runtime typing
environment.

\begin{restatable}[Session Progress]{lemma}{sessgroupprogress}\label{lem:sess-group-progress}
  If $\cseq{\cdot}{\cdot}{\config{C}}$ where $\config{C}$ does not contain an
  unmatched discover, $\config{C} \equiv
  \config{G}[\sgroup]$ and $\sgroup = (\nu s: \rtenv) \config{D}$ with
  $\prog{\rtenv}$, and $\sgroup$ is not a failed session, then $\config{C} \ceval$.
\end{restatable}

There are several steps to proving Lemma~\ref{lem:sess-group-progress}.
First,
we introduce \emph{exception-aware} reduction on runtime typing environments,
which explicitly accounts for zapper threads at the type level, and show that
exception-aware environments threads retain safety and progress. Second, we
introduce \emph{flattenings}, which show that runtime typing
environments containing only unary output choices can type configurations
blocked on communication actions, and that flattenings preserve environment
reducibility. Finally, we show that configurations typable under flat,
reducible typing environments can reduce. Full details can be found in
Appendix~\ref{appendix:proofs}.

We can now show our second main result: in the absence of unmatched discovers,
a configuration can either reduce, or it consists only of accepting and
terminating actors.

\begin{restatable}[Progress]{theorem}{progress}\label{thm:progress}
    Suppose $\cseq{\cdot}{\cdot}{\config{C}}$ where $\config{C}$ is defined
    wrt.\ a well-formed program which satisfies progress, and $\prog{\rtenv}$ for
    each $(\nu s : \rtenv) \config{C}'$ in $\config{C}$.
If $\config{C}$ does not contain an unmatched discover,
  either
  $\exists \config{D}$ such that $\config{C}
  \ceval \config{D}$, or
  $\config{C} \equiv \confzero$, or
  $\config{C} \equiv (\nu b_1 \cdots \nu b_n)(\eactor{b_1}{N_1}{\bot}{\behaviour_1} \parallel \cdots \parallel
    \eactor{b_n}{N_n}{\bot}{\behaviour_n})$
  where each $b_i$ is terminated or accepting.
\end{restatable}

The proof eliminates all failed sessions by the structural congruence
rules; shows that the presence of sessions implies reducibility
(Lem.~\ref{lem:sess-group-progress}); and reasons about disconnected actors.

In addition to each
actor conforming to its session type~(\autoref{thm:configpres}),
Theorem~\ref{thm:progress} guarantees that the system does not deadlock. It
follows that session types ensure safe communication.

Theorem~\ref{thm:progress} assumes the absence of unmatched discovers. This is not a
significant limitation in practice, however, as unmatched discovers can be mitigated with
timeouts, where a timeout would trigger an exception.

\section{Related Work}
\label{sec:related}

\subparagraph{Behavioural typing for actors.}
\citet{MostrousV11} present the first theoretical account of session types in an actor
language; their work effectively overlays a channel-based session
typing discipline on mailboxes using Erlang's
reference generation capabilities.

\citet{NeykovaY17} use MPSTs to specify
communication in an actor system, implemented in Python.
Fowler~\cite{Fowler2016} implements similar ideas in Erlang,
with extensions to allow subsessions~\cite{DemangeonH12} and failure handling.
\citet{NeykovaY17a} later improve the recovery
mechanism of Erlang by using MPSTs to calculate a minimal set of affected roles.
The above works check multiparty session typing dynamically. We are first to
both formalise and implement static multiparty session type checking for an
actor language.

Active objects (AOs)~\cite{BoerCJ07} are actor-like concurrent objects
where results of asynchronous method invocations are returned through futures.
\citet{BagherzadehR17} study \emph{order types} for an AO calculus, which
characterise causality and statically rule out data races. In contrast to MPSTs,
order types work bottom-up through type inference.
\citet{KamburjanDC16} apply an MPST-like system to Core ABS~\cite{JohnsenHSSS10},
a core AO calculus;
they establish soundness via a translation to register automata rather than via an operational semantics.

\citet{deLiguoroP18} introduce \emph{mailbox types}, a type system for
first-class, unordered mailboxes. Their calculus generalises the actor
model, since each process can be associated with more than one mailbox.
Their type discipline allows multiple writers and a single reader for each
mailbox, and ensures conformance, deadlock-freedom, and for many programs,
junk-freedom. Our approach is based on MPSTs and is more
process-centric.

\subparagraph{Non-classical multiparty session types.}
MPSTs were introduced
by~\citet{Honda:2008:MAS:1328897.1328472}.
\emph{Classical} MPST theories are grounded in binary duality: safety follows as
a consequence of \emph{consistency} (pointwise binary duality of interactions
between participants), and deadlock-freedom follows from projectability from a
global type.

Unfortunately, classical MPSTs are restrictive: there are many protocols which
are intuitively safe but not consistent. \citet{ScalasY19} introduced the first
\emph{non-classical} multiparty session
calculus. Instead of ensuring safety using binary duality, they define
an LTS on local types and \emph{safety property} suitable for proving type
preservation; since the type system is \emph{parametric} in the safety property
(inspired by~\citet{IgarashiK04} in the $\pi$-calculus), the property can be
customised in order to guarantee different properties such as deadlock-freedom
or liveness.
\citet{Hu2017} formalise MPSTs with explicit connection actions via an LTS on
types rather than providing a concrete language design or calculus; in our
setting, a calculus is vital in order to account for the impact of adaptation
constructs. A key contribution of our work is the use of non-classical MPSTs to
prove preservation and progress properties for a calculus incorporating MPSTs
with explicit connection actions.

\subparagraph{Adaptation.}
None of the above work considers adaptation. The literature on formal studies of
adaptation is mainly based on process calculi, without programming language
design or implementation. Bravetti \emph{et al.} \cite{BravettiGPZ12} develop a
process calculus that allows parts of a process to be dynamically replaced with
new definitions. Their later work \cite{BravettiGPZ12a} uses temporal logic
rather than types to verify adaptive processes. Di Giusto and P\'{e}rez
\cite{DiGiustoP15} define a session type system for the same process calculus,
and prove that adaptation does not disrupt active sessions.
Later,~\citet{DiGiustoP15a} use an event-based approach so that adaptation
can depend on the state of a session protocol.
Anderson and Rathke \cite{AndersonR12} develop an MPST-like calculus and study
dynamic software update providing guarantees of communication safety and
liveness. Differently from our work, they do not consider runtime discovery of
software components and do not provide an implementation.

Coppo \emph{et al.} \cite{CoppoDV15} consider self-adaptation, in which a system
reconfigures itself rather than receiving external updates.
They define an MPST calculus with self-adaptation
and prove type safety. Castellani \emph{et al.} \cite{CastellaniDP16} extend
\cite{CoppoDV15} to also guarantee properties of secure information flow,
but neither of these works have been implemented.
Dalla Preda \emph{et al.} \cite{PredaGGLM16} develop the AIOCJ system based on
choreographic programming for runtime updates. Their work is implemented in the
Jolie language~\cite{MontesiGZ07}, but they do not consider runtime discovery.

In this work we focus on correct communication in the absence of adversaries, and do not consider security. The literature on security and behavioural types is surveyed by Bartoletti \emph{et al.} \cite{BartolettiCDDGP15} and could provide a basis for future work on security properties.
 
\section{Conclusion and Future Work}
\label{sec:future-conc}

Modern computing increasingly requires software components to \emph{adapt} to
their environment, by \emph{discovering}, \emph{replacing}, and
\emph{communicating} with other components which may not be part of the
system's original design. Unfortunately, up until now, existing programming
languages have lacked the ability to support adaptation both
\emph{safely} and \emph{statically}.
We therefore asked:
\begin{quote}
  \textit{Can a programming language support static (compile-time) verification of safe runtime
  dynamic self-adaptation, i.e., discovery, replacement and communication?}
\end{quote}

We have answered this question in the affirmative by introducing \EnsembleS, an
actor-based language supporting adaptation, which uses multiparty session types
to guarantee communication safety, using explicit connection
actions to invite discovered actors into a session.
We have demonstrated the safety of our system by proving  type soundness
theorems which state that each actor follows its session type, and that
communication does not introduce deadlocks. Our formalism makes essential use of
\emph{non-classical} MPSTs.

\subparagraph{Future work.}

Currently, in both the theory and implementation, each actor only takes part in
a single session. Unlike dynamically-checked implementations of session typing
for actors~\cite{NeykovaY14, Fowler2016}, this means that a message received by an actor in one
session cannot trigger an interaction in another (e.g., the Warehouse example
in~\cite{NeykovaY14}).
A key focus for future work will be to allow actors to partake in multiple sessions.

Currently, \EnsembleS discovery and replacement requires type equality.  We
envisage that we could relax this constraint to subtyping~\cite{ScalasDHY17} or
perhaps bisimilarity on local types to increase expressiveness.

In order to avoid session correlation errors, we require that each role includes
at most a single top-level $\calcwd{accept}$ construct (c.f.~\cite{Hu2017}). It
would be interesting to investigate the more general setting, which would likely
require dependent types.
 
\bibliography{master}

\newpage

\appendix

\begin{adjustwidth}{-70pt}{-60pt}
  \changetext{0pt}{12em}{}{}{}\onecolumn
  \nolinenumbers
  \section{Global types and projection}\label{appendix:projection}

In this section, we show how global types can be projected into the local types
used in the core formalism. Unlike in classic MPST works, it is not the case
that projection alone guarantees safety and deadlock-freedom in itself.
Rather, projection produces local types,
which can be validated separately as described in the main
body of the paper.

This section is mostly a straightforward adaptation of the
work of~\citet{Hu2017} to our setting.

\begin{figure}
  \small
\header{Meta-level definitions}

\[
  \begin{array}{rcl}
    \isoutput{L} & \defeq & L \in \{ \localtysend{\prole}{\ell}{A}, \localtyconn{\prole}{\ell}{A} \} \\
    \isinput{L} & \defeq & L \in \{ \localtyrecv{\prole}{\ell}{A},  \localtyaccept{\prole}{\ell}{A} \} \\
    \subj{\globalsend{\prole}{\qrole}{\ell}{A}},
    \subj{\globalconn{\prole}{\qrole}{\ell}{A}},
    \subj{\globaldisconn{\prole}{\qrole}} & \defeq & \prole \\
    \obj{\globalsend{\prole}{\qrole}{\ell}{A}},
    \obj{\globalconn{\prole}{\qrole}{\ell}{A}},
    \obj{\globaldisconn{\prole}{\qrole}} & \defeq & \qrole \\
  \end{array}
\]

\header{Projection}
{\framebox{$\project{G}{\rrole}$}~\framebox{$\hyproject{G}{\role{r}}{\recenv}$}}

\[
  \begin{array}{rcl}
    \project{G}{\rrole} & = & \hyproject{G}{\rrole}{\emptyset} \\
    \hyproject{(\globalsend{\prole}{\qrole}{\ell}{\tya} \then G)}{\rrole}{\recenv} & = &
      {
        \begin{cases}
          \localtysend{\qrole}{\ell}{\tya} \then (\hyproject{G}{\rrole}{\recenv}) & \rrole = \prole \\
          \localtyrecv{\prole}{\ell}{\tya} \then (\hyproject{G}{\rrole}{\recenv}) & \rrole = \qrole \\
          \hyproject{G}{\rrole}{\recenv} & \rrole \not\in \{ \prole, \qrole \} \\
        \end{cases}
      }
      \rowskip
      \\
\hyproject{(\globalconn{\prole}{\qrole}{\ell}{\tya} \then G)}{\rrole}{\recenv} & = &
      {
        \begin{cases}
          \localtyconn{\qrole}{\ell}{\tya} \then (\hyproject{G}{\rrole}{\recenv}) & \rrole = \prole \\
          \localtyaccept{\prole}{\ell}{\tya} \then (\hyproject{G}{\rrole}{\recenv}) & \rrole = \qrole \\
          \hyproject{G}{\rrole}{\recenv} & \rrole \not\in \{ \prole, \qrole \} \\
        \end{cases}
      }
      \rowskip
      \\
\hyproject{(\globaldisconn{\prole}{\qrole} \then {G})}{\rrole}{\recenv} & = &
     {
       \begin{cases}
          \localtydisconn{\qrole} & \rrole = \prole \text{ and }
          \hyproject{G}{\rrole}{\emptyset} = \localend \\
          \localtywait{\prole} \then (\hyproject{G}{\rrole}{\emptyset}) & \rrole = \qrole \\
\hyproject{G}{\rrole}{\recenv} & \rrole \not\in \{ \prole, \qrole \}
       \end{cases}
     }
     \rowskip
     \\
\hyproject{\sumact{i \in I}{G_i}}{\rrole}{\recenv} & = &
     {
       \begin{cases}
         X & \forall i \in I. (\hyproject{G_i}{\rrole}{\recenv}) = X \\
         \localend & \forall i \in I. (\hyproject{G_i}{\rrole}{\recenv}) = \localend \\
         \sumact{j \in J \subseteq I}{L_j = \hyproject{G_j}{\rrole}{\recenv}} &
         {
           \begin{array}{l}
             \size{I} > 1, \size{J} > 0, \text{ and } \\
             \forall k \in I \setminus J . \hyproject{G_k}{\rrole}{\recenv} = \localend \text{ or } X \in \recenv, \text{ and } \\
             \quad
             \text{either }
             {
               \begin{cases}
                 \forall j \in J. \isoutput{L_j} \\
                 \exists \role{\prole} . \forall j \in J. \isinput{L_j} \wedge \subj{L_j} = \role{\prole}
               \end{cases}
             }
           \end{array}
         }
       \end{cases}
     }
     \rowskip
     \\
     \hyproject{\recty{X}{G}}{\rrole}{\recenv} & = &
     {
       \begin{cases}
         \localend & \hyproject{G}{\rrole}{\recenv, X} = X' \text{ or } \localend \\
         \hyproject{G}{\rrole}{\Delta, X} & \text{otherwise}
       \end{cases}
     } \rowskip \\
     \hyproject{X}{\rrole}{\recenv} & = & X \rowskip \\
     \hyproject{\globalend}{\rrole}{\recenv} & = & \localend
  \end{array}
\]

\caption{Global types and projection}
\label{fig:appendix:global-types}
\end{figure}

\subparagraph{Syntax of types.}
Similar to the presentation of local types, the key difference between this
presentation of global types and standard MPST systems is that instead of having
branching of the form $\role{r} \rightarrow \role{s} \{ \ell_i : G_i\}_i$, where
a single role makes a choice, this presentation of global types provides
arbitrary summations of \emph{global actions}. The standard sending action is
therefore $\globalsend{\prole}{\qrole}{\ell}{\tya} \then G$ which involves role
$\prole$ sending message $\ell(\tya)$ to $\qrole$.

A key technical innovation is \emph{explicit connection actions}
($\globalconn{\prole}{\qrole}{\ell}{\tya}$) which establish a connection between two roles
$\role{\prole}$ and $\role{\qrole}$ by sending a label $\ell$. Conversely,
\emph{disconnection} $\globaldisconn{\prole}{\qrole}$ is a directed
disconnection where $\prole$ disconnects from $\qrole$, and $\qrole$ waits for
$\prole$'s disconnection.

Instead of assuming that all roles are connected at the start of the the
session, in order to communicate, each role must be connected explicitly. This
allows communication paths where one party isn't involved (as in Hu and
Yoshida's travel agent example).

\subparagraph{Projection.}

Global types can be \emph{projected} at a role in order to obtain a local type.
In standard MPST systems, a global type is well-formed if it is closed,
contractive, and projectable at all roles. A well-formed global type is then
guaranteed to be deadlock-free.

Define $\roles{G}$ to be the set of roles contained within global type $G$.

The projection function
$\hyproject{G}{\rrole}{\recenv}$ projects global type $G$ at role $\rrole$, and is
parameterised by a set of recursion variables $\recenv$. The inclusion of the
$\recenv$ parameter allows the pruning of unguarded recursion variables.

Projection on unary choices, recursive types, type variables, and session
termination are standard. Disconnection ensures that the disconnecting role is
unused in the remainder of the protocol.

Now consider the case of projecting $n$-ary choices for $n > 1$.
Suppose we have $\hyproject{\sumact{i \in I}{G_i}}{\rrole}{\recenv}$.
If $\hyproject{G_i}{\rrole}{\recenv} = \localend$ for all $G_i$, then the result of the
projection is $\localend$ (and similarly for a recursion variable $X$).

Otherwise, there will be some subset $J$ of $I$ such that each $L_j$ for $j \in
J$ is a communication action, and each $L_k \in I \setminus J$ projects to
either $\localend$ or a recursion variable.
For each $j \in J$, there is a further condition: the choice must consist of all
send actions, or all receive actions. If the latter, then the receiver must be
the same in all branches. This will ensure syntactically well-formed local
types used in the paper.

\subparagraph{Unfolding.}
We define the 1-unfolding of a global type $G$, $\unf{G}$, as
$\unf{\recty{X}{G}} = \unf{G\{ \globalend / X \}}$ and defined recursively over
the other constructs.

\subparagraph{Role enabling.}

In order to ensure that projection produces syntactically well-formed local
types, we require global types to satisfy \emph{role enabling}: intuitively,
this means that after a choice occurs, each role must be `activated' by
receiving a message before performing any further communications. In turn, this
avoids mixed choice.

Let $R$ range over sets of roles.

\headersig{Role enabling}{$\enableseq{R}{G}$}
\begin{mathpar}
    \inferrule
    { \subj{\globact} \subseteq R \\
      \enableseq{R \cup \obj{\globact}}{G} }
    { \enableseq{R}{\globact \then G} }

    \inferrule
    { \size{I} > 1 \\
      \exists \prole \in R .
      \forall i \in I.
        \subj{\globact_i} = \{ \prole \} \wedge
        \enableseq{\{ \prole \} \cup \obj{\globact_i}}{G_i}
    }
    { \enableseq{R}{\sumact{i \in I}{G_i}} }
\end{mathpar}

\subparagraph{Syntactic validity of projected local types.}

\begin{definition}[Initiator (global type)]
We say that $\prole$ is an \emph{initiator} of $G$ if $\prole \in
\roles{G}$ and $\project{G}{\prole} = \sumact{i \in I}{\globact_i \then G_i}$,
where each $\globact_i$ is a connection action. We say that $\prole$ is the
\emph{unique initiator} of $G$ if it is the only initiator.
\end{definition}

Given the above definitions, we can show that the set of projected local types
are syntactically valid.

\begin{lemma}
    Suppose $\prole$ is the unique initiator of $G$.
If $\enableseq{\prole}{\unf{G}}$ and $(\project{G}{\qrole} = L_\qrole)_{\qrole \in \roles{G}}$,
    then each $L_\qrole$ is syntactically valid.
\end{lemma}
\begin{proof}[Proof sketch.]
    The result can be established by induction on the derivation of
    $\hyproject{G}{\prole}{\recenv}$.

    Due to the partiality of the projection operation, disconnection actions
    must only occur as unary choices, and role enabling coupled with the
    conditions on the projection of non-unary choices preclude mixed choice.
\end{proof}

   \clearpage
  \section{Proofs for Section~\ref{sec:core-calculus}}\label{appendix:proofs}

\subsection{Preservation}

\begin{lemma}[Preservation (Terms)]\label{lem:term-pres}
  If $\tseq{\Gamma}{\loctya}{M}{A}{\loctya'}$ and $M \teval N$, then $\tseq{\Gamma}{\loctya}{N}{A}{\loctya'}$.
\end{lemma}
\begin{proof}
  By induction on the derivation of $M \teval N$.
\end{proof}

\begin{lemma}[Preservation (Equivalence)]\label{lem:equiv-pres}
  If $\cseq{\Gamma}{\rtenv}{\config{C}}$ and $\config{C} \equiv \config{D}$,
  then $\cseq{\Gamma}{\rtenv}{\config{D}}$.
\end{lemma}
\begin{proof}
  By induction on the derivation of $\config{C} \equiv \config{D}$.
\end{proof}

\begin{lemma}[Substitution]\label{lem:substitution}
  If $\tseq{\Gamma, x : A}{S}{M}{B}{S'}$ and $\vseq{\Gamma}{V}{A}$, then
  $\tseq{\Gamma}{S}{M \{ V / x \}}{B}{S'}$.
\end{lemma}
\begin{proof}
  By induction on the derivation of $\tseq{\Gamma, x : A}{S}{M}{B}{S'}$.
\end{proof}

\begin{lemma}[Subterm typability]\label{lem:subterm-typability}
  Suppose $\deriv{D}$ is a derivation of $\tseq{\Gamma}{S}{E[M]}{A}{S'}$.
  Then there exists some subderivation $\deriv{D}'$ of $\deriv{D}$ concluding
  $\tseq{\Gamma}{S}{M}{B}{S''}$ for some type $B$ and local type $S''$, where
  the position of $\deriv{D}'$ in $\deriv{D}$ corresponds to that of the hole in
  $E$.
\end{lemma}
\begin{proof}
  By induction on the structure of $E$.
\end{proof}

\begin{lemma}[Replacement]\label{lem:replacement}
  If:

  \begin{enumerate}
    \item $\deriv{D}$ is a derivation of $\tseq[\loctyc]{\Gamma}{S}{E[M]}{A}{T}$
    \item $\deriv{D}'$ is a subderivation of $\deriv{D}$ concluding
      $\tseq[\loctyc]{\Gamma}{S}{M}{B}{T'}$, where the position of $\deriv{D}'$ in
      $\deriv{D}$ corresponds to that of the hole in $\ep$
    \item $\tseq[\loctyc]{\Gamma}{S}{N}{B}{T'}$
  \end{enumerate}

  Then $\tseq[\loctyc]{\Gamma}{S}{E[N]}{A}{T}$.
\end{lemma}
\begin{proof}
  By induction on the structure of $E$.
\end{proof}

Due to the absence of exception handling frames which constrain the
pre-condition to match that of the failure continuation,
pure contexts admit a more liberal replacement lemma.

\begin{lemma}[Replacement (Pure contexts)]\label{lem:pure-replacement}
  If:

  \begin{enumerate}
    \item $\deriv{D}$ is a derivation of $\tseq[\loctyc]{\Gamma}{S}{\ep[M]}{A}{T}$
    \item $\deriv{D}'$ is a subderivation of $\deriv{D}$ concluding
      $\tseq[\loctyc]{\Gamma}{S}{M}{B}{T'}$, where the position of $\deriv{D}'$ in
      $\deriv{D}$ corresponds to that of the hole in $\ep$
    \item $\tseq[\loctyc]{\Gamma}{S'}{N}{B}{T'}$
  \end{enumerate}

  Then $\tseq[\loctyc]{\Gamma}{S'}{\ep[N]}{A}{T}$.
\end{lemma}
\begin{proof}
  By induction on the structure of $\ep$, noting that $\loctya$ is not constrained
  by exception handling frames.
\end{proof}

\configpres*
\begin{proof}
  By induction on the derivation of $\config{C} \ceval \config{C}'$.
  Where there is a choice of whether $\connstate = \bot$ or $\connstate =
  \sessindexroles{s}{\prole}{\rolesetq}{\loctya}$, we show the latter case; the technique for
  proving the former case is identical.

  \begin{proofcase}{E-Loop}

    \[
      \eactor{a}{\effreturn{V}}{\bot}{M} \ceval
      \eactor{a}{M}{\bot}{M}
    \]

    Assumption:

    \begin{mathparsmall}
      \inferrule*
      {
        a : \pidty{\loctyb} \in \Gamma \\
        \loctya = \loctyb \vee \loctya = \localend \\
        \tseq[\loctyb]{\Gamma}{\loctya}{\effreturn{V}}{\tyb}{\localend} \\
        \inferrule*
        { \tseq[\loctyb]{\Gamma}{\loctyb}{M}{\tya}{\localend} }
        { \bseq[\loctyb]{\Gamma}{M} }
      }
      { \cseq{\Gamma}{a : \loctyb}{\eactor{a}{\effreturn{V}}{\bot}{M}} }
    \end{mathparsmall}

    Recomposing:

    \begin{mathparsmall}
      \inferrule*
      {
        a : \pidty{\loctyb} \in \Gamma \\
        \tseq[\loctyb]{\Gamma}{\loctyb}{M}{\tya}{\localend} \\
        \inferrule*
        { \tseq[\loctyb]{\Gamma}{\loctyb}{M}{\tya}{\localend} }
        { \bseq[\loctyb]{\Gamma}{M}}
      }
      { \cseq{\Gamma}{a : \loctyb}{\eactor{a}{M}{\bot}{M}} }
    \end{mathparsmall}

    as required.
  \end{proofcase}

  \begin{proofcase}{E-New}
    \[
      \eactor{a}{E[\new{u}]}{\connstate}{\behaviour} \ceval
      (\nu b) (\eactor{a}{E[\effreturn{b}]}{\connstate}{\behaviour} \parallel
     \eactor{b}{M}{\bot}{M})
    \]
    where $b$ is fresh, and $u.\mkwd{behaviour} = M$.

    Assumption:
    \begin{mathparsmall}
      \inferrule*
      {
        a : \pidty{\loctya} \in \Gamma \\
        \tseq[\loctya]{\Gamma}{\loctya}{E[\new{u}]}{\tya}{\localend} \\
        \bseq[\loctya]{\Gamma}{\behaviour_1}
      }
      { \cseq
          {\Gamma}
          {a : \loctya, \sessindexroles{s}{\prole}{\rolesetq}{\loctya'}}
          {\eactor{a}{E[\new{u}]}{\newconnsess{s}{\prole}{\rolesetq}}{\behaviour_1}}
      }
    \end{mathparsmall}

    By Lemma~\ref{lem:subterm-typability}:
    \begin{mathparsmall}
      \inferrule*
      { u.\mkwd{sessionType} = \loctyb }
      { \tseq[\loctya]{\Gamma}{\loctya'}{\new{u}}{\pidty{\loctyb}}{\loctya'} }
    \end{mathparsmall}

    By Lemma~\ref{lem:replacement}, $\tseq[\loctya]{\Gamma, b : \pidty{\loctyb}}{\loctyb}{\effreturn{b}}{\pidty{\loctyb}}{\loctya'}$.

    Let $\Gamma' = \Gamma, b : \pidty{\loctyb}$.
    Note that by weakening, everything typable under $\Gamma$ is also typable under $\Gamma'$.

    By \textsc{T-Def}:
    \begin{mathparsmall}
      \inferrule*
      { \tseq[\loctyb]{\cdot}{\loctyb}{M}{\tyb}{\localend} }
      { \dseq{\actordef{u}{\loctyb}{M}} }
    \end{mathparsmall}

    Again by weakening, $\tseq[\loctyb]{\Gamma'}{\loctyb}{M}{\tyb}{\localend}$.

    Recomposing:
    \begin{mathparsmall}
      \inferrule*
      {
        \inferrule*
        {
          \inferrule*
          {
            a : \pidty{\loctya} \in \Gamma' \\\\
            \tseq[\loctya]{\Gamma'}{\loctya}{E[\effreturn{b}]}{\tya}{\localend} \\
            \bseq[\loctya]{\Gamma'}{\behaviour_1}
          }
          { \cseq
            {\Gamma' }
            {a : \loctya, \sessindexroles{s}{\prole}{\rolesetq}{\loctya'}}
              {\eactor{a}{E[\effreturn{b}]}{\newconnsess{s}{\prole}{\rolesetq}}{\behaviour_1}}
          }
          \\
          \inferrule*
          {
            b : \pidty{\loctyb} \in \Gamma' \\
            \tseq[\loctyb]{\Gamma'}{\loctyb}{M}{\tyb}{\localend} \\
            \inferrule*
            { \tseq[\loctyb]{\Gamma'}{\loctyb}{M}{\tyb}{\localend} }
            { \bseq[\loctyb]{\Gamma'}{M} }
          }
          { \cseq
              {\Gamma' }
              { b : \loctyb}
              {\eactor{b}{M}{\bot}{M}}
          }
        }
        {
          \cseq
          { \Gamma' }
          { a : \loctya, b : \loctya }
          { \eactor{a}{E[\effreturn{b}]}{\newconnsess{s}{\prole}{\rolesetq}}{\behaviour_1}
            \parallel \eactor{b}{M}{\bot}{M}
          }
        }
      }
      { \cseq{\Gamma}{a : \loctya }{
        (\nu b)(\eactor{a}{E[\effreturn{b}]}{\newconnsess{s}{\prole}{\rolesetq}}{\behaviour_1}
            \parallel \eactor{b}{M}{\bot}{M}) }
      }
    \end{mathparsmall}

    as required.
  \end{proofcase}

  \begin{proofcase}{E-Replace}

    \[
      \eactor{a}{E[\replace{b}{\behaviour'_2}]}{\connstate}{\behaviour_1} \parallel
      \eactor{b}{M}{\connstate}{\behaviour_2}
    \]

    Assumption:

    \begin{mathparsmall}
      \inferrule*
      {
        \inferrule*
        {
          a : \pidty{\loctyc_a} \in \Gamma \\\\
          \tseq[\loctyc_a]{\Gamma}{\loctya}{E[\replace{b}{\behaviour'_2}]}{\tya}{\localend} \\
          \bseq[\loctyc_a]{\Gamma}{\behaviour_1}
        }
        {
          \cseq
            {\Gamma}
            {a : \loctyc_a,
            \sessindexroles{s}{\prole}{\loctya}{\rolesetr}}
            {\eactor{a}{E[\replace{b}{\behaviour'_2}]}{\connstate}{\newconnsess{s}{\prole}{\rolesetr}}}
        }
        \\
        \inferrule*
        {
          b : \pidty{\loctyc_b} \in \Gamma \\
          \tseq[\loctyc_b]{\Gamma}{\loctya}{M}{\tyb}{\localend} \\
          \bseq[\loctyc_b]{\Gamma}{\behaviour_2}
        }
        {
          \cseq
            {\Gamma}
            {b : \loctyc_b, \sessindexroles{t}{\qrole}{\rolesets}{\loctyb}{\rolesets}}
            { \eactor{b}{M}{\newconnsess{t}{\qrole}{\rolesets}}{\behaviour_2} }
        }
      }
      { \cseq
          {\Gamma}
          {a : \loctyc_a, b : \loctyc_b,
            \sessindexroles{s}{\prole}{\rolesetr}{\loctya},
          \sessindexroles{t}{\qrole}{\rolesets}{\loctyb}}
           {\eactor{a}{E[\replace{b}{\behaviour'_2}]}{\connstate}{\newconnsess{s}{\prole}{\rolesetr}} \parallel
             \eactor{b}{M}{\newconnsess{t}{\qrole}{\rolesets}}{\behaviour_2}
          }
      }
    \end{mathparsmall}

    By Lemma~\ref{lem:subterm-typability}:

    \begin{mathparsmall}
      \inferrule*
      {
        \bseq[\loctyc_b]{\Gamma}{\behaviour'_2} \\
        \vseq{\Gamma}{b}{\pidty{\loctyc_b}}
      }
      { \tseq[\loctyc_a]{\Gamma}{\loctya}{\replace{b}{\behaviour'_2}}{\one}{\loctya} }
    \end{mathparsmall}

    By Lemma~\ref{lem:replacement},
    $\tseq[\loctyc_a]{\Gamma}{\loctya}{E[\effreturn{()}]}{\tya}{\localend}$.

    Recomposing:

    \begin{mathparsmall}
      \inferrule*
      {
        \inferrule*
        {
          a : \pidty{\loctyc_a} \in \Gamma \\\\
          \tseq[\loctyc_a]{\Gamma}{\loctya}{E[\effreturn{()}]}{\tya}{\localend} \\
          \bseq[\loctyc_a]{\Gamma}{\behaviour_1}
        }
        {
          \cseq
            {\Gamma}
            {a : \loctyc_a, \sessindexroles{s}{\prole}{\loctya}{\rolesetr}}
            {\eactor{a}{E[\effreturn{()}]}{\connstate}{\newconnsess{s}{\prole}{\rolesetr}}}
        }
        \\
        \inferrule*
        {
          b : \pidty{\loctyc_b} \in \Gamma \\
          \tseq[\loctyc_b]{\Gamma}{\loctya}{M}{\tyb}{\localend} \\
          \bseq[\loctyc_b]{\Gamma}{\behaviour'_2}
        }
        {
          \cseq
            {\Gamma}
            {b : \loctyc_b, \sessindexroles{t}{\qrole}{\loctyb}{\rolesetr}}
            { \eactor{b}{M}{\newconnsess{t}{\qrole}{\rolesets}}{\behaviour'_2} }
        }
      }
      { \cseq
          {\Gamma}
          {a : \loctyc_a, b : \loctyc_b,
            \sessindexroles{s}{\prole}{\rolesetr}{\loctya},
          \sessindexroles{t}{\qrole}{\rolesets}{\loctyb}}
           {\eactor{a}{E[\effreturn{()}]}{\connstate}{\newconnsess{s}{\prole}{\rolesetr}} \parallel
             \eactor{b}{M}{\newconnsess{t}{\qrole}{\rolesets}}{\behaviour'_2}
          }
      }
    \end{mathparsmall}

    as required.
  \end{proofcase}

  \begin{proofcase}{E-ReplaceSelf}

    \[
      \eactor{a}{E[\replace{a}{\behaviour'}]}{\connstate}{\behaviour} \ceval \eactor{a}{E[\effreturn{()}]}{\connstate}{\behaviour'}
    \]

    Assumption:

    \begin{mathparsmall}
      \inferrule*
      {
        a : \pidty{\loctyb} \in \Gamma \\
        \tseq[\loctyb]{\Gamma}{\loctya}{E[\replace{a}{\behaviour'}]}{A}{\localend} \\
        \bseq[\loctyb]{\Gamma}{\behaviour}
      }
      { \cseq
          {\Gamma}
          { a : \loctyb, \sessindexroles{s}{\prole}{\rolesetq}{\loctya}}
          {
          \eactor{a}{E[\replace{a}{\behaviour'}]}{\newconnsess{s}{\prole}{\rolesetq}}{\behaviour}}
      }
    \end{mathparsmall}

    By Lemma~\ref{lem:subterm-typability}:

    \begin{mathparsmall}
      \inferrule
      {
        \bseq[\loctyb]{\Gamma}{\behaviour'}
        \\
        \vseq{\Gamma}{a}{\pidty{\loctyb}}
      }
      { \tseq[\loctyb]{\Gamma}{\loctya}{\replace{a}{\behaviour'}}{\one}{\loctya} }
    \end{mathparsmall}

    (noting that $\vseq{\Gamma}{b}{\pidty{\loctyb}}$ because $a :
    \pidty{\loctyb} \in \Gamma$, as per the \textsc{T-ConnectedActor} and
    \textsc{T-DisconnectedActor} rules).

    By Lemma~\ref{lem:replacement},
    $\tseq[\loctyb]{\Gamma}{\loctya}{E[\effreturn{()}]}{\tya}{\localend}$.

    Recomposing:
    \begin{mathparsmall}
      \inferrule*
      {
        a : \pidty{\loctyb} \in \Gamma \\
        \tseq[\loctyb]{\Gamma}{\loctya}{E[\effreturn{()}]}{A}{\localend} \\
        \bseq[\loctyb]{\Gamma}{\behaviour'}
      }
      { \cseq
          {\Gamma}
          {a : \loctyb, \sessindexroles{s}{\prole}{\loctya}{\rolesetq} }
          { \eactor{a}{E[\effreturn{()}]}{\newconnsess{s}{\prole}{\rolesetq}}{\behaviour'} }
      }
    \end{mathparsmall}

    as required.
  \end{proofcase}

  \begin{proofcase}{E-Discover}

    \[
      \eactor{a}{E[\newdiscover{\loctyb}]}{\sigma_1}{\behaviour_1}
        \parallel
      \eactor{b}{M}{\sigma_2}{\behaviour_2}
        \ceval
      \eactor{a}{E[\effreturn{b}]}{\sigma_1}{\behaviour_1}
        \parallel
      \eactor{b}{M}{\sigma_2}{\behaviour_2}
    \]

    where $b.\mkwd{sessionType} = \loctyb$ and
    $\neg (M = \effreturn{V} \wedge \behaviour_2 = \bstop)$.

    Assumption:
    \begin{mathparsmall}
      \inferrule*
      {
        \inferrule*
        {
          a : \pidty{\loctya} \in \Gamma \\\\
          \tseq[\loctya]{\Gamma}{\loctya'}{E[\newdiscover{\loctyb}]}{\tya}{\localend} \\
          \bseq[\loctya]{\Gamma}{\behaviour_1}
        }
        {
          \cseq
          { \Gamma }
          { a : \loctya, \sessindexroles{s}{\prole}{\rolesetr}{\loctya'} }
          {
          \eactor{a}{E[\newdiscover{\loctyb}]}{\newconnsess{s}{\prole}{\rolesetr}}{\behaviour_1} }
        }
        \\
        \inferrule*
        {
          b : \pidty{\loctyb} \in \Gamma \\
          \tseq[\loctyb]{\Gamma}{\loctyb'}{M}{\tyb}{\localend} \\
          \bseq[\loctyb]{\Gamma}{\behaviour_2}
        }
        {
          \cseq
          { \Gamma }
          { b : \loctyb, \sessindexroles{t}{\qrole}{\loctyb'} }
          { \eactor{b}{M}{\newconnsess{t}{\qrole}{\rolesets}}{\behaviour_2} }
        }
      }
      {
        \cseq
        { \Gamma }
        { a : \loctya, b : \loctyb,
          \sessindexroles{s}{\prole}{\rolesetr}{\loctya'},
          \sessindexroles{t}{\qrole}{\rolesets}{\loctyb'} }
        {
          \eactor{a}{E[\newdiscover{S}]}{\sigma_1}{\behaviour_1}
            \parallel
          \eactor{b}{M}{\sigma_2}{\behaviour_2}
        }
      }
    \end{mathparsmall}

    By Lemma~\ref{lem:subterm-typability}:

    \begin{mathparsmall}
      \inferrule*
      { }
      { \tseq{\Gamma}{\loctya'}{\newdiscover{\loctyb}}{\pidty{\loctyb}}{\loctya'} }
    \end{mathparsmall}

    Since $b : \pidty{\loctyb} \in \Gamma$, we can show
    $\tseq{\Gamma}{\loctya'}{\newdiscover{\loctyb}}{\pidty{\loctyb}}{\loctya'}$.

    By Lemma~\ref{lem:replacement}, $\tseq{\Gamma}{\loctya'}{E[\effreturn{b}]}{\tya}{\loctya'}$.

    Thus, recomposing:

    \begin{mathparsmall}
      \inferrule*
      {
        \inferrule*
        {
          a : \pidty{\loctya} \in \Gamma \\
          \tseq[\loctya]{\Gamma}{\loctya'}{E[\effreturn{b}]}{\tya}{\localend} \\
          \bseq[\loctya]{\Gamma}{\behaviour_1}
        }
        {
          \cseq
          { \Gamma }
          { a : \loctya, \sessindexroles{s}{\prole}{\rolesetr}{\loctya'} }
          {
          \eactor{a}{E[\newdiscover{\loctyb}]}{\newconnsess{s}{\prole}{\rolesetr}}{\behaviour_1} }
        }
        \\
        \inferrule*
        {
          b : \pidty{\loctyb} \in \Gamma \\
          \tseq[\loctyb]{\Gamma}{\loctyb'}{M}{\tyb}{\localend} \\
          \bseq[\loctyb]{\Gamma}{\behaviour_2}
        }
        {
          \cseq
          { \Gamma }
          { b : \loctyb, \sessindexroles{t}{\qrole}{\rolesets}{\loctyb'} }
          { \eactor{b}{M}{\newconnsess{t}{\qrole}{\rolesets}}{\behaviour_2} }
        }
      }
      {
        \cseq
        { \Gamma }
        { a : \loctya, b : \loctyb,
        \sessindexroles{s}{\prole}{\rolesetr}{\loctya'},
        \sessindexroles{s}{\qrole}{\rolesets}{\loctyb'} }
        {
          \eactor{a}{E[\effreturn{b}]}{\newconnsess{s}{\prole}{\rolesetr}}{\behaviour_1}
            \parallel
            \eactor{b}{M}{\newconnsess{t}{\qrole}{\rolesets}}{\behaviour_2}
        }
      }
    \end{mathparsmall}

    as required.
  \end{proofcase}

  \begin{proofcase}{E-Self}

    \[
      \eactor{a}{E[\self]}{\connstate}{\behaviour} \ceval
      \eactor{a}{E[\self]}{\connstate}{\behaviour}
    \]

    Assumption:

    \begin{mathparsmall}
      \inferrule*
      {
        a : \pidty{\loctyb} \in \Gamma \\
        \tseq[\loctyb]{\Gamma}{\loctya}{E[\self]}{\tya}{\loctya}
      }
      { \cseq
          {\Gamma}
          {a : \loctyb, \sessindexroles{s}{\prole}{\rolesetq}{\loctya}}
          {\eactor{a}{E[\self]}{\newconnsess{s}{\prole}{\rolesetq}}{\behaviour}}
      }
    \end{mathparsmall}

    By Lemma~\ref{lem:subterm-typability}:

    \begin{mathparsmall}
      \inferrule*
      { }
      { \tseq[\loctyb]{\Gamma}{\loctya}{\self}{\pidty{\loctyb}}{\loctya} }
    \end{mathparsmall}

    We can show:

    \begin{mathparsmall}
      \inferrule*
      { \vseq{\Gamma}{a}{\pidty{\loctyb}} }
      { \tseq[\loctyb]{\Gamma}{\loctya}{\effreturn{a}}{\pidty{\loctyb}}{\loctya} }
    \end{mathparsmall}

    By Lemma~\ref{lem:replacement},
    $\tseq[\loctyb]{\Gamma}{\loctya}{E[\effreturn{a}]}{\tya}{\localend}$.

    Recomposing:

    \begin{mathparsmall}
      \inferrule*
      {
        a : \pidty{\loctyb} \in \Gamma \\
        \tseq[\loctyb]{\Gamma}{\loctya}{E[\effreturn{a}]}{\tya}{\loctya}
      }
      { \cseq
          {\Gamma}
          {a : \loctyb, \sessindexroles{s}{\prole}{\rolesetq}{\loctya}}
          {\eactor{a}{E[\effreturn{a}]}{\newconnsess{s}{\prole}{\rolesetq}}{\behaviour}}
      }
    \end{mathparsmall}

    as required.
  \end{proofcase}

 \begin{proofcase}{E-ConnInit}

   \[
     \bl
     \eactor{a}{E[F[\newconn{\ell_j}{V}{b}{\qrole}]]}{\bot}{\behaviour_1} \parallel
       \eactor{b}{E'[F[\newaccept{\prole}{\ell_i(x_i) \mapsto M_i}_{i \in
       I}]]}{\bot}{\behaviour_2} \ceval \\
(\nu s) (\eactor{a}{E[\effreturn{()}]}{\newconnsess{s}{\prole}{\qrole}}{\behaviour_1} \parallel
         \eactor{b}{E'[M_j \{ V / x_i \}]}{\newconnsess{s}{\qrole}{\prole}}{\behaviour_2})
       \el
   \]

   Assumption:

   {\small
   \begin{mathpar}
     \inferrule*
     {
       \inferrule*
       {
         a : \pidty{\loctyc_a} \in \Gamma \\
         \loctya = \loctyc_a \vee \loctya = \localend \\\\
         \tseq[\loctyc_a]{\Gamma}{\loctya}{E[\newconn{\ell_j}{V}{b}{\qrole}]}{A}{\localend}
         \\\\
         \bseq[\loctyc_a]{\Gamma}{\behaviour_1}
       }
       { \cseq
           { \Gamma }
           { a : \loctyc_a }
           { \eactor{a}{E[\newconn{\ell_j}{V}{b}{\qrole}]}{\bot}{\behaviour_1} }
       }
       \\
       \inferrule*
       {
         b : \pidty{\loctyc_b} \in \Gamma \\
         \loctyb = \loctyc_b \vee \loctyb = \localend \\\\
         \tseq[\loctyc_b]{\Gamma}{\loctyb}
           {E'[\newaccept{\prole}{\ell_i(x_i) \mapsto M_i}_{i \in I}}{B}{\localend}
         \\\\
         \bseq[\loctyc_b]{\Gamma}{\behaviour_2}
       }
       {
         \cseq
           { \Gamma }
           { b : \loctyc_b }
           { \eactor{b}{E'[\newaccept{\prole}{\ell_i(x_i) \mapsto M_i}_{i \in I}]}{\bot}{\behaviour_2} }
       }
     }
     {
       \cseq
         {\Gamma}
         {a : \loctyc_a, b : \loctyc_b }
         { \eactor{a}{E[\newconn{\ell_j}{V}{b}{\qrole}]}{\bot}{\behaviour_1} \parallel
           \eactor{b}{E'[\newaccept{\prole}{\ell_i(x_i) \mapsto M_i}_{i \in I}]}{\bot}{\behaviour_2}
         }
     }
   \end{mathpar}
   }

   By Lemma~\ref{lem:subterm-typability}, by case analysis on $F$, wee either
   have:

   \begin{mathparsmall}
     \inferrule*
     {
       \localtyconn{\qrole}{\ell_j}{A_j} \then \loctya'_j \in \{ \loctya_i \}_i \\
\vseq{\Gamma}{b}{\loctyb_b} \\ \ty{\qrole} = \loctyb_b
     }
     { \tseq[\loctyc_a]{\Gamma}{\sumact{i \in I}{\loctya_i}}{\newconn{\ell_j}{V}{b}{\qrole}}{\one}{\loctya'_j} }
   \end{mathparsmall}

   or

  \begin{mathparsmall}
    \inferrule*
    {
       \inferrule*
       {
         \localtyconn{\qrole}{\ell_j}{A_j} \then \loctya'_j \in \{ \loctya_i \}_i \\
\vseq{\Gamma}{b}{\loctyb_b} \\ \ty{\qrole} = \loctyb_b
       }
       { \tseq[\loctyc_a]{\Gamma}{\sumact{i \in I}{\loctya_i}}{\newconn{\ell_j}{V}{b}{\qrole}}{\one}{\loctya'_j} }
       \\
       \tseq[\loctyc_a]{\Gamma}{\sumact{i \in I}{\loctya_i}}{M}{\one}{\loctya'_j}
    }
    {
       \tseq[\loctyc_a]{\Gamma}{\sumact{i \in I}{\loctya_i}}{
         \trycatch
           {\newconn{\ell_j}{V}{b}{\qrole}}
           {M}
       }{\one}{\loctya'_j}
    }
   \end{mathparsmall}.

   Since the evaluation context will be discarded, WLOG we proceed assuming that
   $F = [~]$.

   As $b : \pidty{\loctyc_b} \in \Gamma$, we have that $\loctyb_b = \loctyc_b$
   and therefore that $\ty{\qrole} = \loctyc_b$.

   Again by Lemma~\ref{lem:subterm-typability},
   $
    \tseq[\loctyc_b]
         {\Gamma}
         {\sumact{i \in I}{\localtyaccept{\prole}{\ell_k}{C_k}\then{\loctyb_k}}}
         {\newaccept{\prole}{\ell_i(x_k) \mapsto M_k}_{k \in K}}
         {B'}
         {\loctyb'}
   $.

   By case analysis on $F'$, we either have:

   \begin{mathparsmall}
     \inferrule*
     {
       (\tseq[\loctyc_b]
           {\Gamma, x : C_k}{\loctyb_k}{M_k}{\tyb'}{\loctyb'})_{k \in K}
     }
     {
       \tseq[\loctyc_b]
         {\Gamma}
         {\sumact{i \in I}{\localtyaccept{\prole}{\ell_k}{C_k}\then{\loctyb_k}}}
         {\newaccept{\prole}{\ell_i(x_k) \mapsto M_k}_{k \in K}}
         {B'}
         {\loctyb'}
     }
   \end{mathparsmall}

   or

   \begin{mathparsmall}
     \inferrule*
     {
       \inferrule*
       {
         (\tseq[\loctyc_b]
             {\Gamma, x : C_k}{\loctyb_k}{M_k}{\tyb'}{\loctyb'})_{k \in K}
       }
       {
         \tseq[\loctyc_b]
           {\Gamma}
           {\sumact{i \in I}{\localtyaccept{\prole}{\ell_k}{C_k}\then{\loctyb_k}}}
           {\newaccept{\prole}{\ell_i(x_k) \mapsto M_k}_{k \in K}}
           {B'}
           {\loctyb'}
       } \\
       \tseq[\loctyc_b]
         {\Gamma}
         {\sumact{i \in I}{\localtyaccept{\prole}{\ell_k}{C_k}\then{\loctyb_k}}}
         {M}
         {A}
         {\loctyb'}
     }
     {
        \tseq[\loctyc_b]
           {\Gamma}
           {\sumact{i \in I}{\localtyaccept{\prole}{\ell_k}{C_k}\then{\loctyb_k}}}
           {\trycatch{(\newaccept{\prole}{\ell_i(x_k) \mapsto M_k}_{k \in K})}{M}}
           {B'}
           {\loctyb'}
     }
   \end{mathparsmall}

   Again, we consider the first case.

   Thus, $\loctyc_b = \newaccept{\prole}{\ell_i(x_k) \mapsto M_k}_{k \in K}$.

   We now need to introduce the new runtime typing environment for the new session.
   We begin with a singleton runtime typing environment $\{
   \sessindexroles{s}{\prole}{\emptyset}{\loctyc_a} \}$, and recall that
   $\loctyc_a = \{ \sumact{i \in I}{\loctya_i} \}$ and
   $\localtyconn{\qrole}{\ell_j}{A_j} \then \loctya'_j \in \{ \loctya_i \}_i$.

   Since $\config{C}$ is defined wrt.\ a well-formed program and thus protocols
   are well-formed, $\prole$ is a unique initiator, and we know that $\loctyc_a =
   \ty{\prole}$ and $\safe{ \{ \sessindexroles{s}{\prole}{\emptyset}{\loctyc_a} \}}$.

   As a consequence of safety, we know that $C_j = A_j$.

   We can then show a reduction on typing environments:

   \begin{mathparsmall}
     \inferrule*
     {
       \inferrule*
       {
         \exists j \in I . \loctya_j = \localtyconn{\qrole}{\ell_j}{\tya_j} \then \loctya'_j \\
         \ty{\qrole} = \localtyaccept{\prole}{\ell_k}{C_k}\then{\loctyb_k} \\
         j \in K \\
         C_j = A_j
       }
       { \sessindexroles{s}{\prole}{\emptyset}{\sumact{i \in I}{S_i}}
         \syncannarrow{\ltsconnpair{s}{\prole}{\qrole}{\ell_j}}
         \sessindexroles{s}{\prole}{\qrole}{\loctya'_j},
         \sessindexroles{s}{\qrole}{\prole}{\loctyb_j}
       }
     }
     {
        \sessindexroles{s}{\prole}{\emptyset}{\sumact{i \in I}{S_i}}
         \syncannarrow{\ltsconnpair{s}{\prole}{\qrole}{\ell_j}}
         \sessindexroles{s}{\prole}{\qrole}{\loctya'_j},
         \sessindexroles{s}{\qrole}{\prole}{\loctyb_j}
     }
   \end{mathparsmall}

   Since
   $
       \sessindexroles{s}{\prole}{\emptyset}{\sumact{i \in I}{S_i}}
       \ceval
       \sessindexroles{s}{\prole}{\qrole}{\loctya'_j},
       \sessindexroles{s}{\qrole}{\prole}{\loctyb_j}
   $, it follows by safety that
   $\safe{\sessindexroles{s}{\prole}{\qrole}{\loctya'_j},
   \sessindexroles{s}{\qrole}{\prole}{\loctyb_j}}$.

   Noting that only top-level frames (i.e., $F, F'$) can contain exception-handling frames,
     $E, E'$ are pure.
     Thus, we can show
     $\tseq[\loctyc_a]{\Gamma}{\loctya'_j}{\effreturn{()}}{\tya}{\loctya'_j}$ and so
     by Lemma~\ref{lem:pure-replacement},
     $\tseq[\loctyc_a]{\Gamma}{\loctya'_j}{E[\effreturn{()}]}{\one}{\localend}$.

     Similarly, we can show
     $
     \tseq[\loctyc_b]
         {\Gamma}
         {\loctyb_j}
         {
           M_j \{ V / x_j \}
         }
         {B'}
         {\loctyb'}
     $
     and so by Lemma~\ref{lem:pure-replacement},
     $\tseq[\loctyc_b]{\Gamma}{\loctyb_j}{E'[M_j \{ V / x_j \}]}{\tyb}{\localend}$.

     Recomposing:

 {\small
   \begin{mathparsmall}
     \inferrule*
     {
       \inferrule*
       {
         \inferrule*
         {
           a : \pidty{\loctyc_a} \in \Gamma \\\\
           \tseq[\loctyc_a]{\Gamma}{\loctya'_j}{E[\effreturn{()}]}{A}{\localend}
           \\
           \bseq[\loctyc_a]{\Gamma}{\behaviour_1}
         }
         { \cseq
             { \Gamma }
             { a : \loctyc_a, \sessindexroles{s}{\prole}{\qrole}{\loctya'_j} }
             {
               \eactor{a}{E[\effreturn{()}]}{\newconnsess{s}{\prole}{\qrole}}{\behaviour_1}
             }
         }
         \\
         \inferrule*
         {
           b : \pidty{\loctyc_b} \in \Gamma \\\\
           \tseq[\loctyc_b]{\Gamma}{\loctyb_j}{E'[M_j \{ V / x_j \}]}{B}{\localend}
           \\
           \bseq[\loctyc_b]{\Gamma}{\behaviour_2}
         }
         {
           \cseq
             { \Gamma }
             { b : \loctyc_b, \sessindexroles{s}{\prole}{\prole}{\loctyb_j} }
             {
               \eactor{b}{E'[M_j \{ V / x_j \}]}{\newconnsess{s}{\qrole}{\prole}}{\behaviour_2}
             }
         }
       }
       {
         \cseq
           {\Gamma}
           {
            a : \loctyc_a, b : \loctyc_b,
             \sessindexroles{s}{\prole}{\qrole}{\loctya'_j},
             \sessindexroles{s}{\qrole}{\prole}{\loctyb_j}
           }
           {
             \eactor{a}{E[\effreturn{()}]}{\newconnsess{s}{\prole}{\qrole}}{\behaviour_1} \parallel
             \eactor{b}{E'[M_j \{ V / x_j
             \}]}{\newconnsess{s}{\qrole}{\prole}{\loctyb_j}}{\behaviour_2}
           }
       }
     }
     {
       \cseq{\Gamma}{a : \loctyc_a, b : \loctyc_b}{(\nu s)(
         \eactor{a}{E[\effreturn{()}]}{\newconnsess{s}{\prole}{\qrole}}{\behaviour_1} \parallel
           \eactor{b}{E'[M_j \{ V / x_j
           \}]}{\newconnsess{s}{\qrole}{\prole}}{\behaviour_2}
       )}
     }
   \end{mathparsmall}
   }

   as required.
  \end{proofcase}

  Note: in the remaining communication cases, we consider the case where the top-level
  frame is empty, since the frame will be discarded in the result.

  \begin{proofcase}{E-Conn}

   \[
     \bl
     \eactor{a}{E[\newconn{\ell_j}{V}{b}{\qrole}]}{\newconnsess{s}{\prole}{\rolesetr}}{\behaviour_1} \parallel
     \eactor{b}{E'[\newaccept{\prole}{\ell_i(x_i) \mapsto M_i}_{i \in
     I}]}{\bot}{\behaviour_2} \ceval \\
\eactor{a}{E[\effreturn{()}]}{\newconnsess{s}{\prole}{\rolesetr, \qrole}}{\behaviour_1} \parallel
     \eactor{b}{E'[M_j \{ V / x_i \}]}{\newconnsess{s}{\qrole}{\prole}}{\behaviour_2}
     \el
   \]

   with $j \in K$.

   Assumption:

   {\small
   \begin{mathparsmall}
     \inferrule*
     {
       \inferrule*
       {
         a : \pidty{\loctyc_a} \in \Gamma \\\\
         \tseq[\loctyc_a]{\Gamma}{\loctya}{E[\newconn{\ell_j}{V}{b}{\qrole}]}{A}{\localend}
         \\
         \bseq[\loctyc_a]{\Gamma}{\behaviour_1}
       }
       { \cseq
           { \Gamma }
           { a : \loctyc_a,  \sessindexroles{s}{\prole}{\rolesetr}{\loctya} }
           { \eactor
               {a}
               {E[\newconn{\ell_j}{V}{b}{\qrole}]}
               {\newconnsess{s}{\prole}{\rolesetr}}
               {\behaviour_1}
           }
       }
       \\
       \inferrule*
       {
         b : \pidty{\loctyc_b} \in \Gamma \\
         \loctyb = \loctyc_b \vee \loctyb = \localend \\\\
         \tseq[\loctyc_b]{\Gamma}{\loctyb}
         {E'[\newaccept{\prole}{\ell_i(x_i) \mapsto M_i}_{i \in I}]}{B}{\localend}
         \\\\
         \bseq[\loctyc_b]{\Gamma}{\behaviour_2}
       }
       {
         \cseq
           { \Gamma }
           { b : \loctyc_b }
           { \eactor{b}{E'[\newaccept{\prole}{\ell_i(x_i) \mapsto M_i}_{i \in
           I}]}{\bot}{\behaviour_2} }
       }
     }
     {
       \cseq
         {\Gamma}
         {a : \loctyc_a, b : \loctyc_b \sessindexroles{s}{\prole}{\rolesetr}{\loctya}}
         { \eactor
             {a}
             {E[\newconn{\ell_j}{V}{b}{\qrole}]}
             {\newconnsess{s}{\prole}{\rolesetr}}
             {\behaviour_1}
           \parallel
             \eactor{b}{E'[\newaccept{\prole}{\ell_i(x_i) \mapsto M_i}_{i \in
             I}]}{\bot}{\behaviour_2}
         }
     }
   \end{mathparsmall}
   }

   where $j \in I$ and $\safe{ a : \loctyc_a, b : \loctyc_b,
   \sessindexroles{s}{\prole}{\rolesetr}{\loctya}}$.

   Consider the subderivation
   $\tseq[\loctyc_a]{\Gamma}{\loctya}{E[\newconn{\ell_j}{V}{b}{\qrole}]}{A}{\localend}$.

   By Lemma~\ref{lem:subterm-typability}:

   \begin{mathparsmall}
     \inferrule*
     {
       \localtyconn{\qrole}{\ell_j}{A_j} \then \loctya'_j \in \{ \loctya_i \}_i \\
\vseq{\Gamma}{b}{\pidty{\loctyb_b}} \\ \ty{\qrole} = \loctyb_b
     }
     { \tseq[\loctyc_a]{\Gamma}{\sumact{i \in I}{\loctya_i}}{\newconn{\ell_j}{V}{b}{\qrole}}{\one}{\loctya'_j} }
   \end{mathparsmall}

   Since $b : \pidty{\loctyc_b} \in \Gamma$, we have that $\loctyb_b =
   \loctyc_b$. We also deduce that $\loctya = \sumact{i \in I}{\loctya_i}$
   where $\localtyconn{\qrole}{\ell_j}{A_j} \then \loctya'_j \in \{ \loctya_i \}_i$.

   Also by Lemma~\ref{lem:subterm-typability}:
   \begin{mathparsmall}
     \inferrule*
     {
       (\tseq[\loctyc_b]
           {\Gamma, x : C_k}{\loctyb_k}{M_k}{\tyb'}{\loctyb'})_{k \in K}
     }
     {
       \tseq[\loctyc_b]
         {\Gamma}
         {\sumact{i \in I}{\localtyaccept{\prole}{\ell_k}{C_k}\then{\loctyb_k}}}
         {\newaccept{\prole}{\ell_i(x_k) \mapsto M_k}_{k \in K}}
         {B'}
         {\loctyb'}
     }
   \end{mathparsmall}

   Thus, $\loctyc_b = \newaccept{\prole}{\ell_i(x_k) \mapsto M_k}_{k \in K}$.

   We can then show a reduction on typing environments:

   \begin{mathparsmall}
     \inferrule*
     {
       \inferrule*
       {
         \exists j \in I . \loctya_j = \localtyconn{\qrole}{\ell_j}{\tya_j} \then \loctya'_j \\
         \ty{\qrole} = \localtyaccept{\prole}{\ell_k}{C_k}\then{\loctyb_k} \\
         j \in K \\
         C_j = A_j
       }
       { \sessindexroles{s}{\prole}{\rolesetr}{\sumact{i \in I}{S_i}}
         \syncannarrow{\ltsconnpair{s}{\prole}{\qrole}{\ell_j}}
         \sessindexroles{s}{\prole}{\rolesetr, \qrole}{\loctya'_j},
         \sessindexroles{s}{\qrole}{\prole}{\loctyb_j}
       }
     }
     {
         \sessindexroles{s}{\prole}{\rolesetr}{\sumact{i \in I}{S_i}}
         \syncannarrow{\ltsconnpair{s}{\prole}{\qrole}{\ell_j}}
         \sessindexroles{s}{\prole}{\rolesetr, \qrole}{\loctya'_j},
         \sessindexroles{s}{\qrole}{\prole}{\loctyb_j}
     }
   \end{mathparsmall}

   (noting that as a consequence of safety, we know that $C_j = A_j$).

     Since
   $
       \sessindexroles{s}{\prole}{\rolesetr}{\sumact{i \in I}{S_i}}
       \ceval
       \sessindexroles{s}{\prole}{\rolesetr, \qrole}{\loctya'_j},
       \sessindexroles{s}{\qrole}{\prole}{\loctyb_j}
   $, it follows by safety that
   $\safe{\sessindexroles{s}{\prole}{\rolesetr, \qrole}{\loctya'_j},
   \sessindexroles{s}{\qrole}{\prole}{\loctyb_j}}$.

     Noting that only top-level frames (i.e., $F, F'$) can contain exception-handling frames,
     $E, E'$ are pure.
     We can show
     $\tseq[\loctyc_a]{\Gamma}{\loctya'_j}{\effreturn{()}}{\tya}{\loctya'_j}$ and so
     by Lemma~\ref{lem:pure-replacement},
     $\tseq[\loctyc_a]{\Gamma}{\loctya'_j}{E[\effreturn{()}]}{\one}{\localend}$.

     Similarly, we can show
     $
     \tseq[\loctyc_b]
         {\Gamma}
         {\loctyb_j}
         {
           M_j \{ V / x_j \}
         }
         {B'}
         {\loctyb'}
     $
     and so by Lemma~\ref{lem:pure-replacement},
     $\tseq[\loctyc_b]{\Gamma}{\loctyb_j}{E'[M_j \{ V / x_j \}]}{\tyb}{\localend}$.

     Recomposing:

   {\small
   \begin{mathparsmall}
     \inferrule*
     {
       \inferrule*
       {
         a : \pidty{\loctyc_a} \in \Gamma \\\\
         \tseq[\loctyc_a]{\Gamma}{\loctya'_j}{E[\newconn{\ell_j}{V}{b}{\qrole}]}{A}{\localend}
         \\
         \bseq[\loctyc_a]{\Gamma}{\behaviour_1}
       }
       { \cseq
           { \Gamma }
           { a : \loctyc_a, \sessindexroles{s}{\prole}{\rolesetr, \qrole}{\loctya'_j} }
           { \eactor
               {a}
               {E[\effreturn{()}]}
               {\newconnsess{s}{\prole}{\rolesetr, \qrole}}
               {\behaviour_1}
           }
       }
       \\
       \inferrule*
       {
         b : \pidty{\loctyc_b} \in \Gamma \\
         \tseq[\loctyc_b]{\Gamma}{\loctyb_j}{E'[M_j \{ V / x_j \}]}{\tyb}{\localend}
         \\\\
         \bseq[\loctyc_b]{\Gamma}{\behaviour_2}
       }
       {
         \cseq
           { \Gamma }
           { b : \loctyc_b, \sessindexroles{s}{\qrole}{\prole}{\loctyb_j} }
           { \eactor{b}{E'[M_j \{ V / x_j \}]}{\newconnsess{s}{\qrole}{\prole}}{\behaviour_2} }
       }
     }
     {
       \cseq
         {\Gamma}
         {a : \loctyc_a, b : \loctyc_b, \sessindexroles{s}{\prole}{\rolesetr,
         \qrole}{\loctya'_j}, \sessindexroles{s}{\qrole}{\prole}{\loctyb_j}}
         { \eactor
             {a}
             {E[\effreturn{()}]}
             {\newconnsess{s}{\prole}{\rolesetr, \qrole}}
             {\behaviour_1}
           \parallel
           \eactor{b}{E'[M_j \{ V / x_j \}]}{\newconnsess{s}{\qrole}{\prole}}{\behaviour_2}
         }
     }
   \end{mathparsmall}
   }

   with $\safe{a : \loctyc_a, b : \loctyc_b,  \sessindexroles{s}{\prole}{\rolesetr, \qrole}{\loctya'_j},
   \sessindexroles{s}{\qrole}{\prole}{\loctyb_j}}$
   as required.
  \end{proofcase}

  \begin{proofcase}{E-Comm}

    Let $\rtenv_1 =  a : \loctyc_a, \sessindexroles{s}{\prole}{\rolesetr}{\loctya_a} $ and
    $\rtenv_2 =  b : \loctyc_b, \sessindexroles{s}{\qrole}{\rolesets}{\loctya_b}$.

    Assumption:

    {\small
    \begin{mathparsmall}
      \inferrule*
      {
        \inferrule*
        { a : \pidty{\loctyc_a} \in \Gamma \\\\
          \tseq[\loctyc_a]{\Gamma}{\loctya_a}{E[\newsend{\ell_j}{V}{\qrole}]}{A}{\localend}
          \\\\
          \bseq[\loctyc_a]{\Gamma}{\behaviour_1}
        }
        {
          \cseq
            { \Gamma }
            { \rtenv_1 }
            {
              \eactor{a}{E[\newsend{\ell_j}{V}{\qrole}]}{\newconnsess{s}{\prole}{\rolesetr}}{\behaviour_1}
            }
        }
        \\
        \inferrule*
        {
          b : \pidty{\loctyc_b} \in \Gamma \\\\
          \tseq[\loctyc_b]{\Gamma}{S_b}{E'[\newrecv{\prole}{\ell_i(x_i)
          \mapsto M_i}_{i \in I}]}{B}{\localend} \\\\
          \bseq[\loctyc_b]{\Gamma}{\behaviour_2}
        }
        {
          \cseq
            { \Gamma }
            { \rtenv_2 }
            {
              \eactor{b}{E'[\newrecv{\prole}{\ell_i(x_i) \mapsto M_i}_{i \in
              I}]}{\newconnsess{s}{\qrole}{\rolesets}}{\behaviour_2}
            }
        }
      }
      {
        \cseq
          {\Gamma}
          { \rtenv_1, \rtenv_2 }
          {
            \eactor{a}{E[\newsend{\ell_j}{V}{\qrole}]}{\newconnsess{s}{\prole}{\rolesetr}}{\behaviour_1} \parallel
            \eactor{b}{E'[\newrecv{\prole}{\ell_i(x_i) \mapsto M_i}_{i \in
            I}]}{\newconnsess{s}{\qrole}{\rolesets}}{\behaviour_2}
          }
      }
    \end{mathparsmall}
  }

  where $j \in I$ and $\safe{\rtenv_1, \rtenv_2}$.

  Consider the subderivation
  $\tseq[\loctyc_a]{\Gamma}{S_a}{E[\newsend{\ell_j}{V}{\qrole}]}{A}{\localend}$.
  By Lemma~\ref{lem:subterm-typability}:

  \begin{mathparsmall}
    \inferrule*
    { \localtysend{\qrole}{\ell_j}{A_j} \in \{ S_i\}_{i \in I} \\
      \vseq{\Gamma}{V}{A_j}
    }
    { \tseq[\loctyc_a]{\Gamma}{\sumact{i \in I}{S_i}}{\newsend{\ell_j}{V}{\qrole}}{\one}{S_j}  }
  \end{mathparsmall}

  Next, consider the subderivation
  $\tseq[\loctyc_b]{\Gamma}{S_b}{E'[\newrecv{\prole}{\ell_i(x_i)
  \mapsto M_i}_{i \in I}]}{B}{\localend}$.
  Again by Lemma~\ref{lem:subterm-typability}:
  \begin{mathparsmall}
    \inferrule*
    {
      (\tseq[\loctyc_b]{\Gamma, x_k : B_k}{T_k}{M_i}{B}{T'})_{k \in K}
    }
    { \tseq[\loctyc_b]
        {\Gamma}
        {\sumact{k \in K}{\localtyrecv{\prole}{\ell_k}{B_k}\then T_k}}
        {\newrecv{\prole}{\ell_i(x_i) \mapsto M_i}_{i \in I}}
        {A}
        {T'}
    }
  \end{mathparsmall}

  From the assumption we know that $j \in I$. As a result of the two
  subderivations above, we can refine our definitions of $\rtenv_1$ and
  $\rtenv_2$:

  \begin{itemize}
    \item $\rtenv_1 =  a : \loctyc_a, \sessindexroles{s}{\prole}{\rolesetr}{\sumact{i
      \in I}{S_i}}$, where $\localtysend{\qrole}{\ell_j}{A_j}\then S'_j \in \{ S_i\}_{i \in I}$
    \item $\rtenv_2 =  b : \loctyc_b,
      \sessindexroles{s}{\qrole}{\rolesets}{\sumact{k \in K}{\localtyrecv{\prole}{\ell_k}{B_k}\then T_k}}$
  \end{itemize}

  Since $\safe{\rtenv_1, \rtenv_2}$, we have that $j \in K$, $B_j = A_j$.

  Thus we can construct a reduction on typing environments:

  \begin{mathparsmall}
    \inferrule*
    {
      \inferrule*
      {
        \inferrule*
        { \localtysend{s}{\ell_j}{A_j} \then S'_j \in \{ S_i\}_i }
        { \sessindexroles{s}{\prole}{\rolesetr}{\sumact{i \in I}{S_i}}
            \annarrow{\ltscomm[!]{s}{\prole}{\qrole}{\ell_j}{A_j}}
          \sessindexroles{s}{\prole}{\rolesetr}{S'_j}
        }
        \\
        \inferrule*
        { j \in K }
        { \sessindexroles{s}{\qrole}{\rolesets}{\sumact{k \in K}{\localtyrecv{\prole}{\ell_k}{B_k}\then T_k}}
          \annarrow{\ltscomm[?]{s}{\qrole}{\prole}{\ell_j}{A_j}}
          \sessindexroles{s}{\qrole}{\rolesets}{T_j}
        }
      }
      {
        \sessindexroles{s}{\prole}{\rolesetr}{\sumact{i \in I}{S_i}},
        \sessindexroles{s}{\qrole}{\rolesets}{\sumact{k \in K}{\localtyrecv{\prole}{\ell_k}{B_k}\then T_k}}
        \syncannarrow{\ltspair{s}{\prole}{\qrole}{\ell_j}}
        \sessindexroles{s}{\prole}{\rolesetr}{S'_j}, \sessindexroles{s}{\qrole}{\rolesets}{T_j}
      }
    }
    { a: \loctyc_a, b : \loctyc_b,
      \sessindexroles{s}{\prole}{\rolesetr}{\sumact{i \in I}{S_i}},
      \sessindexroles{s}{\qrole}{\rolesets}{\sumact{k \in K}{\localtyrecv{\prole}{\ell_k}{B_k}\then T_k}}
        \syncannarrow{\ltspair{s}{\prole}{\qrole}{\ell_j}}
      a: \loctyc_a, b : \loctyc_b, \sessindexroles{s}{\prole}{S'_j}, \sessindexroles{s}{\qrole}{T_j}
    }
  \end{mathparsmall}

  Let $\rtenv' = a: \loctyc_a, b : \loctyc_b, \sessindexroles{s}{\prole}{S'_j},
  \sessindexroles{s}{\qrole}{T_j}$.

  Since $\safe{\rtenv_1, \rtenv_2}$ and $\rtenv_1, \rtenv_2 \ceval \rtenv'$, by the definition of safety, $\safe{\rtenv'}$.

   Noting that only top-level frames (i.e., $F, F'$) can contain exception-handling frames,
     $E, E'$ are pure.
  By Lemma~\ref{lem:pure-replacement},
  $\tseq[\loctyc_a]{\Gamma}{S'_j}{E[\effreturn{()}]}{A}{\localend}$.

  By Lemma~\ref{lem:substitution}, $\tseq[\loctyc_b]{\Gamma}{T_j}{M_j \{ V / x_j \}}{B}{T'}$.

  By Lemma~\ref{lem:pure-replacement}, $\tseq[\loctyc_b]{\Gamma}{T_j}{E'[M_j \{ V / x_j \}]}{B}{\localend}$.

  Letting $\rtenv'_1 = a: \loctyc_a, \sessindexroles{s}{\prole}{S'_j}$ and
  $\rtenv'_2 = b : \loctyc_b, \sessindexroles{s}{\qrole}{T_j}$, recomposing:

    {\small
    \begin{mathparsmall}
      \inferrule*
      {
        \inferrule*
        { a : \pidty{\loctyc_a} \in \Gamma \\\\
          \tseq[\loctyc_a]{\Gamma}{S_j}{E[\effreturn{()}]}{A}{\localend}
          \\\\
          \bseq[\loctyc_a]{\Gamma}{\behaviour_1}
        }
        {
          \cseq
            { \Gamma }
            { \rtenv'_1 }
            {
              \eactor{a}{E[\effreturn{()}]}{\newconnsess{s}{\prole}{\rolesetr}}{\behaviour_1}
            }
        }
        \\
        \inferrule*
        {
          b : \pidty{\loctyc_b} \in \Gamma \\\\
          \tseq[\loctyc_b]{\Gamma}{T_j}{E'[M_j \{ V / x_j \}]}{B}{\localend} \\\\
          \bseq[\loctyc_b]{\Gamma}{\behaviour_2}
        }
        {
          \cseq
            { \Gamma }
            { \rtenv'_2 }
            {
              \eactor{b}{E'[M_j \{ V / x_j
              \}]}{\newconnsess{s}{\qrole}{\rolesets}}{\behaviour_2}
            }
        }
      }
      {
        \cseq
          {\Gamma}
          { \rtenv' }
          {
            \eactor{a}{E[\effreturn{()}]}{\newconnsess{s}{\prole}{\rolesetr}}{\behaviour_1} \parallel
            \eactor{b}{E'[M_j \{ V / x_j \}]}{\newconnsess{s}{\qrole}{\rolesets}}{\behaviour_2}
          }
      }
    \end{mathparsmall}
  }

  with $\safe{\rtenv'}$, as required.
  \end{proofcase}

  \begin{proofcase}{E-Disconn}

    \[
      \eactor{a}{E[\newwait{\qrole}]}{\newconnsess{s}{\prole}{\rolesetr, \qrole}}{\behaviour_1}
        \parallel \eactor{b}{E'[\newdisconn{\prole}]}{\newconnsess{s}{\qrole}{\prole}}{\behaviour_2}
      \ceval
      \eactor{a}{E[\effreturn{()}]}{\newconnsess{s}{\prole}{\rolesetr}}{\behaviour_1} \parallel
        \eactor{b}{E'[\effreturn{()}]}{\bot}{\behaviour_2}
    \]

    Let $\rtenv = a : \loctyc_a, \sessindexroles{s}{\prole}{\rolesetr}{\loctya},
    b : \loctyc_b, \sessindexroles{s}{\qrole}{\prole}{\loctyb}$.

    Assumption:
    {\small
    \begin{mathparsmall}
      \inferrule*
      {
        \inferrule*
        {
          \pidty{\loctyc_a} \in \Gamma \\\\
          \tseq[\loctyc_a]{\Gamma}{\loctya}{E[\newwait{\qrole}]}{B}{\localend} \\
          \bseq[\loctyc_a]{\Gamma}{\behaviour_1}
        }
        {
          \cseq
            {\Gamma}
            { a : \loctyc_a, \sessindexroles{s}{\prole}{\rolesetr,
            \qrole}{\loctya} }
            { \eactor{a}{E[\newwait{\qrole}]}{\newconnsess{s}{\prole}{\rolesetr,
            \qrole}}{\behaviour_1} }
        }
        \\
        \inferrule*
        {
          \pidty{\loctyc_b} \in \Gamma \\\\
          \tseq[\loctyc_b]{\Gamma}{\loctyb}{E'[\newdisconn{\prole}]}{B}{\localend} \\
          \bseq[\loctyc_b]{\Gamma}{\behaviour_2}
        }
        {
          \cseq
            {\Gamma}
            { b : \loctyc_b, \sessindexroles{s}{\qrole}{\prole}{\loctyb} }
            { \eactor{b}{E'[\newdisconn{\prole}]}{\newconnsess{s}{\qrole}{\prole}}{\behaviour_2} }
        }
      }
      {
        \cseq
          { \Gamma }
          { \rtenv }
          { \eactor{a}{E[\newwait{\qrole}]}{\newconnsess{s}{\prole}{\rolesetr,
            \qrole}}{\behaviour_1}
            \parallel \eactor{b}{E'[\newdisconn{\prole}]}{\newconnsess{s}{\qrole}{\prole}}{\behaviour_2}}
      }
    \end{mathparsmall}
  }

  By Lemma~\ref{lem:subterm-typability}:
  \begin{mathparsmall}
    \inferrule*
    { }
    { \tseq[\loctyc_a]{\Gamma}{\localtywait{\qrole}\then \locty'}{\newwait{\qrole}}{\one}{\locty'} }
  \end{mathparsmall}

  Also by Lemma~\ref{lem:subterm-typability}:
  \begin{mathparsmall}
    \inferrule*
    { }
    { \tseq[\loctyc_b]{\Gamma}{\localtydisconn{\prole}}{\newdisconn{\prole}}{\one}{\loctyb'} }
  \end{mathparsmall}

  Thus, $\locty = \localtywait{\qrole}\then \locty'$ and
  $\loctyb = \localtydisconn{\qrole}$

  We can show a reduction on runtime typing environments:

  \begin{mathparsmall}
    \inferrule*
    {
      \inferrule*
      {
        \inferrule*
        { }
        {
          \sessindexroles{s}{\prole}{\rolesetr, \qrole}{\localtywait{\qrole}\then \locty'}
          \annarrow{\ltswait{s}{\prole}{\qrole}}
          \sessindexroles{s}{\prole}{\rolesetr}{\locty'}
        }
        \\
        \inferrule*
        { }
        {
          \sessindexroles{s}{\qrole}{\prole}{\localtydisconn{\prole}}
          \annarrow{\ltsdisconn{s}{\qrole}{\prole}}
          \cdot
        }
      }
      {
        \sessindexroles{s}{\prole}{\rolesetr, \qrole}{\localtywait{\qrole}\then \locty'},
        \sessindexroles{s}{\qrole}{\prole}{\localtydisconn{\prole}}
        \syncannarrow{\ltsdisconnpair{s}{\prole}{\qrole}}
        \locty'
      }
    }
    {
    a : \loctyc_a,
      \sessindexroles{s}{\prole}{\rolesetr, \qrole}{\localtywait{\qrole}\then \locty'},
      b : \loctyc_b,
      \sessindexroles{s}{\qrole}{\prole}{\localtydisconn{\qrole}}
\syncannarrow{\ltsdisconnpair{s}{\prole}{\qrole}}
a : \loctyc_a,
      \sessindexroles{s}{\prole}{\rolesetr}{\locty'},
      b : \loctyc_b
    }
  \end{mathparsmall}

  Since $\safe{\rtenv}$, and $\rtenv \synceval \rtenv'$, by the definition of safety, we have that
  $\safe{\rtenv'}$.

   Noting that only top-level frames (i.e., $F, F'$) can contain exception-handling frames,
     $E, E'$ are pure.
  We can show that $\tseq[\loctyc_a]{\Gamma}{\locty'}{\effreturn{()}}{\one}{\locty'}$, so
  by Lemma~\ref{lem:pure-replacement}, we have that
  $\tseq[\loctyc_a]{\Gamma}{\locty'}{E[\effreturn{()}]}{\tya}{\locty'}$.

  By the same argument,
  $\tseq[\loctyc_b]{\Gamma}{\localend}{E'[\effreturn{()}]}{\one}{\localend}$.

  Thus we can show (by \textsc{T-UnconnectedActor}, noting that $\localend$ is a
  permissible precondition):

  \begin{mathparsmall}
    \inferrule*
      {
        \pidty{\loctyc_a} \in \Gamma \\\\
        \tseq[\loctyc_a]{\Gamma}{\localend}{E'[\effreturn{()}]}{B}{\localend} \\
        \bseq[\loctyc_a]{\Gamma}{\behaviour_2}
      }
      {
        \cseq
          {\Gamma}
          { b : \loctyc_a }
          { \eactor{b}{E'[\effreturn{()}]}{\bot}{\behaviour_2} }
      }
  \end{mathparsmall}

  Recomposing:

  {\small
  \begin{mathparsmall}
    \inferrule*
    {
        \inferrule*
        {
          a : \pidty{\loctyc_a} \in \Gamma \\\\
          \tseq{\Gamma}{\loctya}{E[\newwait{\qrole}]}{B}{\localend} \\
          \bseq{\Gamma}{\behaviour_1}
        }
        {
          \cseq
            {\Gamma}
            { a : \loctyc_a, \sessindexroles{s}{\prole}{\rolesetr}{\loctya'} }
            { \eactor{a}{E[\effreturn{()}]}{\newconnsess{s}{\prole}{\rolesetr}}{\behaviour_1} }
        }
        \\
      \inferrule*
      {
        b : \loctyc_b \in \Gamma \\
        \tseq[\loctyc_b]{\Gamma}{\localend}{E'[\effreturn{}()]}{B}{\localend} \\
        \bseq[\loctyc_b]{\Gamma}{\behaviour_2}
      }
      {
        \cseq
          {\Gamma}
          { b : \loctyc_b }
          { \eactor{b}{E'[\effreturn{()}]}{\bot}{\behaviour_2} }
      }
    }
    { \cseq
        {\Gamma}
        { a : \loctyc_a, \sessindexroles{s}{\prole}{\rolesetr}{\loctya'}, b: \loctyc_b  }
        {\eactor{a}{E[\effreturn{()}]}{\newconnsess{s}{\prole}{\rolesetr}}{\behaviour_1}
          \parallel
         \eactor{b}{E'[\effreturn{()}]}{\bot}{\behaviour_2}
        }
    }
  \end{mathparsmall}
  }

  with $\safe{a : \loctyc_a, \sessindexroles{s}{\prole}{\rolesetr}{\loctya'}, b: \loctyc_b}$, as required.
  \end{proofcase}

  \begin{proofcase}{E-Complete}

    \[
(\nu s)(\eactor{a}{\effreturn{V}}{\newconnsess{s}{\prole}{\emptyset}}{\behaviour})
      \ceval      \eactor{a}{\effreturn{V}}{\bot}{\behaviour}
    \]

    Assumption:

    \begin{mathpar}
      \inferrule*
      {
        \inferrule*
        {
          a : \pidty{\loctyb} \\
          \tseq{\Gamma}{\loctya}{\effreturn{V}}{A}{\localend} \\
          \bseq{\Gamma}{\behaviour}
        }
        {
          \cseq
            {\Gamma}
            { a : \loctyb, \sessindexroles{s}{\prole}{\emptyset}{\loctya} }
            {
              \eactor
                {a}
                {\effreturn{V}}
                {\newconnsess{s}{\prole}{\emptyset}}
                {\behaviour}
            }
        }
      }
      {
        \cseq
          {\Gamma}
          {\cdot}
          {
            (\nu s)(\eactor{a}{\effreturn{V}}{\newconnsess{s}{\prole}{\emptyset}}{\behaviour})
          }
      }
    \end{mathpar}

    Since by \textsc{T-Return}, the pre- and post-conditions must match,
    it must be the case that $\loctya = \localend$.

    Thus we can show:

    \begin{mathpar}
      \inferrule*
      {
        a : \pidty{\loctyb} \\
        \tseq{\Gamma}{\localend}{\effreturn{V}}{A}{\localend} \\
        \bseq{\Gamma}{\behaviour}
      }
      {
        \cseq
          {\Gamma}
          { a : \loctyb }
          {
            \eactor
              {a}
              {\effreturn{V}}
              {\newconnsess{s}{\prole}{\emptyset}}
              {\behaviour}
          }
      }
    \end{mathpar}

    as required.
  \end{proofcase}

  \begin{proofcase}{E-CommRaise}
    \[
      \eactor{a}{E[M]}{\newconnsess{s}{\prole}{\rolesetr}}{\behaviour} \parallel
      \zap{\sessindextwo{s}{\qrole}}
        \ceval
      \eactor{a}{E[\raiseexn]}{\newconnsess{s}{\prole}{\rolesetr}}{\behaviour}
      \parallel \zap{\sessindextwo{s}{\qrole}}
    \]

    where $\subj{M} = \prole$.

    The proof is by cases on $M$, where $M$ must be a communication action.
    The cases have the same structure, so we show the case where $M =
    \newwait{\qrole}$.

    Assumption:

    \begin{mathpar}
      \inferrule*
      {
        \inferrule*
        {
          \tseq{\Gamma}{\localtywait{\qrole} \then
          \loctya'}{E[\newwait{\qrole}]}{A}{\localend} \\
          \bseq{\Gamma}{\behaviour}
        }
        {
        \cseq
          {\Gamma}
          {a : \loctyb,
            \sessindexroles{s}{\prole}{\rolesetr}{\localtywait{\qrole}\then
          \loctya'}}
          {\eactor{a}{E[\newwait{\qrole}]}{\newconnsess{s}{\prole}{\rolesetr}}{\behaviour}}
        }
        \\
        \inferrule*
        { }
        {
          \cseq
            {\Gamma}
            {\sessindexroles{s}{\qrole}{\rolesets}{\loctyc}}
            {\zaptwo{s}{\qrole}}
        }
      }
      {
        \cseq
          {\Gamma}
          { a : \loctyb,
            \sessindexroles{s}{\prole}{\rolesetr}{\localtywait{\qrole}\then \loctya'}
            \sessindexroles{s}{\qrole}{\rolesets}{\loctyc}
          }
          { \eactor
              {a}
              {E[\newwait{\qrole}]}
              {\newconnsess{s}{\prole}{\rolesetr}}
              {\behaviour}
            \parallel
            \zap{\sessindextwo{s}{\qrole}}
          }
      }
    \end{mathpar}

    By Lemma~\ref{lem:subterm-typability}, there exists some $\loctya''$ such
    that:
    \begin{mathpar}
      \inferrule*
      { }
      {
        \tseq{\Gamma}{\localtywait{\qrole} \then
          \loctya'}{\newwait{\qrole}}{\one}{\loctya''}
      }
    \end{mathpar}

    By \textsc{T-Raise}, $\raiseexn$ can have any precondition, return type,
    and postcondition. Therefore, by Lemma~\ref{lem:pure-replacement},
    $\tseq{\Gamma}{\localtywait{\qrole}}{E[\raiseexn]}{A}{\localend}$, and
    therefore:

    \begin{mathpar}
      \inferrule*
      {
        \inferrule*
        {
          \tseq{\Gamma}{\localtywait{\qrole} \then
          \loctya'}{E[\raiseexn]}{A}{\localend} \\
          \bseq{\Gamma}{\behaviour}
        }
        {
        \cseq
          {\Gamma}
          {a : \loctyb,
            \sessindexroles{s}{\prole}{\rolesetr}{\localtywait{\qrole}\then
          \loctya'}}
          {\eactor{a}{E[\newwait{\qrole}]}{\newconnsess{s}{\prole}{\rolesetr}}{\behaviour}}
        }
        \\
        \inferrule*
        { }
        {
          \cseq
            {\Gamma}
            {\sessindexroles{s}{\qrole}{\rolesets}{\loctyc}}
            {\zaptwo{s}{\qrole}}
        }
      }
      {
        \cseq
          {\Gamma}
          { a : \loctyb,
            \sessindexroles{s}{\prole}{\rolesetr}{\localtywait{\qrole}\then \loctya'}
            \sessindexroles{s}{\qrole}{\rolesets}{\loctyc}
          }
          { \eactor
              {a}
              {E[\raiseexn]}
              {\newconnsess{s}{\prole}{\rolesetr}}
              {\behaviour}
            \parallel
            \zap{\sessindextwo{s}{\qrole}}
          }
      }
    \end{mathpar}

    as required.
  \end{proofcase}

  \begin{proofcase}{E-FailS}
    \[
      \eactor{a}{P[\raiseexn]}{\newconnsess{s}{\prole}{\rolesetr}}{\behaviour}
        \ceval
      \eactor{a}{\raiseexn}{\bot}{\behaviour}
        \parallel \zap{\sessindextwo{s}{\prole}}
    \]

    Assumption:

    \begin{mathpar}
    \inferrule
    {
      a : \pidty{\loctyb} \in \Gamma \\\\
      \tseq{\Gamma}{\loctya}{P[\raiseexn]}{A}{\localend} \\
      \bseq{\Gamma}{\behaviour}
    }
    { \cseq{\Gamma}{a : \loctyb, \sessindexroles{s}{\prole}{\rolesetq}{\locty}}{
      \eactor{a}{P[\raiseexn]}{\newconnsess{s}{\prole}{\rolesetq}}{\behaviour}}
    }
    \end{mathpar}

    Since $\raiseexn$ is typable under any precondition and postcondition, and
    has an arbitrary return type,
    it follows that $\tseq{\Gamma}{\localend}{\raiseexn}{A}{\localend}$.

    Recomposing:

    \begin{mathpar}
      \inferrule*
      {
        \inferrule*
        {
          a : \pidty{\loctyb} \in \Gamma \\\\
          \tseq{\Gamma}{\localend}{\raiseexn}{A}{\localend} \\
          \bseq{\Gamma}{\behaviour}
        }
        { \cseq{\Gamma}
          {a : \loctyb}
          {\eactor{a}{\raiseexn}{\newconnsess{s}{\prole}{\rolesetq}}{\behaviour}}
        }
        \\
        \inferrule*
        { }
        {
          \cseq{\Gamma}
            {\sessindexroles{s}{\prole}{\rolesetq}{\locty}}
            {\zaptwo{s}{\prole}}
        }
      }
      {
        \cseq
          { \Gamma }
          { a : \loctyb, \sessindexroles{s}{\prole}{\rolesetq}{\locty} }
          { \eactor{a}{\raiseexn}{\bot}{\behaviour}
             \parallel \zap{\sessindextwo{s}{\prole}}
          }
      }
    \end{mathpar}
  \end{proofcase}

  \begin{proofcase}{E-FailLoop}

    Similar to \textsc{E-Loop}.
  \end{proofcase}

  \begin{proofcase}{E-LiftM}
    Immediate by Lemma~\ref{lem:term-pres}.
  \end{proofcase}

  \begin{proofcase}{E-Equiv}

    Immediate by Lemma~\ref{lem:equiv-pres}.
  \end{proofcase}

  \begin{proofcase}{E-Par}

    Immediate by the IH and rules \textsc{E-Cong1} and \textsc{E-Cong2}.
  \end{proofcase}

  \begin{proofcase}{E-Nu}

    Immediate by the IH, noting that the IH ensures that the resulting environment is also safe.
  \end{proofcase}
 \end{proof}

 \section{Progress}

\subsection{Overview}
The progress proof requires several steps. We overview them below.

\begin{description}
  \item[Canonical forms]
    Canonical forms allow us to reason globally about a configuration by putting it in a structured form. Lemma~\ref{lem:canonical-forms} states that any well-typed, closed configuration can be put into canonical form.
  \item[Exception-aware runtime typing environments]
    Runtime typing environments do not account explicitly for zapper threads.
    This makes sense when analysing protocols statically to check their safety
    and progress properties, but is inconvenient when reasoning about
    configurations.

    The second step is to introduce \emph{exception-aware} runtime typing
    environments which account explicitly for zapper threads and the propagation
    of exceptions, and an exception-aware runtime typing
    system~(\autoref{def:zap-envs}). We then show that configurations including
    zapper threads are typable under the exception-aware runtime typing
    system~(\autoref{lem:zap-typ}).

    We refine our notion of environment progress to include the possibility of
    failed sessions~(\autoref{def:zap-progress}), and show that given a runtime environment
    $\rtenv$ satisfying safety and progress, a derived exception-aware runtime
    environment $\zapenv$ also satisfies safety and
    progress~(\autoref{lem:zap-env-props}).
\item[Flattenings]
    Hu \& Yoshida's formulation of multiparty session types relaxes the
    directedness constraint for output choices, but this relaxation makes it
    more difficult to reason about reduction of environments being reflected by
    configurations.

    In this step, we introduce \emph{flattenings}
    (Definitions~\ref{def:flattening-types}
    and~\ref{def:flattening-envs}),
    which restict each top-level output choice to a single option, and show that
    if a configuration is \emph{ready} (i.e., all actors are blocked on
    communication actions), then it is typable under a flattened
    environment~(\autoref{lem:flattening-typ}). We then show that flattenings preserve
    environment reducibility~(\autoref{lem:flattening-reducibility}).
\item[Session progress]
    The penultimate step is to show that sessions can make progress.
    A key result is~\autoref{lem:flat-progress}, which states that a ready
    configuration typable under a flat exception-aware typing environment can
    reduce.
    The session progress lemma~(\autoref{lem:sess-group-progress}) follows
    from~\autoref{lem:flat-progress} and the previous results.
  \item[Progress]
    Finally, progress follows from~\autoref{lem:sess-group-progress} and case analysis on the
    disconnected actors.
\end{description}

\subsection{Auxiliary Definitions}

Let $\Psi$ range over typing environments containing only runtime names: $\Psi
::= \cdot \midspace \Psi, a : \pidty{S}$.

Term reduction satisfies a form of progress: a well-typed term is a value, can
reduce, or is a communication or concurrency construct:

\begin{lemma}[Progress (terms)]\label{lem:term-progress}
  If $\tseq[\loctyb]{\Psi}{\loctya}{M}{A}{\loctya'}$, then either $M =
  \effreturn{V}$ for some value $V$; there exists some $N$ such that $M \teval
  N$; or there exists some $E$
  such that $M$ can be written $E[N]$ for some $E, N$ where $N$ is either
  $\raiseexn$ or an adaptation, communication, or concurrency construct.
\end{lemma}

\begin{lemma}[Session typability]\label{lem:sgroup-typ}
  If $\cseq{\cdot}{\cdot}{\config{G}[\sgroup]}$, then there exist $\Psi, \rtenv$
  such that
  $\cseq{\Psi}{\rtenv}{\sgroup}$ and
  $\rtenv$ only consists of entries of the form $a : \loctya$.
\end{lemma}
\begin{proof}
  By induction on the structure of $\config{G}$, noting that entries of the
  form $\sessindexroles{s}{\prole}{\rolesetq}{\loctya}$ are linear;
  by the definition of $\sgroup = (\nu s)\config{C}$, actors and zapper threads
  in $\config{C}$ must refer to only to $s$.
\end{proof}

\subsection{Canonical forms}

To help us state a progress property, we consider configurations in
\emph{canonical form}.

A canonical form consists of binders for all connected actor names, followed by
binders for all disconnected actor names, followed by all sessions,
followed by all disconnected actors.

\begin{definition}[Canonical form]
  A configuration is in \emph{canonical form} if it is either $\confzero$ or can be written:
  $(\nu a_1 \cdots a_l)(\nu b_1 \cdots b_m)(\sgroup_1 \parallel \cdots \parallel \sgroup_n \parallel
    \eactor{b_1}{M_1}{\bot}{\behaviour_1} \parallel \cdots \parallel
    \eactor{b_m}{M_m}{\bot}{\behaviour_n})$.
\end{definition}

Every well-typed, closed configuration can be written in canonical form.

\begin{lemma}[Canonical forms]\label{lem:canonical-forms}
  If $\cseq{\cdot}{\cdot}{\config{C}}$, then $\exists \config{D}
  \equiv \config{C}$ where $\config{D}$ is in canonical form.
\end{lemma}
\begin{proof}[Proof sketch]
  Due to Lemma~\ref{lem:equiv-pres}, by induction on typing derivations we can
  move all actor name restrictions to the top level, followed by all session
  name restrictions, followed by all connected actors, followed by all zapper
  threads, followed by all disconnected actors.

Since the typing rules ensure each actor only participates in a single
  session, we can then group each actor and zapper thread according to its
  session, in order to arrive at a canonical form.
\end{proof}

\subsection{Exception-aware runtime typing environments}

Due to \textsc{E-RaiseP}, zapper threads expose additional reductions to those
allowed by standard environment reduction. In particular, given the presence of
a zapper thread $\zaptwo{s}{\prole}$, any other participant in the remainder of
the session blocked on $\prole$ can reduce. In order to account for this, we
introduce the notion of exception-aware runtime typing environments, environment
reduction, and runtime typing.

\begin{definition}[Exception-aware runtime typing environments, environment
  reduction, and runtime typing]\label{def:zap-envs}

  An \emph{exception-aware runtime typing environment} $\zapenv$ is defined as
  follows:
  \[
    \begin{array}{lrcl}
      \zapenv & ::= & a : S \midspace \sessindexroles{s}{\prole}{\rolesetq}{S} \midspace
      \zapindex{s}{\prole}
    \end{array}
  \]

  The \emph{exception-aware runtime typing relation}
  $\zseq{\Gamma}{\zapenv}{\config{C}}$ is defined to be the standard runtime
  typing relation $\cseq{\Gamma}{\rtenv}{\config{C}}$ but with rule
  \textsc{T-Zap} defined as:

  \begin{mathpar}
    \inferrule
    [T-Zap]
    { }
    { \zseq{\Gamma}{\zapindex{s}{\prole}}{\zaptwo{s}{\prole}} }
  \end{mathpar}

  We extend the set of labels $\synclbl$ with a zapper label
  $\lblzap{s}{\prole}{\qrole}$, which states that role $\qrole$ has
  attempted to interact with a cancelled role $\prole$ and has become cancelled
  itself as a result.

  The \emph{exception-aware environment reduction relation} is defined as the
  rules of the standard runtime typing relation $\rtenv \synceval \rtenv'$ with the
  addition of the following rules:

\begin{mathpar}
    \inferrule
    [ET-CommFail]
    { \exists j \in I . S_j = \localcomm{\prole}{\ell_j}{\tya_j}\then \locty'_j }
    {
      \zapindex{s}{\prole},
      \sessindexroles{s}{\qrole}{\rolesetr}{\sumact{i \in I}{S_i}}
      \syncannarrow{\lblzap{s}{\prole}{\qrole}}
      \zapindex{s}{\prole},
      \zapindex{s}{\qrole}
    }

    \inferrule
    [ET-DisconnFail]
    { }
    {
      \zapindex{s}{\prole},
      \sessindexroles{s}{\qrole}{\prole}{\localtydisconn{\qrole}}
      \syncannarrow{\lblzap{s}{\prole}{\qrole}}
      \zapindex{s}{\prole},
      \zapindex{s}{\qrole}
    }

    \inferrule
    [ET-WaitFail]
    { }
    {
      \zapindex{s}{\prole},
      \sessindexroles{s}{\qrole}{\prole}{\localtywait{\qrole} \then \loctya}
      \syncannarrow{\lblzap{s}{\prole}{\qrole}}
      \zapindex{s}{\prole},
      \zapindex{s}{\qrole}
    }
  \end{mathpar}
\end{definition}

\begin{definition}[Zapped roles]
  Define the \emph{zapped roles} of session
  \[
    \sgroup =
  (\nu s)(\eactor{a_1}{M_1}{\newconnsess{s}{\prole_1}{\rolesetq_1}}{\behaviour_1}
  \parallel \cdots \parallel
  \eactor{a_m}{M_m}{\newconnsess{s}{\prole_m}{\rolesetq_m}}{\behaviour_m}
  \parallel
  \zap{\sessindextwo{s}{\prole_{m + 1}}}
  \parallel \cdots
  \parallel
  \zap{\sessindextwo{s}{\prole_{n}}}
  )
  \]
  as $\prole_{m + 1}, \ldots, \prole_n$.
\end{definition}

\begin{definition}[Zapped environment]
  Given a runtime environment
  \[
    \Delta = a_1 : \loctya_1, \ldots, a_l : \loctya_l,
      \sessindexroles{s}{\prole_1}{\rolesetr_1}{\loctyb_1},
      \ldots
      \sessindexroles{s}{\prole_m}{\rolesetr_m}{\loctyb_m},
      \sessindexroles{s}{\qrole_1}{\rolesetr_1}{\loctyb'_1},
      \ldots
      \sessindexroles{s}{\qrole_n}{\rolesetr_n}{\loctyb'_n}
  \]
  and a set of roles $\rolesetq = \qrole_1, \ldots, \qrole_n$,
  the \emph{zapped environment} $\zapenv = \zapfn{\rtenv}{\rolesetq}$ is defined
  as:
\[
    \Delta = a_1 : \loctya_1, \ldots, a_l : \loctya_l,
      \sessindexroles{s}{\prole_1}{\rolesetr_1}{\loctyb_1},
      \ldots
      \sessindexroles{s}{\prole_m}{\rolesetr_m}{\loctyb_m},
      \zapindex{s}{\qrole_1},
      \ldots
      \zapindex{s}{\qrole_n}
  \]
\end{definition}

\begin{definition}[Failed environment]
  An exception-aware runtime typing environment $\zapenv$ is a \emph{failed
  environent}, written $\failed{\zapenv}$, if $\zapenv $ is of the form $\zapindex{s}{\prole_1}, \ldots,
  \zapindex{s}{\prole_n}$.
\end{definition}

\begin{lemma}[Zapped Typeability]\label{lem:zap-typ}
  If $\cseq{\Gamma}{\rtenv}{\sgroup}$ where $\sgroup = (\nu s : \rtenv')
  \config{C}$, the zapped roles of $\sgroup$ are $\rolesetp$, and $\zapenv =
  \zapfn{\rtenv'}{\rolesetp}$, then
  $\zseq{\Gamma}{\rtenv, \zapenv}{\config{C}}$.
\end{lemma}
\begin{proof}
  A straightforward induction on the derivation of
  $\cseq{\Gamma}{\rtenv}{\config{C}}$, noting that the definition of sessions
  and \textsc{T-Session} ensures that $\rtenv$ only contains actor
  runtime names, and that the  modified \textsc{T-Zap} rule is satisfied by the
  fact that $\zapenv = \zapfn{\rtenv}{\rolesetp}$.
\end{proof}

\begin{corollary}\label{cor:zap-typ}
  If $\cseq{\cdot}{\cdot}{\config{C}}$ where $\config{C}$ is in canonical form,
  then $\zseq{\cdot}{\cdot}{\config{C}}$.
\end{corollary}
\begin{proof}
  Follows directly by applying Lemma~\ref{lem:zap-typ} to each session.
\end{proof}

\begin{corollary}[Exception-aware session typability]\label{cor:sgroup-typ}
  If $\zseq{\cdot}{\cdot}{\config{G}[\sgroup]}$, then there exist $\Psi, \rtenv$
  such that $\rtenv$ only consists of entries of the form $a : \loctya$ and
  $\zseq{\cdot}{\cdot}{\sgroup}$.
\end{corollary}

Definition~\ref{def:env-progress} extends straightforwardly to exception-aware runtime typing
environments.

\begin{definition}[Progress (exception-aware environments)]\label{def:zap-progress}
  An exception-aware runtime typing environment $\zapenv$ \emph{satisfies progress}, written
  $\prog{\zapenv}$, if:
  \begin{itemize}
    \item \textit{(Role progress)} for each $\sessindexroles{s}{\prole_i}{\rolesetq_i}{\loctya_i} \in
      \zapenv$ such that $\tyactive{\loctya_i}$, it is the case that $\zapenv
      \ceval^* \zapenv' \annarrow{\ltslbl}$ with $\prole \in \roles{\ltslbl}$.
\item \textit{(Eventual communication)}
      if $\zapenv \syncevalstar \zapenv' \annarrow{\ltscomm[{!}]{s}{\prole}{\qrole}{\ell}{A}}$,
      then either
        $\zapenv' \syncannarrow{\seq{\synclbl}} \zapenv''
          \annarrow{\ltscomm[{?}]{s}{\qrole}{\prole}{\ell}{A}}$,
        or
        $\zapenv' \syncannarrow{\seq{\synclbl}} \zapenv''
        \syncannarrow{\lblzap{s}{\qrole}{\prole}}$,
        where $\prole \not\in \roles{\seq{\synclbl}}$
\item \textit{(Correct termination)}
      $\zapenv \syncevalstar \zapenv' \not\synceval$ implies either
      $\finished{\zapenv}$ or $\failed{\zapenv}$.
  \end{itemize}
\end{definition}

If a runtime typing environment is deadlock-free, then any corresponding zapped
evironment will be deadlock-free.

\begin{lemma}[Exception-aware environments preserve safety and progress]\label{lem:zap-env-props}
  Suppose $\safe{\rtenv}$ and $\prog{\rtenv}$ for some $\zapenv$.

  If $\zapenv = \zapfn{\rtenv}{\rolesetp}$ for some set of roles $\rolesetp$,
  then $\safe{\zapenv}$ and $\prog{\zapenv}$.
\end{lemma}
\begin{proof}[Proof sketch]\hfill
{
  \subparagraph*{Safety.}
  Safety follows since exception-aware environments do not modify session types, but
  only replace them with zappers. Interaction with a zapper is vacuously safe,
  so it follows that if $\safe{\rtenv}$ then $\safe{\zapenv}$.

  \subparagraph*{Progress.}
  For (role progress), we know that any active role in $\rtenv$ will eventually reduce.
  Presence of a zapped role means that any other role in the session trying to
  reduce with the zapped role will itself be zapped due to \textsc{ET-CommFail},
  \textsc{ET-DisconnFail}, or \textsc{ET-WaitFail}; connections will not occur,
  but this does not matter since the accepting role will not be in the session.

  For (eventual communication),
  we have that
  if $\rtenv \syncevalstar \rtenv' \annarrow{\ltscomm[{!}]{s}{\prole}{\qrole}{\ell}{A}}$,
      then $\rtenv' \syncannarrow{\seq{\synclbl}} \rtenv''
        \annarrow{\ltscomm[{?}]{s}{\qrole}{\prole}{\ell}{A}} \rtenv'''$,
        where $\prole \not\in \roles{\seq{\synclbl}}$.

  Suppose
  $\zapenv \syncevalstar \zapenv'
  \annarrow{\ltscomm[{!}]{s}{\prole}{\qrole}{\ell}{A}}$.
  If $\zapenv'$ contains zapped roles, either they are irrelevant to $\prole$'s reduction
  and we can use the original reduction sequence to show
  $ \zapenv \syncevalstar \zapenv''
  \annarrow{\ltscomm[{?}]{s}{\qrole}{\prole}{\ell}{A}} \rtenv'''$,
  or they are relevant to reduction and the exception propagates, resulting in
  $\zapenv' \syncannarrow{\seq{\synclbl}} \zapenv''
  \syncannarrow{\lblzap{s}{\qrole}{\prole}} \rtenv'''$.

  For (correct termination), we know that $\rtenv\not\synceval$ implies that
  $\finished{\rtenv}$. By (1), we know that each active role in $\zapenv$ will
  eventually reduce, so it follows that either $\finished{\zapenv}$ (if all
  communications occur successfully) or $\failed{\zapenv}$ if a role has failed
  and the failure has propagated to all other participants.
}
\end{proof}

\subsection{Flattenings}

\begin{definition}[Input / Output-Directed Choices]
  A choice session type $\sumact{i \in I}{\locact_i \then \loctya_i}$ is a
  \emph{output-directed} if each $\locact_i$ is either of the form
  $\localtysend{\prole_i}{\ell_i}{A_i}$.

  A choice session type is \emph{input-directed} if it is of the form
  $\sumact{i \in I}{\localcomm[\commdir]{\prole}{\ell_i}{A_i}}$ for some
  $\commdir \in \{ {?}, {??} \}$.

  For convenience, we write $\outloctya$ and $\outsumact{i \in I}{\locact_i
  \then \loctya_i}$ to denote an output-directed choice, and $\inloctya$ and
  $\insumact{i \in I}{\locact_i \then \loctya_i}$ to denote an input-directed
  choice.
\end{definition}

\begin{definition}[Output-flat session type]
  An output-directed session type $\outsumact{i \in I}{\loctya_i \then \loctya_i
  }$ is \emph{output-flat} if it is a unary choice, i.e., can be written
  $\localcomm[{!}]{\prole}{\ell}{A} \then \loctya$ or
  $\localcomm[{!!}]{\prole}{\ell}{A} \then \loctya$.

  An exception-aware runtime typing environment $\zapenv$ is \emph{output-flat},
  written $\oflat{\zapenv}$, if each output-directed choice type in $\zapenv$ is
  output-flat.
\end{definition}

We now define a \emph{flattening} of an output choice. A session type is a
flattening of an output sum if it is a unary sum consisting of one of the
choices.

\begin{definition}[Flattening (types)]\label{def:flattening-types}
  A session type $\loctya'$ is a \emph{flattening of} an output-directed choice
  $\loctya = \outsumact{i \in I}{\locact_i \then \loctyb_i}$ if $\loctya' =
  \locact_j \then \loctyb_j$ for some $j \in I$.
\end{definition}

We can extend flattening to environments by flattening all output-directed
choices.

\begin{definition}[Flattening (environments)]\label{def:flattening-envs}
  Given an exception-aware runtime typing environment:
  \[
    \zapenv =
    \sessindexroles{s}{\prole_1}{\rolesetr_1}{\outloctya_{1}}, \ldots,
    \sessindexroles{s}{\prole_l}{\rolesetr_l}{\outloctya_{l}},
    \sessindexroles{s}{\qrole_1}{\rolesets_1}{\loctyb_1},
    \ldots,
    \sessindexroles{s}{\qrole_m}{\rolesets_m}{\loctyb_m},
    \zapindex{s}{\trole_1}, \ldots, \zapindex{s}{\trole_n}
  \]
  where each $\loctyb$ is not an output-directed choice, we say that $\zapenv'$ is a
  \emph{flattening} of $\zapenv$ if:
\[
    \zapenv =
    \sessindexroles{s}{\prole_1}{\rolesetr_1}{\outloctya{'}_{1}}, \ldots,
    \sessindexroles{s}{\prole_l}{\rolesetr_l}{\outloctya{'}_{l}},
    \sessindexroles{s}{\qrole_1}{\rolesets_1}{\loctyb_1},
    \ldots,
    \sessindexroles{s}{\qrole_m}{\rolesets_m}{\loctyb_m},
    \zapindex{s}{\trole_1}, \ldots, \zapindex{s}{\trole_n}
  \]

  where each $\outloctya{'}_{i}$ is a flattening of $\outloctya_{i}$.
\end{definition}

Progress states that however an environment reduces, an active role will be
eventually be able to reduce. In the case that the environment does not reduce,
it is final.

\begin{definition}[Ready]
  A configuration $\config{C}$ is \emph{ready}, written $\ready{\config{C}}$,
  if all subconfigurations are either zapper threads or actors evaluating terms
  of the form $E[M]$ where $M$ is a communication action.
\end{definition}

\begin{lemma}[Flattening typability]\label{lem:flattening-typ}
  If $\zseq{\cdot}{\cdot}{\config{G}[\sgroup]}$ where $\sgroup = (\nu
  s: \zapenv)\config{C}$ and $\ready{\config{C}}$,
  then there exists some $\zapenv'$ such that $\oflat{\zapenv'}$ where
  $\zapenv'$ is a flattening of $\zapenv$, and
  $\zseq{\cdot}{\cdot}{\config{G}[(\nu s : \zapenv') \config{C}]}$.
\end{lemma}
\begin{proof}
  By Lemma~\ref{cor:sgroup-typ}, we have that there exist $\Psi, \rtenv$ such
  that $\rtenv$ does not contain session entries, and
  $\zseq{\Psi}{\rtenv}{\sgroup}$.

  By \textsc{T-Session}, $\zseq{\Psi}{\rtenv, \zapenv}{\config{C}}$.

  The result follows by induction on the derivation
  of $\zseq{\Psi}{\rtenv, \zapenv}{\config{C}}$. In the case of \textsc{T-Zap},
  the result follows immediately.

  In the case of \textsc{E-Actor}, we have that:
  \begin{enumerate}
    \item $\config{C} =\eactor{a}{M}{\newconnsess{s}{\prole}{\rolesetr}}{\behaviour}$
    \item $\zseq{\Psi}{a : \loctyb, \sessindexroles{s}{\prole}{\rolesetr}{\loctya}}{\config{C}}$
    \item $\tseq{\Psi}{\loctya}{M}{A}{\localend}$
  \end{enumerate}

  Since $\ready{\config{C}}$, it must be the case that the actor is evaluating a
  term of the form $E[N]$ where $N$ is a communication action.

  By Lemma~\ref{lem:subterm-typability}, $\tseq{\Psi}{\loctya}{N}{B}{\loctya'}$ for
  some $B, \loctya'$.

  In the case that $N$ is a $\calcwd{disconnect}$, $\calcwd{wait}$,
  $\calcwd{receive}$, or $\calcwd{accept}$ action, then the session type cannot
  be output-directed and the result follows immediately.

  In the case that $M = \newconn{\ell_j}{V_j}{W}{\qrole}$, by \textsc{T-Connect} it must be the case
  that $\loctya = \outsumact{i \in I}{\locact_i \then \loctya'_i}$ where
  $\localtyconn{\qrole}{\ell_j}{A_j} \in \{ \locact_i \}_{i \in I}$ and
  $\vseq{\Psi}{V_j}{A_j}$. Again by \textsc{T-Connect} and
  Lemma~\ref{lem:replacement}, we can show that
  $\tseq{\Psi}{\localtyconn{\qrole}{\ell_j}{A_j} \then
  \loctya'_i}{E[\newconn{\ell_j}{V_j}{W}{\qrole}]}{B}{\loctya'}$, where
  $\localtyconn{\qrole}{\ell_j}{A_j} \then \loctya'_i$ is a flattening of
  $\loctya$, as required.

  The same argument holds for send actions.
\end{proof}

\begin{lemma}[Flattening preserves
  reducibility]\label{lem:flattening-reducibility}
  Suppose $\safe{\zapenv}$, $\prog{\zapenv}$, and $\zapenv \synceval$.

  If $\hat{\zapenv}$ is a flattening of $\zapenv$, then $\hat{\zapenv} \synceval$.
\end{lemma}
\begin{proof}
  Assume for the sake of contradiction that $\hat{\zapenv} \not\synceval$.

  For this to be the case, each $\synclbl$ such that $\zapenv
  \syncannarrow{\synclbl}$ must either be $\ltspair{s}{\prole}{\qrole}{\ell}$ or
  $\ltsconnpair{s}{\prole}{\qrole}{\ell}$,
  but instead, $\hat{\zapenv}
  \annarrow{\ltscomm[{!}]{s}{\prole}{\qrole'}{\ell'}{\tya}}$ for some $\qrole'
  \ne \qrole$.
  Note that $\synclbl$ cannot be a disconnection action, since these are
  unaffected by flattening, and it cannot be the case that
  $\hat{\zapenv}
  \annarrow{\ltsconnpair{s}{\prole}{\qrole'}{\ell'}{\tya}}$, since this is a
  synchronisation action and could reduce.

  By \textit{(eventual communication)}, either
  $\zapenv \syncannarrow{\seq{\synclbl'}} \zapenv'
  \annarrow{\ltscomm[{?}]{s}{\prole}{\qrole'}{\ell'}{A}}$ or
$\zapenv \syncannarrow{\seq{\synclbl'}} \zapenv'
  \syncannarrow{\lblzap{\qrole'}{\prole}}$.

  Pick some $\synclbl$ such that $\seq{\synclbl'}$ contains no roles affected by
  flattening; one such $\synclbl$ must exist since otherwise there would be a
  cycle, contradicting progress.

  It follows that either:

  \begin{itemize}
    \item $\hat{\zapenv} \syncannarrow{\seq{\synclbl}} \zapenv'
      \annarrow{\ltscomm[{!}]{s}{\prole}{\qrole'}{\ell'}{\tya}}
      \annarrow{\ltscomm[{?}]{s}{\qrole'}{\prole}{\ell'}{\tya}}$
      meaning
      $\hat{\zapenv} \syncannarrow{\seq{\synclbl}} \zapenv'
      \syncannarrow{\ltspair{s}{\prole}{\qrole'}{\ell'}}$, a contradiction; or
    \item
      $\zapenv \syncannarrow{\seq{\synclbl'}} \zapenv'
      \syncannarrow{\lblzap{\qrole'}{\prole}}$, also a contradiction.
  \end{itemize}

\end{proof}

\subsection{Session progress}

\begin{lemma}[Readiness]\label{lem:readiness}
  Suppose $\zseq{\cdot}{\cdot}{\config{C}}$ where
  $\config{C} = \config{G}[\sgroup]$,
  $\config{C}$ does not contain unmatched discovers,
  $\sgroup$ is not a failed session,
  $\sgroup = (\nu s: \zapenv)\config{D}$, and $\prog{\zapenv}$.

  Either $\config{C} \ceval$ or $\ready{\config{D}}$.
\end{lemma}
\begin{proof}
  For each $\eactor{a}{M}{\newconnsess{s}{\prole}{\rolesetq}}{\behaviour}$,
  by Lemma~\ref{lem:term-progress}, $M$ can either reduce or is either a value,
  adaptation construct, or communication construct.

  \begin{itemize}
    \item If $M$ can reduce, then the configuration can reduce by
      \textsc{E-LiftM}.
    \item If $M = E[\raiseexn]$, then the configuration either can reduce by
      \textsc{E-TryRaise} and \textsc{E-LiftM}, or \textsc{E-FailS}.
    \item If $M$ is an adaptation action, due to the absence of unmatched
      discovers, the configuration can reduce by \textsc{E-New},
      \textsc{E-Replace}, \textsc{E-ReplaceSelf}, \textsc{E-Discover}, or
      \textsc{E-Self}.
    \item If $M$ is a value, then by \textsc{T-ConnectedActor},
      $\tseq{\Psi}{\localend}{\effreturn{V}}{A}{\localend}$. As a consequence of
      $\prog{\zapenv}$, $a$ must be the only actor in the session and thus the
      configuration could reduce by \textsc{E-Complete}.
    \item It follows that all terms must be evaluating communication actions,
      satisfying $\ready{\config{D}}$.
  \end{itemize}
\end{proof}

\begin{lemma}[Exception-aware session progress]\label{lem:flat-progress}
  If $\zseq{\cdot}{\cdot}{\config{C}}$ where:
  \begin{enumerate}
    \item $\config{C} \equiv \config{G}[\sgroup]$,
    \item $\sgroup = (\nu s : \zapenv) \config{D}$,
    \item $\oflat{\zapenv}$,
    \item $\ready{\config{C}}$,
    \item $\zapenv \synceval$
  \end{enumerate}
  then $\config{C} \ceval$.
\end{lemma}
\begin{proof}
  Since $\zapenv \synceval$, there exists some environment $\zapenv'$ and
  label $\synclbl$ such that $\zapenv \syncannarrow{\synclbl} \zapenv'$.

  The proof is by induction on the derivation of $\zapenv \syncannarrow{\synclbl} \zapenv'$.

  \begin{proofcase}{ET-Conn}
    \begin{mathpar}
      \inferrule
      { \exists j \in I . S_j = \localtyconn{\qrole}{\ell_j}{\tya_j}\then \locty'_j \\\\
        \ty{\qrole} = \sumact{k \in K}{\localtyaccept{\prole}{\ell_k}{\tyb_k} \then \loctyb_k} \\
        j \in K \\
        \tya_j = \tyb_j
      }
      {
        \sessindexroles{s}{\prole}{\rolesetr}{\sumact{i \in I}{S_i}}
        \syncannarrow{\ltsconnpair{s}{\prole}{\qrole}{\ell_j}}
        \sessindexroles{s}{\prole}{\rolesetr, \qrole}{\locty'_j},
        \sessindexroles{s}{\qrole}{\prole}{\loctyb_j}
      }
    \end{mathpar}

    Since $\oflat{\zapenv}$, we have that
    $\sessindexroles{s}{\prole}{\rolesetr}{\localtyconn{\qrole}{\ell_j}{\tya_j}\then
    \locty'_j}$. Thus, $\config{C}$ must
    contain an actor evaluating $E[\newconn{\ell_j}{\tya_j}{V}{b}]$ for some
    actor name $b$. Due to the definition of canonical forms, actor $b$ must be
    a subconfiguration of $\config{G}[\sgroup]$ and thus the configuration can
    reduce by either \textsc{E-Conn} or \textsc{E-ConnFail}.
  \end{proofcase}
  \begin{proofcase}{ET-Comm}
    \begin{mathpar}
      \inferrule
        [ET-Comm]
        { \zapenv_1 \annarrow{\ltscomm[{!}]{s}{\prole}{\qrole}{\ell}{\tya}} \zapenv'_1
          \\
        \zapenv_2 \annarrow{\ltscomm[{?}]{s}{\qrole}{\prole}{\ell}{\tya}} \zapenv'_2 }
        { \zapenv_1, \zapenv_2 \syncannarrow{\ltspair{s}{\prole}{\qrole}{\ell}} \zapenv'_1, \zapenv'_2 }
    \end{mathpar}

    By \textsc{ET-Act}, $\zapenv_1$ must contain
    $\sessindexroles{s}{\prole}{\rolesetr}{\loctya}$; since $\oflat{\zapenv}$,
    $\loctya = \localcomm[{!}]{\qrole}{\ell}{A}$.

    Also by \textsc{ET-Act}, $\zapenv_2$ must contain
    $\sessindexroles{s}{\prole}{\rolesetr}{\sumact{i \in I}{\locactb_i \then
    \loctyb'_i}}$ with $\localcomm[{?}]{\prole}{\ell}{A} \in \{ \locactb_i \}$.

    Thus by typing, $\config{C}$ must contain an actor
    $\eactor{a}{E[\newsend{\ell}{V}{\qrole}]}{\newconnsess{s}{\prole}{\rolesetr}}{\behaviour_a}$,
    and another actor
    $\eactor{b}{E'[\newrecv{\prole}{\ell_i(x_i) \mapsto
    M_i}]}{\newconnsess{s}{\qrole}{\rolesets}}{\behaviour_b}$, which would
    reduce by \textsc{E-Comm}.
  \end{proofcase}
  \begin{proofcase}{ET-CommFail}
    \begin{mathpar}
    \inferrule
    [ET-CommFail]
    { \exists j \in I . S_j = \localcomm{\prole}{\ell_j}{\tya_j}\then \locty'_j }
    {
      \zapindex{s}{\prole},
      \sessindexroles{s}{\qrole}{\rolesetr}{\sumact{i \in I}{\locact_i \then
      \loctya'_i}}
      \syncannarrow{\lblzap{s}{\prole}{\qrole}}
      \zapindex{s}{\prole},
      \zapindex{s}{\qrole}
    }
    \end{mathpar}

    We prove the case where $\commdir = {!}$. Since $\oflat{\zapenv}$, by typing we have
    that $\sumact{i \in I}{\locact_i \then \loctya'_i} =
    \localtysend{\ell}{A}{\qrole} \then \loctya'$ and $\config{C}$ must contain an actor
    $\eactor{a}{E[\newsend{\ell}{V}{\qrole}]}{\newconnsess{s}{\prole}{\rolesetr}}{\behaviour_a}$
    and a zapper thread $\zaptwo{s}{\qrole}$, which could then reduce by
    \textsc{E-RaiseP}.
  \end{proofcase}

  Case \textsc{E-Disconn} is similar to \textsc{E-Conn}, and cases
  \textsc{E-DisconnFail} and \textsc{E-WaitFail} are similar to \textsc{E-CommFail}.

  Cases \textsc{ET-Rec}, \textsc{ET-Cong1}, and \textsc{ET-Cong2} follow by the
  induction hypothesis.
\end{proof}

\begin{lemma}\label{lem:zapenv-reduction}
  If $\prog{\zapenv}$ and neither $\finished{\zapenv}$ nor $\failed{\zapenv}$,
  then $\zapenv \synceval$.
\end{lemma}
\begin{proof}
  By contradiction. Assume that $\prog{\zapenv}$ and neither
  $\finished{\zapenv}$ nor $\failed{\zapenv}$.

  By $\prog{\zapenv}$, each active role must eventually reduce.

  If no roles are active, then $\finished{\zapenv}$ or $\failed{\zapenv}$: a
  contradiction.

  Otherwise, by the definition of $\prog{\zapenv}$
    for each $\sessindexroles{s}{\prole}{\rolesetq}{\loctya}$
    there must exist some $\seq{\synclbl}$ such that $\zapenv
    \syncannarrow{\seq{\synclbl'}}{}^* \zapenv' \syncannarrow{\synclbl}$ where $\prole \in
    \subj{\synclbl}$: a contradiction.
\end{proof}

\begin{lemma}\label{lem:zapenv-ready-not-final}
  If $\zseq{\cdot}{\cdot}{\config{\sgroup}}$ where $\sgroup = (\nu s:
  \zapenv)\config{C}$ and $\ready{\config{C}}$, then $\zapenv$ is not final.
\end{lemma}
\begin{proof}
  By contradiction.  If $\zapenv$ were final, it would consist only of session
  types of type $\localend$. No communication actions are typable under
  $\localend$, contradicting $\ready{\config{C}}$.
\end{proof}

\sessgroupprogress*
\begin{proof}
  \begin{itemize}
    \item By Lemma~\ref{lem:sgroup-typ}, $\exists. \Psi, \rtenv'$ such that
      $\cseq{\Psi}{\rtenv'}{\sgroup}$.
    \item By definition, $\sgroup =
      (\nu s:
      \rtenv)\config{D}$,
      where $\config{D} = \eactor{a_1}{M_1}{\newconnsess{s}{\prole_1}{\rolesetq_1}}{\behaviour_1}
  \parallel \cdots \parallel
  \eactor{a_m}{M_m}{\newconnsess{s}{\prole_m}{\rolesetq_m}}{\behaviour_m}
  \parallel
  \zap{\sessindextwo{s}{\prole_{m + 1}}}
  \parallel \cdots
  \parallel
  \zap{\sessindextwo{s}{\prole_{n}}}$.
    \item By \textsc{T-Session}, $\safe{\rtenv}$.
    \item By definition, the zapped roles of $\config{C}$ are $\prole_{m + 1},
      \ldots, \prole_n$. Let us denote this set as $\rolesetp$.
    \item Let $\zapenv = \zapfn{\rtenv}{\rolesetp}$.
    \item By Lemma~\ref{lem:zap-typ}, $\zseq{\Psi}{\rtenv', \zapenv}{\config{D}}$
      and so by Corollary~\ref{cor:zap-typ},
      $\zseq{\cdot}{\cdot}{\config{G}[\sgroup]}$.
    \item By Lemma~\ref{lem:zap-env-props}, $\safe{\zapenv}$ and
      $\prog{\zapenv}$.
    \item By Lemma~\ref{lem:readiness}, either $\config{C} \ceval$ or
      $\ready{\config{D}}$. As $\config{C} \ceval$ satisfies the theorem
      statement, we proceed assuming $\ready{\config{D}}$.
    \item By Lemma~\ref{lem:flattening-typ}, there exists some $\hat{\zapenv}$ such
      that $\hat{\zapenv}$ is a flattening of $\zapenv$ and
      $\zseq{\cdot}{\cdot}{\config{G}[(\nu s: \hat{\zapenv})\config{D}]}$.
    \item By Lemma~\ref{lem:zapenv-ready-not-final}, $\hat{\zapenv}$ is not final.
    \item By Lemma~\ref{lem:flattening-reducibility}, $\hat{\zapenv} \synceval$.
    \item Thus, by Lemma~\ref{lem:flat-progress}, $\config{C} \ceval$ as
      required.
  \end{itemize}
\end{proof}

\subsection{Progress}

\progress*
\begin{proof}
  \begin{itemize}
    \item By Lemma~\ref{lem:canonical-forms}, write $\config{C}$ in
        canonical form.
    \item Eliminate all failed sessions via the equivalence $(\nu
      s)(\zap{s}{\prole_1} \parallel \cdots \parallel \zaptwo{s}{\prole_n}) \parallel \config{C}
      \equiv \config{C}$. Heneceforth assume no sessions are failed.
    \item By repeated application of Lemma~\ref{lem:sess-group-progress}, either
      $\config{C}$ can reduce or it contains no sessions.
    \item Thus we only need to consider disconnected actors. For each disconnected
      actor $\eactor{a}{M}{\bot}{\behaviour}$, by Lemma~\ref{lem:term-progress},
      either:
      \begin{enumerate}
        \item $M = \effreturn{V}$. If $\behaviour = N$, then the configuration
          can reduce by \textsc{E-Loop}. Otherwise, if $\behaviour = \bstop$,
          then the actor is terminated, as required.
\item $M = E[\raiseexn]$.
          If $E$ contains an exception handler, the configuration can reduce by
          \textsc{E-TryRaise}. If not, $E$ is pure: if $\behaviour = N$, then
          the configuration can reduce by \textsc{E-FailLoop}. Otherwise, if
          $\behaviour = \bstop$, then the actor is terminated, as required.
\item $M = E[N]$ where $N$ is a communication or concurrency action,
          which, following the logic in Lemma~\ref{lem:sess-group-progress}, can
          either reduce by a configuration reduction rule or is a communication
          action.
Since:
          \begin{itemize}
            \item by \textsc{T-UnconnectedActor} each disconnected actor must
          either have type $\localend$ or be following its statically-defined
          session type
            \item by well-formedness, each statically-defined session type must
              correspond to a role in a protocol
            \item by well-formedness, each protocol must have a unique initiator
            \item since the program satisfies progress, all protocols satisfy progress
          \end{itemize}

          it must be the case that either $a$ is an initiator of a session and
          blocked on $\calcwd{connect}$, which could either reduce by
          \textsc{E-ConnInit} or \textsc{E-ConnFail}, or be accepting, as
          required.
      \end{enumerate}
  \end{itemize}
\end{proof}
 \end{adjustwidth}

\end{document}